\documentclass[11pt, a4paper]{article}

\usepackage{amsmath, amsfonts, amssymb, amsthm, mathtools}
\usepackage{graphicx, subcaption, rotating}
\usepackage{setspace}
\usepackage[left=2.5cm, right=2.5cm, top=3cm, bottom=3cm, a4paper]{geometry}
\usepackage{natbib} 
\usepackage{url} 
\usepackage[inline]{enumitem}
\usepackage{multirow}
\usepackage{xcolor}
\usepackage{appendix}
\usepackage{authblk}
\usepackage{multirow}
\usepackage{booktabs}
\usepackage{ulem}
\usepackage{hyperref}

\usepackage{subfiles}
\usepackage{xr}

\theoremstyle{definition}
\newtheorem{thm}{Theorem}[section]
\theoremstyle{definition}

\theoremstyle{definition}
\newtheorem{lem}[thm]{Lemma}
\theoremstyle{definition}
\newtheorem{prop}[thm]{Proposition}
\theoremstyle{definition}

\theoremstyle{remark}
\newtheorem{rem}{Remark}[section]
\theoremstyle{definition}

\theoremstyle{definition}

\newcommand{\minitab}[2][l]{\begin{tabular}{#1}#2\end{tabular}}

\newcommand{\cG}{\mathcal{G}}

\newcommand{\cI}{\mathcal{I}}

\newcommand{\cK}{\mathcal{K}}

\newcommand{\cM}{\mathcal{M}}
\newcommand{\cN}{\mathcal{N}}

\newcommand{\bbR}{\mathbb{R}}

\newcommand{\bbZ}{\mathbb{Z}}

\newcommand{\bA}{\mathbf{A}}
\newcommand{\bB}{\mathbf{B}}
\newcommand{\bC}{\mathbf{C}}
\newcommand{\bD}{\mathbf{D}}

\newcommand{\bI}{\mathbf{I}}

\newcommand{\bO}{\mathbf{O}}

\newcommand{\bV}{\mathbf{V}}

\newcommand{\bbeta}{\boldsymbol{\beta}}

\newcommand{\bzeta}{\boldsymbol{\zeta}}
\newcommand{\bDelta}{\boldsymbol{\Delta}}

\newcommand{\bv}{\mathbf{v}}

\newcommand{\bx}{\mathbf{x}}

\newcommand*\diff{\mathop{}\!\mathrm{d}}

\renewcommand{\b}{\boldsymbol}

\newcommand{\E}[0]{\mathsf{E}}
\newcommand{\Var}[0]{\mathsf{Var}}
\newcommand{\Cov}[0]{\mathsf{Cov}}
\renewcommand{\l}{\left}
\renewcommand{\r}{\right}
\def\wh{\widehat}
\def\wt{\widetilde}
\newcommand{\mbf}{\mathbf}
\newcommand{\mc}{\mathcal}
\newcommand{\bmx}{\begin{bmatrix}}
\newcommand{\emx}{\end{bmatrix}}
\newcommand{\nn}{\nonumber}

\allowdisplaybreaks

\begin{document}




\title{\bf Moving sum procedure for change point detection under piecewise linearity}
\author[1]{Joonpyo Kim}
\author[2]{Hee-Seok Oh}
\author[3]{Haeran Cho}
\affil[1]{Sejong University, Republic of Korea}
\affil[2]{Seoul National University, Republic of Korea} 
\affil[3]{University of Bristol, United Kingdom}
\date{} 

\maketitle
\begin{abstract}
We propose a computationally and statistically efficient procedure for segmenting univariate data under piecewise linearity. The proposed moving sum (MOSUM) methodology detects multiple change points where the underlying signal undergoes discontinuous jumps and/or slope changes. It controls the family-wise error rate at a given significance level and achieves consistency in multiple change point detection, with a minimax optimal estimation rate when the signal is piecewise linear and continuous, all under weak assumptions permitting serial dependence and heavy-tailedness. Computationally, the complexity of the MOSUM procedure is $O(n)$, which, combined with its good performance on simulated datasets, makes it highly attractive compared to the existing methods. We further demonstrate its good performance on a real data example on rolling element-bearing prognostics.
\end{abstract}

\noindent%

\setstretch{1.5} 
\section{Introduction}
\label{sec:one}

Data segmentation, a.k.a.\ multiple change point detection, is an active field of research in time series analysis and signal processing, and numerous applications are found, e.g.,\ in climatology \citep{reeves2007review}, genomics \citep{niu2012screening}, neuroscience \citep{aston2012}, neurophysiology \citep{messer2014} and finance \citep{bardwell2019most}. We refer to \cite{truong2020} and \cite{cho2021} for an overview of the recent developments. In particular, the review articles demonstrate that the problem of detecting multiple change points in the mean of univariate time series has received significant attention, and several state-of-the-art methods exist for this canonical change point problem.

By contrast, there are far fewer papers that address the problem of detecting change points under piecewise linearity, where the signal underlying the data undergoes discontinuous jumps or slope changes. In practice, it is rarely known in advance whether the data is best approximated by a piecewise constant signal or a piecewise linear one, and many time series datasets exhibit complex features that may not be represented as piecewise constant functions. We demonstrate this using a dataset first analyzed by \cite{bearing2}. Accelerometers were installed on four test bearings and recorded their vibrations from February 12th to February 19th in 2004, when a failure occurred in one of the bearings. Focusing on the data obtained from the faulty bearing, Figure~\ref{bearing1.plot} shows that a drastic change in the trend is observed after February 16th, followed by severe instabilities from February 17th onward. This demonstrates the necessity for a methodology that detects both abrupt jumps and continuous changes in the trend while being agnostic to the type of changes in order to infer the onset of the mechanical fault.

\begin{figure}[!htb]
\centering
\includegraphics[width=0.6\textwidth]{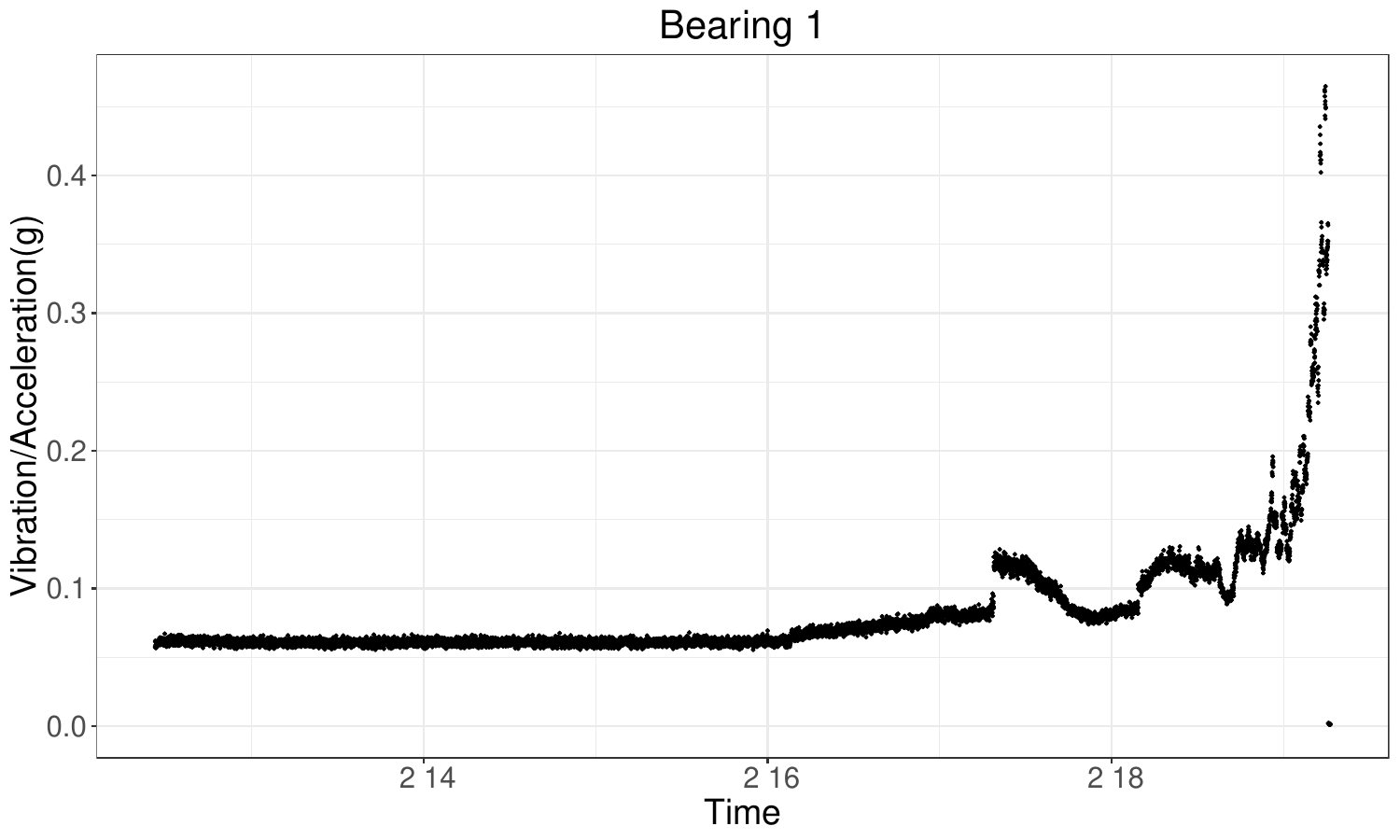}
\vspace{-2mm}
\caption{Time series data ($n = 9830$) generated by taking a one-minute average of the raw observations recorded by an accelerometer installed on a faulty test bearing between 10:32:39 on February 12th and 06:22:39 on February 19th in 2004. We take the average to remove any periodic vibration and focus on detecting changes in the trend of the data.}
\label{bearing1.plot}
\end{figure}

There are several methods for detecting the changes under piecewise linearity, which, as in the case of the canonical mean change point problem \citep{cho2021}, are categorized into those based on the application of localized testing, and those based on global optimization of an objective function. In the first category, \cite{baranowski2019} proposed the narrowest-over-threshold (NOT) methodology, which, as a variant of the wild binary segmentation \citep{fryzlewicz2014}, identifies local sections of the data that contain features (i.e.,\ slope or intercept changes) using contrast functions tailored for detecting the particular changes of interest. \cite{maeng2019detecting} found a sparse representation of the data via wavelets constructed for well-capturing piecewise linear signals, extending the approach of \cite{fryzlewicz2018tail}. \cite{anastasiou2022detecting} suggested localizing change points by iteratively expanding the local intervals under inspection. Methods based on minimizing $\ell_0$-penalized cost functions belong to the second category, which includes the CPOP methodology \citep{fearnhead2019}, where the dynamic programming popularly adopted for the mean change point detection problems \citep{jackson2005algorithm, killick2012optimal}, is also employed to find the best piecewise linear and continuous fit to the data. \cite{yu2022localising} provide theoretical investigation into such an estimator for the problem of localizing change points in piecewise polynomials of general degrees. In addition, we mention the literature on piecewise polynomial regression or spline smoothing with the knots at fixed \citep{green1993nonparametric} or unfixed \citep{mammen1997locally, tibshirani2014, guntuboyina2020adaptive, spiriti2013knot} locations where typically, the aim is to control the $\ell_2$-risk of the estimated signal. The problem of real-time monitoring of changes in streaming settings has also received attention \citep{wu2015online, wen2018multiple, xu2023online}, but our primary focus lies in {\it offline} data segmentation.

We propose a moving window-based methodology for the data segmentation problem under piecewise linearity. Referred to as the moving sum (MOSUM) procedure, it scans for multiple jumps and slope changes using the detector statistic which compares the local estimators of intercept and slope parameters from adjacent moving windows. MOSUM procedures have popularly been adopted in the data segmentation literature for their computational efficiency, from detecting change points in the mean of univariate time series \citep{eichinger2018, cho2019twostage} and regime shifts in multivariate renewal processes \citep{kirch2021moving}, to segmenting multivariate \citep{yau2016inference} and high-dimensional \citep{cho2022high} time series under parametric models. \cite{kirch2022data} provided a general change point detection methodology based on estimating equations.
Distinguished from these efforts, we permit the presence of time-varying trends in the data, which requires careful treatment both theoretically and methodologically.

The proposed MOSUM procedure is shown to (i)~control the (asymptotic) family-wise error rate at a given significance level, (ii)~achieve consistency in estimating both the total number and the locations of the change points, and further, (iii)~exactly match the minimax optimal rate of estimation when the underlying signal is piecewise linear and continuous. Our theoretical results are derived under mild conditions permitting serial dependence and heavy-tailedness, which are considerably more general than the independence and (sub)-Gaussianity assumptions in the existing literature. Computationally, thanks to the use of moving windows, the MOSUM procedure is highly efficient with the $O(n)$ complexity, making it particularly attractive in analyzing large datasets. The R code implementing our method is available at 
\url{https://github.com/Joonpyo-Kim/MovingSumLin}.

 \section{MOSUM procedure under piecewise linearity}
\label{sec:two}

\subsection{Methodology} 

We consider the following model
\begin{align} 
\label{eq:model}
X_i = f_i + \epsilon_i = \sum_{j = 1}^{J_n + 1}( \alpha_{0, j}  + \alpha_{1, j} t_{i}) \cdot \mathbb{I}_{\{k_{j-1} + 1 \le i \le k_{j } \}} + \epsilon_i, \quad  i = 1, \ldots, n, 
\end{align}
where $t_i = i \Delta t$ denotes the time points with $[0, T]$ as the observation period and $\Delta t = T/n$. We assume that $\{\epsilon_i\}_{i = 1}^n$ is a stationary sequence of random variables with $\E(\epsilon_i) = 0$, $\Var(\epsilon_i) = \sigma^2$ and the long-run variance (LRV) $\tau^2 = \sigma^2 + 2 \sum_{h = 1}^\infty \Cov(\epsilon_0, \epsilon_h)$, where $\sigma^2, \tau^2 \in (0, \infty)$, and it is allowed to be serially dependent as specified later. Under the model~\eqref{eq:model}, $f_i = \E(X_i)$ is piecewise linear with $J_n$ change points denoted by $k_j, \, j = 1, \ldots, J_n$ (with $k_0 = 0$ and $k_{J_n + 1} = n$), at which either the intercept or the slope or both, undergo changes. 
That is, denoting by $\b\alpha_j = (\alpha_{0, j}, \alpha_{1, j})^\top$, we have $\b\alpha_j \ne \b\alpha_{j + 1}$ for all $j = 1, \ldots, J_n$.
We permit $J_n \to \infty$ as $n \to \infty$ provided that change points are sufficiently distanced away from one another as specified later.
When $\alpha_{1, j} = 0$ for all $j$, the model~\eqref{eq:model} becomes the canonical change point model with piecewise constant $f_i$. 

Under the model in~\eqref{eq:model}, our aim is two-fold,
(i)~to test the null hypothesis of no change point $\mc H_0: \, J_n = 0$ against $\mc H_1: \, J_n > 0$, and
(ii)~if $\mc H_0$ is rejected, to estimate the total number $J_n$ and the locations $k_j$ of the change points. 
To achieve the above goals, we propose a moving window-based methodology that scans for (possibly) multiple change points by comparing the local parameter estimates from the adjacent windows.
Specifically, let $G$ denote a bandwidth satisfying $2G < n$, and define
$\cI^+(k) = \{k + 1, \ldots, k + G\}$ and $\cI^{-}(k) = \{k - G + 1, \ldots, k\}$ for $G \le k \le n - G$.
Then, at each time point~$k$, we regress $X_i$ onto $\mbf x_{i, k} = (1, (i - k)/G)^\top$ for $i \in \cI^+(k)$ (resp.\ $\cI^-(k)$) to obtain the least squares estimator $\wh{\bbeta}^{+}(k) = (\wh{\beta}_{0}^{+}(k), \wh{\beta}_{1}^{+}(k))^{\top}$ (resp.\ $\wh{\bbeta}^{-}(k)$).
The choice of the regressor $\mbf x_{i, k}$ allows the intercept and the slope estimators to be treated on an equal footing.
Then, if neither discontinuous jump nor slope change occurs on $\cI^{-}(k) \cup \cI^{+}(k)$, we expect $\Vert \wh{\bbeta}^{+}(k) - \wh{\bbeta}^{-}(k) \Vert$ to be small and vice versa, where $\Vert \cdot \Vert$ denotes the Euclidean norm. 

Based on these observations, we propose the following Wald-type MOSUM statistic
\begin{align} 
\label{waldstat}
W_{k, n}(G) = \frac{\sqrt{G}}{\wh\tau_k} \l\Vert \b\Sigma^{-1/2} \l( \wh{\bbeta}^{+}(k) - \wh{\bbeta}^{-}(k) \r) \r\Vert, \, G \le k \le n - G,
\end{align}
where $\b\Sigma$ is a $2 \times 2$ diagonal matrix with diagonal elements 8 and 24 (motivated by the distribution of $\wh{\bbeta}^{+}(k) - \wh{\bbeta}^{-}(k)$ under $\mc H_0$), and {$\wh\tau_k >0$} denotes a (possibly) location-dependent estimator of $\tau$.
Then, for some $\alpha \in (0, 1)$, we reject $\mc H_0$ if
$W_n(G) := \max_{G \le k \le n - G} W_{k, n}(G)$ exceeds a critical value $C_n(G, \alpha)$
obtained from the asymptotic null distribution of $W_n(G)$ (see Theorem~\ref{thm3.1} below),
the choice of which ensures that the test controls the family-wise error rate at the prescribed level $\alpha$ when $W_{k, n}(G)$ is scanned over $k = G, \ldots, n - G$.

\begin{figure}[!h]
\centering
\includegraphics[width = 0.9\textwidth]{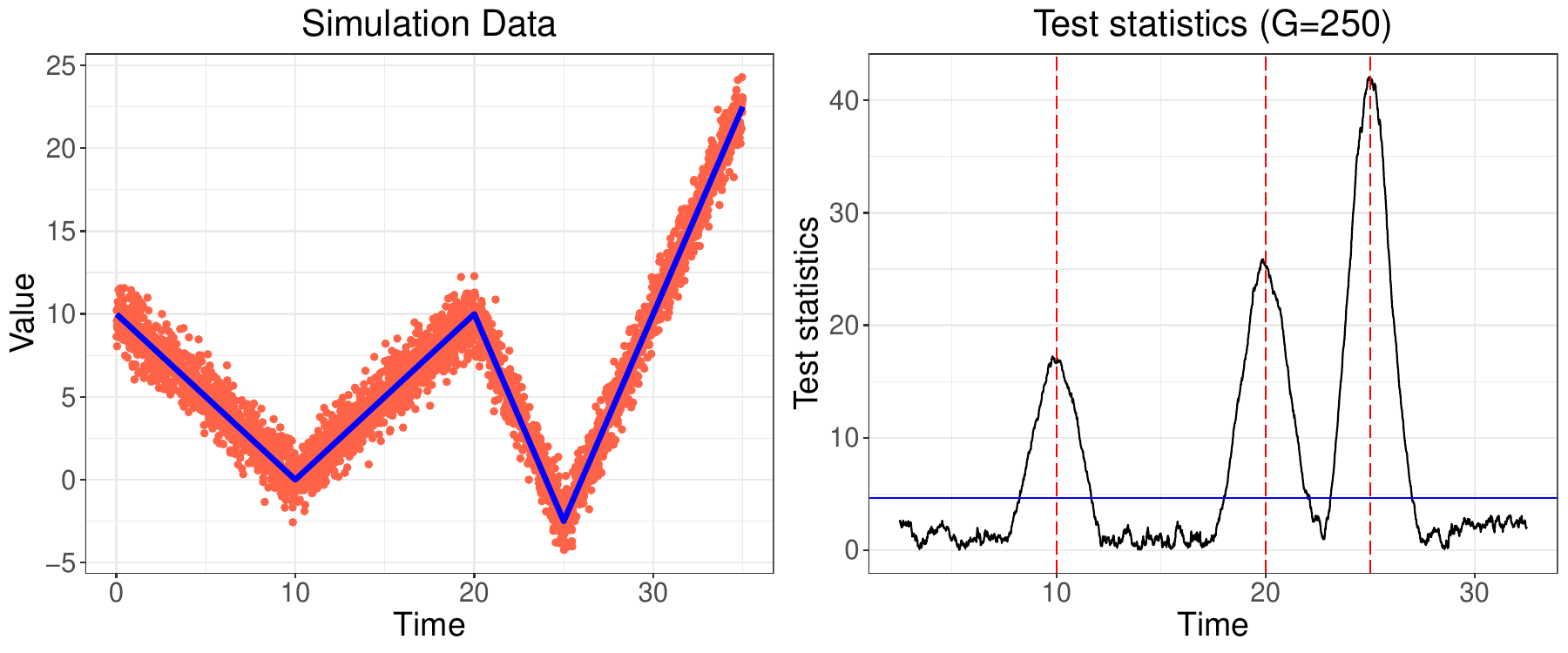}
\vspace{-2mm}
\caption{A realization from a dataset of length $n = 3500$ generated by the model~\ref{m:two} and the errors in~\ref{e:one} (see Section~\ref{sec:MOSUM_simul}) with $\sigma = 1$ (left) and the corresponding MOSUM statistics $W_{k,n}(G)$ with $G = 250$ for all $k = G, \ldots, n - G$ (right).
Estimated change point locations (vertical broken lines) and $C_n(G, \alpha)$ with $\alpha = 0.05$ (horizontal line) are also displayed.}
\label{W_kn_eg}
\end{figure}

By construction, the statistic $W_{k, n}(G)$ is expected to take large values in the intervals around the change points $k_j, \, j = 1, \ldots, J_n$, while its value is small for $k$ sufficiently far from the change points, see the right panel of Figure~\ref{W_kn_eg}. Therefore, we propose to estimate their locations with the local maximizers of $W_{k, n}(G)$ around which the statistics are significantly large by exceeding $C_n(G, \alpha)$. Specifically, motivated by the selection rule proposed by \cite{eichinger2018} for the mean change point detection problem, we identify all pairs of indices $(v_j, w_j), \, j = 1, \ldots, \wh J_n$, which simultaneously satisfy: (a)~$W_{k,n} (G) \ge C_n(G, \alpha)$ for $k \in \{ v_j, \ldots, w_j\}$, (b)~$W_{k,n} (G) < C_n(G, \alpha)$ for $k \in \{ v_j-1,  w_j+1 \}$, and (c)~$w_j - v_j \ge \eta G$ with a fixed $\eta \in (0, 1/2)$. 
Then, we estimate $J_n$ by $\wh J_n$ and the locations of the change points by 
$\wh k_j = \underset{v_j \le k \le w_j}{\arg\max} \, W_{k, n}(G)$ for $j = 1, \ldots, \wh J_n$.
With $\eta$ appropriately chosen, this rule allows for simultaneous estimation of all the $J_n$ change points without incurring any duplicate estimators.

\begin{rem}
\label{rem:pw:lin}
The statistic $W_{k, n}(G)$ in~\eqref{waldstat} 
bears  a resemblance to the Wald-type MOSUM statistic applied to the change point detection problem in linear regression \citep{kirch2022data}, a problem extensively studied in the change point literature \citep{csorgo1997, bai1998estimating, bai2003computation}. However, such methods have typically been analyzed under the (second-order) stationarity of the covariates, which precludes the existence of the (possibly) time-varying trend. As such, the investigation into the behavior of $W_{k, n}(G)$ requires a careful treatment of the presence of the trend when investigating the theoretical properties, which we discuss in Section~\ref{sec3}. 
\end{rem}


\subsection{Theoretical properties} 
\label{sec3}

\subsubsection{Asymptotic null distribution}
\label{sec:null}

In this section, we derive the asymptotic null distribution of $W_n(G)$ from which the critical value $C_n(G, \alpha)$ is obtained. 
On the stationary sequence $\{\epsilon_i\}_{i \in \mathbb{Z}}$, we require mild conditions permitting serial dependence and heavy-tailedness, which greatly relaxes the independence and (sub-)Gaussianity assumptions in the literature on piecewise linear modelling. 
\begin{enumerate}[noitemsep, wide, labelindent=0pt, label = (A\arabic*)]

\item \label{a:two} There exists a standard Wiener process $\{W(t): \, 0 \le t < \infty\}$ and $\nu > 0$ such that $\vert \sum_{i = 1}^{n} \epsilon_{i} - \tau W(n) \vert = O(n^{1/(2+\nu)})$ a.s.

\item \label{a:three} There exist constants $\gamma>2$ and $C_0, C_1 > 0$ such that, for any $0 \le \ell < r < \infty$, we have $\E ( \vert \sum_{i=\ell+1}^{r} \epsilon_i \vert^{\gamma} ) \le C_0 |r-\ell|^{\gamma/2}$ and $\E ( \vert \sum_{i=\ell+1}^{r} i \epsilon_i \vert^{\gamma} ) \le C_1 |r-\ell|^{3\gamma/2}$.
\end{enumerate}


The independence and (sub-)Gaussianity assumptions are commonly found in the literature on piecewise linear and polynomial modeling, see, e.g.,\ \cite{baranowski2019}, \cite{fearnhead2019}, \cite{yu2022localising} and \cite{maeng2019detecting}.
By contrast, we only require mild conditions permitting serial dependence and heavy-tailedness of $\{\epsilon_i\}_{i \in \mathbb{Z}}$.
The strong invariance assumed in~\ref{a:two} holds under a weak dependence condition of mixing-type \citep{kuelbs1980almost} or a functional dependence condition \citep{berkes2014komlos}.
In addition, the assumption~\ref{a:three} is shown to hold for many time series, see, e.g.,\ 
Lemma~\ref{iid_lem}.

Condition~\ref{c:one} below requires that away from the change points, the local estimator of LRV is consistent and bounded away from zero.
On the other hand, around the change points, it is sufficient to have $\wh{\tau}_k^2$ bounded, see~\ref{c:two}.

\begin{enumerate}[noitemsep, wide, labelindent=0pt, label = (B\arabic*)]
\item \label{c:one} We have $\vert \wh{\tau}_k^2 - \tau^2 \vert = o_P(\log^{-1}(n/G))$ and $\wh{\tau}_k^{-2}  = O_P(1)$ uniformly over all $k$ satisfying $\min_{1 \le j \le J_n} \vert k - k_j \vert \ge G$.
\item \label{c:two} $\max_{G \le k \le n - G} \wh{\tau}_k^{2} = O_P (1)$.
\end{enumerate}

We propose a MOSUM-based estimator of $\tau^2$ that satisfies~\ref{c:one} and~\ref{c:two}, see Remark~\ref{rem:est:var} for the case of independent $\{\epsilon_i\}_{i \in \mathbb{Z}}$ and 
Section~\ref{supp:sec:est:var} for the serially dependent setting. 

\begin{thm} \label{thm3.1}
Assume that~\ref{a:two} and~\ref{c:one} hold, and that the bandwidth $G$ satisfies
\begin{align}
G/n \to 0 \quad \text{and} \quad G^{-3/2} n^{1 + \frac{1}{2 + \nu}}\sqrt{\log(n)} \to 0. \label{eq:cond:G}
\end{align}
Then, under $\mc H_0: \, J_n = 0$, we have
$a_G {W_n(G)} - b_G \xrightarrow[n\to\infty]{d} \Gamma_2$, where
$a_G = \sqrt{2\log(n/G)}$, $b_G = 2 \log(n/G) + \log\log(n/G) + \log(H)$ with some constant $H > 0$,
and $\Gamma_2$ is a random variable following a Gumbel distribution with $\Pr(\Gamma_2 \le z) = \exp(-2\exp(-z))$. 
\end{thm}
Based on Theorem~\ref{thm3.1}, we select the critical value as $C_n(G, \alpha) = a_{G}^{-1}(b_{G} - \log(- \log (1-\alpha)/2))$ that controls the family-wise error rate at the given significance level $\alpha \in (0, 1)$.
Note that $C_n(G, \alpha)$ is fully determined by $n$, $G$, and $\alpha$ once the constant $H$ is set.
Related to the auto-covariance function of the bivariate Gaussian process $\mbf Z(t)$,
there are instances where $H$ can be specified exactly (see, for instance, \cite{steinebach1996})
but this is not the case in our setting. We discuss the choice of $H$ in Section~\ref{sec:supp:critical}.

\subsubsection{Consistency in multiple change point estimation} \label{subsec:consistency}


In order to measure the size of change at each $k_j$, we define $\b\Delta_j = (\Delta_{j}^{(0)}, \Delta^{(1)}_j)^{\top}$ with
$\Delta^{(0)}_j= (\alpha_{0, j} - \alpha_{0, j + 1}) + (\alpha_{1, j} - \alpha_{1, j+1}) t_{k_j}$ and
$\Delta^{(1)}_j = G ( \alpha_{1, j} - \alpha_{1, j+1}) \Delta t$.
Here, $\vert \Delta^{(0)}_j\vert$ denotes the size of any jump that occurs at the change point $k_j$ in $f_i$,
and $\vert \Delta^{(1)}_j \vert$ the size of a slope change at $k_j$. The multiplicative factor of $G$ in $\Delta^{(1)}_j$ is introduced in order to place the effects of the two types of changes in a comparable scale. 
Figure~\ref{fig_jump} provides a graphical illustration of $\Delta_j^{(0)}$ and $\Delta^{(1)}_j$.
We also define
\begin{align}
\label{eq:dj}
{d_j := \l\vert f_{k_j + 1} - 2 f_{k_j} + f_{k_j - 1} \r\vert
= \l\vert \Delta^{(0)}_j+ G^{-1} \Delta^{(1)}_j \r\vert,}
\end{align}
where in the relevant literature, $d_j$ or a closely related quantity is adopted to measure the size of changes. 
For instance, if $f_i$ is piecewise constant, then $\Delta^{(1)}_j = 0$ for all~$j$ and $d_j$ denotes the jump size at $k_j$. On the other hand, if $f_i$ is piecewise linear and continuous, then $\Delta^{(0)}_j = 0$ for all~$j$. With these definitions, the following conditions are imposed on the size of changes and the spacing between the change points.
\begin{figure}[!htb]
\centering
\includegraphics[height=0.25\textwidth]{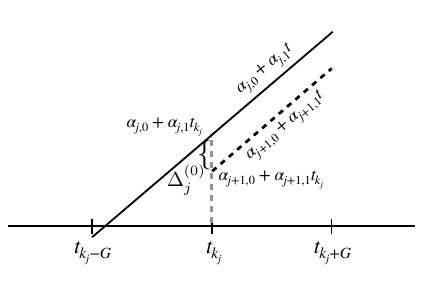}
\includegraphics[height=0.25\textwidth]{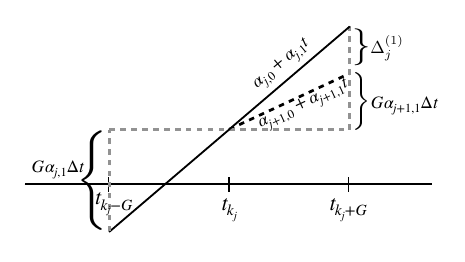}
\vspace{-2mm}
\caption{Left: the slope of $f_i$ remains unchanged while a discontinuous jump occurs at $t_{k_j}$.
Right: the slope of $f_i$ changes but the signal remains continuous. In both plots, linear curves corresponding to the pre- and the post-change parameters are given in solid and dashed lines, respectively, and the size of $\Delta_j^{(0)}$ (left) and $\Delta_j^{(1)}$ (right) is highlighted.}
\label{fig_jump}
\end{figure}

\begin{enumerate}[noitemsep, wide, labelindent=0pt, label = (C\arabic*)]
\item \label{b:one} $\min_{0 \le j \le J_n} \vert k_{j+1} - k_{j}\vert > 2G$. 
\item \label{b:two} $(\log(n/G))^{-1/2} \sqrt{G} \min_{1 \le j \le J_n} \Vert \bDelta_{j} \Vert \to \infty$ as $n \to 
\infty$. 
\end{enumerate}
Assumption~\ref{b:one} requires that the bandwidth $G$ does not exceed half the distance between any two adjacent change points.
Provided that~\ref{b:one} is met, we permit $J_n \to \infty$ as $n \to \infty$.
Jointly, \ref{b:one}--\ref{b:two} place a lower bound on the size of changes for their detection, namely
\begin{align}
\label{eq:detection:lb}
\frac{\min_{0 \le j \le J_n} (k_{j + 1} - k_j) \cdot \min_{1 \le j \le J_n} \Vert \bDelta_{j} \Vert^2}{\log(n)} \to \infty.
\end{align}

Then, Theorem~\ref{5.1} establishes the consistency of the MOSUM procedure in detecting multiple change points and derives the rate of localization.
\begin{thm} \label{5.1}
Assume that \ref{a:two}--\ref{a:three}, \ref{c:one}--\ref{c:two}, and \ref{b:one}--\ref{b:two} are held.
Suppose that $G$ satisfies~\eqref{eq:cond:G} and $\alpha = \alpha_n$ is chosen such that
\begin{align} 
\label{alpha_n}
\alpha_n \to 0 \quad \text{and} \quad \log^{-1/2}(n/G) C_n(G, \alpha_n) = O(1).
\end{align}
Then, as $n \to \infty$, the set of change point estimators $\{\wh k_j, \, j = 1, \ldots, \wh J_n: \, \wh k_1 < \ldots < \wh k_{\wh J_n}\}$ returned by the MOSUM procedure satisfies:
\begin{enumerate}[label = (\roman*)]
\item \label{thm:est:one} 
$\Pr \left(\wh{J}_n = J_n  \text{ and } \max_{1 \le j  \le J_n} \vert \wh{k}_j - k_j \vert < G \right) \to 1$.

\item \label{thm:est:two} Additionally, assume that $f_i$ is piecewise linear and continuous.
Also, we use $\wh\tau_k^2 = \wh\tau^2$ which satisfies $\Pr(\wh\tau^2 > 0) \to 1$ and $\wh\tau^2 = O_P(1)$.
Then, there exists a fixed constant $c_0 > 0$ such that, for all $c_0 \le \xi \le G$ and $1 \le j \le J_n$,
$\Pr \left( \vert \wh{k}_j \mathbb{I}_{\{j \le \wh J_n\}} - k_j \vert \ge \xi \right) = O\l( d_j^{-\gamma} \xi^{-3\gamma/2} \r) + o(1)$, where
$d_j = \vert f_{k_j + 1} - 2 f_{k_j} + f_{k_j - 1} \vert = \vert G^{-1} \Delta^{(1)}_j \vert$.
\end{enumerate}
\end{thm}

Theorem~\ref{5.1}~\ref{thm:est:one} shows that the MOSUM procedure achieves consistency in estimating~$J_n$, and it locates a single estimator within the interval of length $G$ from each $k_j$. 
Further, when $f_i$ is continuous, \ref{thm:est:two}~derives the rate of localization, which implies that
$\max_{1 \le j \le J_n} d_j^{2/3} \vert \wh k_j - k_j \vert = O_P ( J_n^{1/\gamma})$.
In particular, when $d_j = O(n^{-1})$ (which ensures the boundedness of $f_i$) and the number of change points is finite (i.e.,\ $J_n = O(1)$), the resultant rate, $\max_{1 \le j \le J_n} \vert \wh k_j - k_j \vert = O_P(n^{2/3})$, matches the minimax lower bound derived in \cite{raimondo1998minimax} in the context of a single sharp `cusp' estimation (see Theorem~4.6 therein).
In Theorem~\ref{5.1}~\ref{thm:est:two}, the condition that $\wh{\tau}_k^2 = \wh\tau^2$ is made for the ease of the proof, and this estimator $\wh\tau^2$ is required only to be bounded appropriately without being consistent. 


\begin{rem}[Comparison with the existing results]
\label{rem:comp}
Our theoretical results are derived under considerably weaker conditions compared to the existing literature.
Most notably, we assume that $\E(\vert \epsilon_i \vert^\gamma) < \infty$ with some finite $\gamma > 2$ only {through~\ref{a:three}}, whereas it is commonly assumed that $\{\epsilon_i\}$ is a sequence of (sub-)Gaussian random variables with the exception of \cite{maeng2019detecting},
and the latter still requires that {\it all} moments of $\epsilon_i$ exist. 
When continuity is imposed on $f_i$ (such that $\Delta^{(0)}_j= 0$),
the condition~\eqref{eq:detection:lb} 
is analogous to those found in \cite{baranowski2019} (permit diverging $J_n$ as in this paper) and \cite{fearnhead2019} (assume $J_n = O(1)$).
Also, in this case, the rate of localization obtained in Theorem~\ref{5.1}~\ref{thm:est:two}
is comparable to those obtained in the above papers,
or even sharper when the number of change points grows slowly as $J_n = o(\log(n))$.
We mention that \cite{maeng2019detecting} and \cite{yu2022localising} derived the rate of localization without assuming continuity; 
we defer the discussion of the case of discontinuous $f_i$ to the Supplementary Material.
\end{rem}

\begin{rem}[Variance estimation] 
\label{rem:est:var} 
There exist estimators of the variance $\sigma^2$ and LRV $\tau^2$ that are robust to the presence of multiple mean shifts \citep{eichinger2018, dette2020multiscale, chan2022optimal, mcgonigle2023robust}, which are combined with the mean change point detection procedures.
We propose a MOSUM-based local estimator of LRV that extends the estimator of \cite{eichinger2018} in 
Section~\ref{supp:sec:est:var} and show that it fulfils~\ref{c:one}--\ref{c:two} (Theorem~\ref{thm_sigma}).
In the special case of independent $\{\epsilon_i\}_{i \in \mathbb{Z}}$ where $\tau^2 = \sigma^2$, 
the proposed estimator for $k = G, \ldots, n - G$, is
\begin{align} \label{sigma_k}
\wh{\sigma}_k^2 = \frac{1}{2} \l(\wh{\sigma}_{k, -}^{2} + \wh{\sigma}_{k, +}^{2} \r),
\quad \text{where} \quad
\wh\sigma_{k, \pm}^2 = \frac{1}{G - 2} \sum_{i \in \cI^{\pm}(k)} \left( X_i - \wh{\beta}_0^{\pm}{(k)} - \frac{i - k}{G} \wh{\beta}_1^{\pm}{(k)} \right)^2.
\end{align}
\end{rem}

\begin{rem}
\label{rem:discont}
When $f_i$ is piecewise linear and continuous, the quantity $\wt{W}_k(G) = \Vert \b\Sigma^{-1/2} [\bbeta^{+}(k) - \bbeta^{-}(k)] \Vert$ (where $\bbeta^{\pm}(k)$ is obtained by regressing $f_i$ on $\mathbf{x}_{i, k}$) 
attains a single local maximum at each $k_j$, which leads to the desirable behaviour of $W_{k, n}(G)$ observed in Figure~\ref{W_kn_eg} and the localization property in Theorem~\ref{5.1}~\ref{thm:est:two}.
If $f_i$ is discontinuous (i.e.,\ $\Delta_{j}^{(0)} \neq 0$), $\wt{W}_k(G)$ attains multiple peaks within the interval $\{k_j - G + 1, \ldots, k_j + G\}$ and, although one peak is located at $k_j$, it is not necessarily the local maximizer. However, combined with the local variance estimator proposed in~\eqref{sigma_k}, the statistic $W_{k, n}(G)$ tends to attain clear local maxima at the true change points due to the upward bias in $\wh{\sigma}_k^2$ at $k \ne k_j$, and hence, performs well empirically. See 
Section~\ref{sec:discont} for further discussions, and Section~\ref{sec:tuning} for how this behaviour may be exploited for the diagnosis of the types of changes. 
\end{rem}

\section{Numerical considerations}
\label{sec:three}

\subsection{Computational complexity}
\label{sec:comp}

The computational complexity of the proposed MOSUM procedure is $O(n)$.
This is due to the sequential update available for the coefficient estimators $\wh{\b\beta}^{\pm}(k)$ and the local variance estimator in~\eqref{sigma_k}, see 
Section B.2 for the updating equations.
In Section~\ref{sec:time}, we numerically demonstrate the competitiveness of the proposed MOSUM procedure, where it takes a fraction of the time taken for other methods to process large datasets.

\subsection{Multiscale extension}
\label{sec:multiscale}

If the bandwidth $G$ is chosen too small, the MOSUM procedure may lack detection power, 
while when $G$ is too large, the violation of the condition~\ref{b:one} makes it difficult to detect or locate change points which are close to one another.
Generally, it is well-recognized in the literatures that a moving window-type procedure applied with a single bandwidth lacks adaptivity.
One remedy is to apply the procedure with multiple bandwidths,
say $\mc G = \{G_b, \, 1 \le b \le B: \, G_1 < \ldots < G_B\}$
and prune down the set of estimators to remove any duplicate estimators.
Let $\wh{\mc K}_b = \{\wh k_{b, j}, \, 1 \le j \le \wh J_b \}$ 
denote the set of estimators obtained with $G_b$ as the bandwidth, where $\wh k_{b, j}$ are ordered in the decreasing order of the corresponding MOSUM statistic, i.e.,\ $W_{\wh{k}_{b, 1}, n}(G_{b})  \ge W_{ \wh{k}_{b, 2}, n}(G_{b}) \ge \ldots \ge W_{ \wh{k}_{b, \wh J_b}, n}(G_{b})$. 
Supposing that the bandwidths are sorted according to some measure of importance as $G_{\pi(b)}, \, b = 1, \ldots, B$, we propose to sequentially accept $\wh k_{j, \pi(b)}$ for $j = 1, \ldots, \wh J_{\pi(b)}$ and $b = 1, \ldots, B$, to the set of final estimators $\wh{\mathcal K}$ if $\wh k_{j, \pi(b)}$ is sufficiently distanced away from the already accepted estimators.
That is, starting with $\wh{\mathcal K} = \emptyset$, we check whether $\min_{\wh k \in \wh{\mc K}} \vert \wh k_{j, \pi(b)} - \wh k \vert > \theta G_{\pi(b)}$ for increasing $j$ and $b$, with a pre-determined constant $\theta \in (0, 1]$, and if so, accepts $\wh k_{j, \pi(b)}$ to $\wh{\mathcal K}$ (with the convention $\min \emptyset = \infty$).

When the bandwidths are sorted in the increasing order such that $G_{\pi(b)} = G_b$, this coincides with the bottom-up merging proposed by \cite{messer2014}.
Instead, we propose to adopt the Bayesian information criterion ${\rm BIC}(\mc K) = n \log\l( \frac{{\rm RSS}(\mc K)}{n} \r) + 2 \big( \vert \mc K \vert + 1 \big) \log(n)$ for bandwidth sorting, where ${\rm RSS}(\mc K)$ denotes the residual sum of squares of the model fitted under~\eqref{eq:model} with $\mc K$ as the set of change points. 
Then, we find $\pi(\cdot)$ satisfying ${\rm BIC}(\wh{\mc K}(G_{\pi(1)})) \le {\rm BIC}(\wh{\mc K}(G_{\pi(2)})) \le \ldots \le {\rm BIC}(\wh{\mc K}(G_{\pi(B)}))$.
Although not reported here, we numerically examined the use of alternative information criteria such as AIC and the cross-validation measure of \cite{zou2020consistent}, which performed similarly well as the proposed BIC-based sorting. On the other hand, the bottom-up merging tends to produce more false positives and attain poorer localization accuracy by preferring the estimators from the finer bandwidths, a phenomenon also observed by \cite{cho2019twostage} in the context of the univariate mean change point detection problem. Investigating whether the results reported in Theorem~\ref{5.1} extend to the multiscale procedure is interesting, but it is beyond the scope of this paper, which we leave for future research. 


\subsection{Practical issues in implementation} 
\label{sec:tuning}

\paragraph{Critical value.}
The theoretically motivated critical value $C_n(G, \alpha)$ given in Section~\ref{sec:null} requires the selection of $\alpha$.
In view of the condition on $\alpha$ in~\eqref{alpha_n}, we use $\alpha = 0.05$ throughout this paper. 
It also involves some unknown constant $H$ through $b_G$, which we set $\log(H) \approx 0.7284$ based on extensive numerical experiments, see 
Section~\ref{sec:supp:critical} for details.
\vspace{-10pt}

\paragraph{Bandwidths.}
For the multiscale MOSUM procedure described in Section~\ref{sec:multiscale}, we use a set of bandwidths $\{G_1, \ldots, G_B\}$ generated as a Fibonacci sequence following \cite{cho2019twostage}. 
Namely, for given $G_0 = G_1$, we generate $G_b = G_{b - 1} + G_{b - 2}$ for $b \ge 2$ until $G_{B + 1}$ exceeds $n/\log_{10}(n)$ while $G_B < n/\log_{10}(n)$.
In view of the condition~\eqref{eq:cond:G}, we adopt $G_1 = 10$ ($n=500$) or $G_1 = 50$ ($n = 2500, 3500$) in simulation studies 
and $G_1 = 100$ for real data analysis, which are set to be greater than $0.01n$ for the sample size in consideration.
\vspace{-10pt}

\paragraph{Tuning parameters $\eta$ and $\theta$.}
In our numerical experiments, varying the value of $\eta$ used in the estimation rule does not lead to noticeably different performance within the range $\eta \in [0.2, 0.4]$, and a similar conclusion is drawn for the choice of $\theta$ adopted in the multiscale extension, provided that $\theta \ge 0.8$.
As a rule of thumb, we recommend $(\eta, \theta) = (0.3, 0.8)$, since choosing too large values for these parameters may prevent detection of some change points.  
\vspace{-10pt}

\paragraph{Diagnostic.} Investigating whether a change relates to a continuous change in the slope ($\Delta_j^{(0)} = 0$) or a discontinuous jump (due to a change of the intercept, $\Delta_j^{(0)} \ne 0$) can be interesting in practice. For this purpose, we can adopt the visualization of the MOSUM statistics $W_{k, n}(G)$ as a diagnostic tool, based on the fact that $W_{k,n}(G)$ behaves differently around the change point $k_j$ depending on the values of $\Delta_j^{(0)}$ and $\Delta_j^{(1)}$, i.e.,\ $W_{k,n}(G)$ has unimodal peak around $k_j$ only when $\Delta_j^{(0)} = 0$. We provide further illustrative examples in Section~\ref{sec:discont}. 

\section{Numerical experiments}
\label{sec:MOSUM_simul}

\subsection{Simulation studies}
\label{sec:sim}

\paragraph{Data generation.}
We consider the following different scenarios for the generation of $f_i$: 
\begin{enumerate*}[label=(M\arabic*)] 
\setcounter{enumi}{-1}
\item \label{m:zero} no change point ($J = 0$),
\item \label{m:one} piecewise linear with three change points ($J = 3$),
\item \label{m:two} piecewise linear and continuous with $J = 3$,
\item \label{m:three} piecewise linear with $J = 6$, and
\item \label{m:four} piecewise constant with $J = 3$. 
\end{enumerate*} 
We have $n = 3500$ under~\ref{m:zero}--\ref{m:two} and~\ref{m:four}, while $n = 2500$ under~\ref{m:three}.
See 
Section~\ref{sec:comp:sim} for a detailed description, which also reports results obtained with a shorter sample size ($n = 500$). 
In all cases, we have $t_i = 0.01i$. 
We note that the sample sizes are comparable to those of \cite{baranowski2019} and \cite{maeng2019detecting}.
For the generation of $\{\epsilon_i\}_{i = 1}^n$, we consider a sequence of i.i.d.\ random variables with 
\begin{enumerate*}[label=(E\arabic*)] 
\item \label{e:one} Gaussian,
\item \label{e:two} scaled $t_5$,
\item \label{e:three} scaled Laplace distributions, as well as
\item \label{e:four} an AR process: $\epsilon_i = \rho \epsilon_{i - 1} +  \sqrt{1 - \rho^2} \sigma Z_i$ with $\rho \in \{ 0.3, 0.7 \}$ and $Z_i \sim_{\text{i.i.d.}} \mc N(0, 1)$.
\end{enumerate*}
We vary $\Var(\epsilon_i) = \sigma^2$ with $\sigma \in \{ 0.5, 1, 1.5, 2 \}$, but only report the results with $\sigma = 1$ in the main text unless specified otherwise; see 
Section~\ref{sec:comp:sim:res} for the full results.
\vspace{-10pt}

\paragraph{Tuning parameters and competitors.}
We apply the multiscale extension of the MOSUM procedure, referred to as `MOSUM' below. The tuning parameters, including the set of bandwidths, are chosen as described in Section~\ref{sec:tuning}. Also, unless stated otherwise, we use the MOSUM-based local estimator of variance given in~\eqref{sigma_k} as $\wh{\tau}_k$. For comparison, we include the narrowest-over-threshold method proposed by \cite{baranowski2019}, the $\ell_0$-penalized least squares estimation method of \cite{fearnhead2019}, and the wavelet-based method of \cite{maeng2019detecting}, referred to as NOT.pwLin [R package {\bf not}], CPOP [R package {\bf cpop}] and TGUW [R package {\bf trendsegmentR}], respectively. NOT.pwLin takes as an input whether $f_i$ is continuous and accordingly, we separately report the results from NOT.pwLin with the continuity imposed (NOT.pwLinCont). These are applied along with the recommended default tuning parameters.
\vspace{-10pt}

\paragraph{Performance metrics.}
For each setting, we report the results from $1000$ replications according to the following measures of performance.
Let $\mc K = \{t_{k_j}, \, 1 \le j \le J: \, k_1 < \cdots < k_J\}$ denote the set of true change points,
and $\wh{\mc K} = \{t_{\wh k_j}, \, 1 \le j \le \wh{J}: \, \wh k_1 < \cdots < \wh k_{\wh J} \}$
the set of change point estimators.
Then, we compute $\text{\textsf{COUNTscore}} = \vert\wh{J}- J\vert$,
$\text{\textsf{MAXscore1}} = \max_{1 \le j \le J} \min_{1 \le j^{\prime} \le \wh{J}} \vert t_{\wh{k}_{j^{\prime}} }  - t_{k_j} \vert$ and $\text{\textsf{MAXscore2}} = \max_{1 \le j^{\prime} \le \wh{J}} \min_{1 \le j \le J} \vert t_{\wh{k}_{j^{\prime}}}   - t_{k_j} \vert$.
\textsf{COUNTscore} evaluates the accuracy in estimating $J$, and \textsf{MAXscore1} and \textsf{MAXscore2} assess both detection and localization accuracy. 
\textsf{MAXscore1} is large when a true change point is undetected (false negative), while \textsf{MAXscore2} is large when a spurious estimator is detected far from true change points (false positive).
For all three, smaller values indicate better performance.

\paragraph{Results.}
Table~\ref{sim:tab:m1} shows that the proposed MOSUM procedure accurately estimates the total number and locations of the change points. Applied to the datasets generated under~\ref{m:one}--\ref{m:three}, MOSUM performs as well as, or slightly outperforms, NOT.pwLin regardless of the error distribution or the types of changes and their frequency. TGUW tends to perform worse than MOSUM or NOT.pwLin in all settings, both in terms of detection and estimation accuracy, and its performance deteriorates much more severely when the errors are generated from heavy-tailed distributions under \ref{e:two}--\ref{e:three}. NOT.pwLinCont and CPOP pre-suppose that the signals are piecewise linear and continuous, and as such, they perform well under~\ref{m:two} but poorly in other scenarios, and tend to over-estimate the number of change points by detecting spurious estimators in order to approximate the discontinuous signal by introducing additional segments (see Table~\ref{supp:simul1.table.bic} and Figure~\ref{fig_contapprox}).
The detection performance of MOSUM under~\ref{m:two} is not far behind NOT.pwLinCont and slightly better than NOT.pwLin.
Also, in this scenario, the localization accuracy measured by \textsf{MAXscore1} and \textsf{MAXscore2} becomes worse in the presence of heavy-tailed errors for all methods. 
Under~\ref{m:three}, the signal has frequent change points ($J = 6$) and the distance between adjacent change points is shorter; the smallest distance between change points is $100$ in comparison with $500$ under \ref{m:one}--\ref{m:two}.
Here, MOSUM is still highly competitive and outperforms the competitors in estimation accuracy, particularly when the data is heavy-tailed.

Table~\ref{sim:tab:e4} concerns the case of serially correlated errors generated under~\ref{e:four}. 
We consider two approaches: (i)~we continue to use the variance estimator in~\eqref{sigma_k} for calibration (`MOSUM') with ignoring the serial dependence, and (ii)~we use the difference-based estimator of the LRV proposed in \cite{chan2022optimal} obtained from the entire sample (`MOSUM.dlrv').
In the presence of weak serial dependence ($\rho = 0.3$), MOSUM does reasonably well in not returning spurious estimators.
However, when the serial dependence becomes stronger with $\rho = 0.7$, such an approach suffers from the calibration issue, for which MOSUM.dlrv provides a reasonably good solution.
Even so, the performance is worse than the independent setting as the signal-to-noise ratio decreases with increasing $\rho$. 
The implementation of NOT.pwLin do not permit the user to supply an alternative scaling parameter, and the default choice fails to adequately suppress the spurious false positives; TGUW also performs worse although its implementation accommodates serial dependence. 

Table~\ref{sim:tab:m4} considers the case when the signal is piecewise constant with $\alpha_{1, j} = 0$ in~\eqref{eq:model}. Here, we additionally consider NOT.pwConst (`piecewise constant') as proposed in \cite{baranowski2019} besides MOSUM, NOT.pwLin, and TGUW.
The MOSUM procedure shows comparable or better performance than NOT.pwConst regardless of $n$ when the noise level is small ($\sigma = 1$), but its performance deteriorates when $\sigma = 2$.
Due to increased noise level, MOSUM sometimes approximates the signal with three linear segments rather than two constant segments around $k_1$ (see Figure~\ref{fig:M4:MOSUM}).

Finally, we examine the performance of different methods when no change point is present ($J = 0$), see Table~\ref{sim:tab:m0}. 
We observe that MOSUM successfully avoids detecting any spurious estimators in almost all realizations, even when the data is heavy-tailed under \ref{e:two}--\ref{e:three}.
NOT-based methods work well even when $\epsilon_i$ is not Gaussian, and generally, NOT.pwLinCont is more conservative than NOT.pwLin.
CPOP and TGUW suffer greatly when the error distribution is heavy-tailed, and the former, in particular, detects a large number of false positives.


\begin{table}[!htb]
\centering
\caption{\small \textbf{\ref{m:one}--\ref{m:two} with $J = 3$ and $n = 3500$ and \ref{m:three} with $J = 6$ and $n = 2500$.} Results from MOSUM, NOT.pwLin, NOT.pwLinCont, CPOP and TGUW when the errors are generated as in \ref{e:one}--\ref{e:three} with $\sigma = 1$.
We report the average and standard error (in parentheses) of the performance metrics over $1000$ realizations.}
\label{sim:tab:m1}
\resizebox{\columnwidth}{!}{\footnotesize
\begin{tabular}{ c c c c c c c c}
\toprule
Model &  Error & Metric & MOSUM & NOT.pwLin & NOT.pwLinCont & CPOP & TGUW \\ 
\cmidrule(lr){1-3} \cmidrule(lr){4-8}
\ref{m:one} & \ref{e:one} & \textsf{COUNTscore} & 0.001 (0.0316) & 0.003 (0.0547) & -- & -- & 0.029 (0.1679)\\
& & \textsf{MAXscore1} & 0.088 (0.0601) & 0.123 (0.0653) & -- & -- & 0.152 (0.1142)\\
& & \textsf{MAXscore2} & 0.093 (0.1545) & 0.131 (0.2561) & -- & -- & 0.158 (0.1386)\\
\cmidrule(lr){2-3} \cmidrule(lr){4-8}
& \ref{e:two} &    \textsf{COUNTscore} & 0 (0) &0.023 (0.1962) & -- & -- & 0.472 (1.1741)\\
& & \textsf{MAXscore1} & 0.083 (0.0574) & 0.117 (0.0657) & -- & -- & 0.172 (0.1219)\\
& & \textsf{MAXscore2} & 0.083 (0.0574) & 0.184 (0.776) & -- & -- & 0.84 (1.9763)\\
\cmidrule(lr){2-3} \cmidrule(lr){4-8}
& \ref{e:three} & \textsf{COUNTscore} & 0 (0) & 0.002 (0.0632) & -- & -- & 0.64 (1.3506)\\
& & \textsf{MAXscore1} & 0.083 (0.0582) & 0.117 (0.0664) & -- & -- & 0.176 (0.1259)\\
& & \textsf{MAXscore2} & 0.083 (0.0582) & 0.124 (0.2406) & -- & -- & 1.155 (2.3102)\\\cmidrule(lr){1-3} \cmidrule(lr){4-8}

\ref{m:two} & \ref{e:one} & \textsf{COUNTscore} & 0 (0) & 0 (0) & 0 (0) & 0.003 (0.0547) & 0.069 (0.2652)\\
& & \textsf{MAXscore1} & 0.186 (0.0883) & 0.262 (0.1016) & 0.047 (0.0253) & 0.05 (0.0261) & 0.372 (0.1817)\\
& & \textsf{MAXscore2} & 0.186 (0.0883) & 0.262 (0.1016) & 0.047 (0.0253) & 0.054 (0.0969) & 0.393 (0.2488)\\
\cmidrule(lr){2-3} \cmidrule(lr){4-8}
& \ref{e:two} & \textsf{COUNTscore} &  0.003 (0.0547) & 0.019 (0.1633) & 0.004 (0.0632) & 5.989 (4.1427) & 0.659 (1.237)\\
& & \textsf{MAXscore1} & 0.336 (0.5639) & 0.454 (0.5608) & 0.123 (0.5478) & 0.204 (0.3175) & 0.751 (0.6129)\\
& & \textsf{MAXscore2} & 0.306 (0.1857) & 0.482 (0.728) & 0.102 (0.311) & 4.634 (3.2025) & 1.501 (1.9696)\\ 
\cmidrule(lr){2-3} \cmidrule(lr){4-8}
& \ref{e:three} & \textsf{COUNTscore} & 0.002 (0.0447) & 0.016 (0.1542) & 0.002 (0.0447) & 3.606 (3.3289) & 0.866 (1.443)\\
& & \textsf{MAXscore1} & 0.314 (0.4752) & 0.424 (0.4657) & 0.111 (0.4489) & 0.168 (0.3693)  & 0.747 (0.5885)\\
& & \textsf{MAXscore2} & 0.294 (0.182)   & 0.433 (0.4659) & 0.092 (0.0677) & 3.221 (3.3163) & 1.761 (2.2033)\\ 
\midrule

\ref{m:three} & \ref{e:one} & \textsf{COUNTscore} & 0 (0) & 0.003 (0.0547) & -- & -- & 0.196 (0.5328)\\
& & \textsf{MAXscore1} & 0.182 (0.0943) & 0.244 (0.0965) & -- & -- & 0.319 (0.1676)\\
& & \textsf{MAXscore2} &  0.182 (0.0943) & 0.248 (0.1671) & -- & -- & 0.433 (0.4911)\\
\cmidrule(lr){2-3} \cmidrule(lr){4-8}
& \ref{e:two} & \textsf{COUNTscore} & 0 (0) & 0.025 (0.1685) & -- & -- & 0.485 (0.9406)\\
& & \textsf{MAXscore1} & 0.18 (0.0917) & 0.239 (0.1029) & -- & -- & 0.384 (0.1921)\\
& & \textsf{MAXscore2} & 0.18 (0.0917) & 0.248 (0.1997) & -- & -- & 0.607 (0.6723)\\
\cmidrule(lr){2-3} \cmidrule(lr){4-8}
& \ref{e:three} &  \textsf{COUNTscore} & 0 (0) & 0.007 (0.0947) & -- & -- & 0.518 (1.0221)\\
& & \textsf{MAXscore1} & 0.177 (0.0976) & 0.244 (0.0999) & -- & -- & 0.392 (0.1863)\\
& & \textsf{MAXscore2} & 0.177 (0.0976) & 0.251 (0.203)   & -- & -- & 0.646 (0.7245)\\
\bottomrule
\end{tabular}
}
\end{table}

\begin{table}[!htb]
\centering
\caption{\small \textbf{\ref{m:one} with $J = 3$ and $n = 3500$.} Results from the MOSUM, MOSUM.dlrv (MOSUM applied with the LRV estimator of \cite{chan2022optimal}), NOT.pwLin and TGUW when \textbf{the errors are generated as in \ref{e:four}} with $\sigma = 1$. We report the average and standard error (in parentheses) of the performance metrics over $1000$ realizations.}
\label{sim:tab:e4}
{\footnotesize 
\begin{tabular}{ cc cccc }
\toprule
$\rho$ & Metric   & MOSUM & MOSUM.dlrv & NOT.pwLin & TGUW  \\
\cmidrule(lr){1-2} \cmidrule(lr){3-6}
$0.3$ & \textsf{COUNTscore} & 0.044 (0.2369) & 0.492 (0.5237) & 0.123 (0.4517) & 2.236 (2.1272)\\
& \textsf{MAXscore1} &  0.101 (0.0686) & 0.44 (0.3841) & 0.142 (0.0795) & 0.171 (0.1247)\\
&\textsf{MAXscore2} & 0.24 (0.8439) & 0.658 (0.5796) & 0.442 (1.3325) & 3.403 (3.1662)\\
\cmidrule(lr){1-2} \cmidrule(lr){3-6}
$0.7$ & \textsf{COUNTscore} & 8.233 (2.8791) & 0.72 (0.5549) & 12.368 (6.1324) & 123.405 (10.9049)\\
& \textsf{MAXscore1} & 0.151 (0.1231) & 0.558 (0.3844) & 0.219 (0.1727) & 0.104 (0.095)\\
& \textsf{MAXscore2} & 7.594 (1.5732) & 0.953 (0.5848) & 7.784 (2.2462) & 9.806 (0.1351)\\
\bottomrule
\end{tabular}
}
\end{table}

\begin{table}[!htb]
\centering
\caption{\small \textbf{\ref{m:four} with $J = 3$ and $n = 3500$.} Results from MOSUM, NOT.pwLin, NOT.pwConst, and TGUW when the errors are generated as in \ref{e:one} with $\sigma \in \{1, 2\}$.
We report the average and standard error (in parentheses) of the performance metrics over $1000$ realizations.
}
\label{sim:tab:m4}
{\footnotesize
\begin{tabular}{c c c c c c}
\toprule
$\sigma$ & Metric & MOSUM & NOT.pwLin & NOT.pwConst & TGUW \\ 
\cmidrule(lr){1-2} \cmidrule(lr){3-6}
$1$ &  \textsf{COUNTscore} &  0 (0) & 0.001 (0.0316) & 0.008 (0.0891) & 0.073 (0.2603)\\
& \textsf{MAXscore1} &  0.001 (0.0026) & 0.001 (0.0026) & 0.001 (0.0026) & 0.006 (0.0131)\\
& \textsf{MAXscore2} & 0.001 (0.0026) & 0.001 (0.0048) & 0.035 (0.4835) & 0.008 (0.0185)\\
\cmidrule(lr){1-2} \cmidrule(lr){3-6}

$2$ &  \textsf{COUNTscore} & 0.162 (0.382) & 0 (0) & 0.019 (0.1366) & 0.263 (0.4879)\\
& \textsf{MAXscore1} & 0.121 (0.2746) & 0.007 (0.0108) & 0.007 (0.011) & 0.049 (0.0761)\\
& \textsf{MAXscore2} & 0.155 (0.3468) & 0.007 (0.0108) & 0.062 (0.5403) & 0.068 (0.1884)\\
\bottomrule
\end{tabular}
}
\end{table}

\begin{table}[!htb]
\centering
\caption{\small \textbf{\ref{m:zero} with $J = 0$ and $n = 3500$.} Results from the MOSUM, NOT.pwLin, NOT.pwLinCont, CPOP and TGUW when the errors are generated as in \ref{e:one}--\ref{e:three}. We report the average and standard error (in parentheses) of \textsf{COUNTscore} over $1000$ realizations.}
\label{sim:tab:m0}
{\footnotesize
\begin{tabular}{ cc  ccccc }
\toprule
 Error    & $\sigma$     & MOSUM & NOT.pwLin & NOT.pwLinCont & CPOP & TGUW \\ 
\cmidrule(lr){1-2} \cmidrule(lr){3-7}
\ref{e:one} 	&		      $0.5$ & 0 (0)                 & 0 (0) & 0 (0)                 & 0.002 (0.0447) & 0.002 (0.0632)\\
& $1$    & 0 (0)                 & 0 (0) & 0 (0)                 & 0 (0)                 & 0 (0)                \\
& $1.5$ & 0.001 (0.0316) & 0 (0) & 0.001 (0.0316) & 0.001 (0.0316) & 0 (0)                \\
& $2$    & 0 (0)                 & 0 (0) & 0 (0)                 & 0 (0)                 & 0 (0)                \\
\cmidrule(lr){1-2} \cmidrule(lr){3-7}
\ref{e:two} & $0.5$  & 0 (0) & 0.022 (0.2039) & 0 (0)                 & 5.668 (4.1193)  & 0.572 (1.5822) \\
& $1$    &  0 (0) & 0.021 (0.1913) & 0.001 (0.0316) & 5.993 (4.1367)  & 0.516 (1.5219) \\
& $1.5$  & 0 (0) & 0.008 (0.1482) & 0 (0)                 & 5.883 (4.2099)  & 0.558 (1.5694) \\
& $2$    &  0 (0) & 0.015 (0.1637) & 0.002 (0.0632) & 5.773 (4.2525)  & 0.62 (1.7054) \\ 
\cmidrule(lr){1-2} \cmidrule(lr){3-7}
\ref{e:three}&    $0.5$ & 0 (0) & 0.01 (0.1411)   & 0 (0) & 3.362 (3.4004)  & 0.849 (1.9764) \\
& $1$    & 0 (0) & 0.017 (0.1864) &0 (0) & 3.397 (3.2376)  & 0.904 (1.9567)\\
& $1.5$ & 0 (0) & 0.013 (0.1577) &0 (0) & 3.184 (3.0465)  & 0.859 (1.9367)\\
& $2$    & 0 (0) & 0.004 (0.0894) &0 (0) & 3.292 (3.2477)  & 0.754 (1.8409)\\
\bottomrule
\end{tabular}
}
\end{table}

\subsection{Execution time}
\label{sec:time}

We report the average execution time of different methods when applied to $100$ realizations generated under~\ref{m:one} with $n = 3500$ and i.i.d.\ Gaussian errors as in~\ref{e:one}. See Table~\ref{table:comp.time}. The sequential update available for the MOSUM statistics makes the computational complexity of the proposed method very low (see Section~\ref{sec:supp:comp}).
The single-bandwidth MOSUM procedure requires less than $1$ ms. When applied with the set of bandwidths $\mc G$ chosen as described in Section~\ref{sec:multiscale} with $\vert \mc G \vert = 6$, the multiscale extension still requires less than $5$ ms on average. In contrast, NOT.pwLin and NOT.pwLinCont are much slower, with the average execution time exceeding $200$ ms. Their performance is followed by TGUW, and the dynamic programming-based CPOP takes more than $20$ seconds for its execution. This demonstrates the computational advantage of the MOSUM procedure, particularly when the datasets are long. Although not reported here, similar observations are made across different simulation scenarios.

\begin{table}[!h]
\centering
\caption{Average execution time (in milliseconds) of change point detection methods over 100 realizations generated as in~\ref{m:one} ($n = 3500$). All experiments are conducted on \textsf{R} 4.2.1 with the processor of 2100MHz 12-core i7-1260P and 16GB of memory.}
\begin{tabular}{c|c|c|c|c|c|c|c|}\cline{2-7}
 & \multirow{2}{*}{\minitab[c]{MOSUM \\ ($G=200$)}} &\multirow{2}{*}{\minitab[c]{MOSUM \\ (multiscale)}} & \multicolumn{2}{c|}{NOT} &  \multirow{2}{*}{CPOP} & \multirow{2}{*}{TGUW} \\\cline{4-5}
 & &  & pwLin & pwLin.Cont & &  \\
\hline
\multicolumn{1}{|c|}{Time (ms)} & {0.685} & {3.853} & 187.88 & 521.55 & $> 2 \times 10^4$ & 1843.37 \\
\hline
\end{tabular}
\label{table:comp.time}
\end{table}

\section{Real data analysis}
\label{sec:real}

We analyze a time series dataset described in the Introduction, which is of length $n = 9830$. See Figure~\ref{bearing1.plot}.
We apply the multiscale extension of the MOSUM procedure described in Section~\ref{sec:multiscale} with the set of bandwidths $\mc G = \{100, 200, 300, 500, 800, 1300, 2100\}$. The estimated change points are shown in the top panel of Figure~\ref{bearing1.compared}, along with the estimator of the piecewise linear signal. There are $19$ change points detected, and many of them are detected after February 16th. 
Although omitted here, there is little autocorrelation left in the residuals from the piecewise linear fit to the data, which supports using a MOSUM-based variance estimator in~\eqref{sigma_k} rather than an estimator of LRV.
\begin{figure}[!h]
\centering
\includegraphics[width=0.5\textwidth]{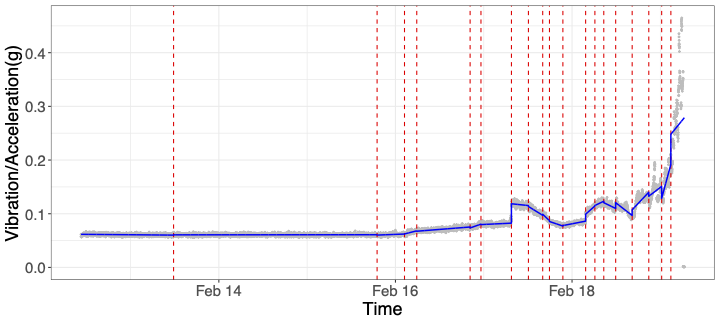}
\hspace{-3mm}
\includegraphics[width=0.5\textwidth]{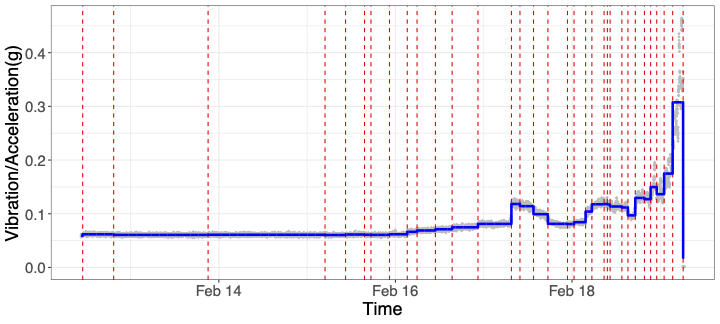}
\\
\includegraphics[width=0.5\textwidth]{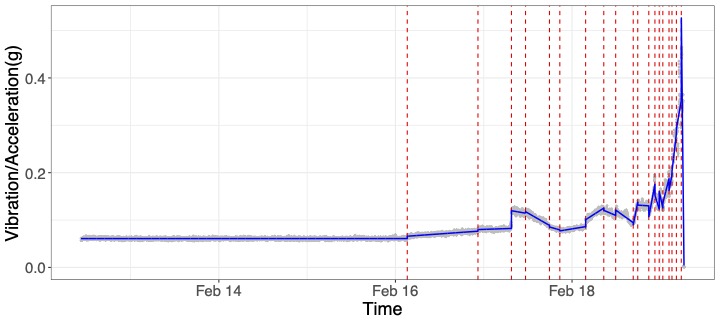}
\hspace{-3mm}
\includegraphics[width=0.5\textwidth]{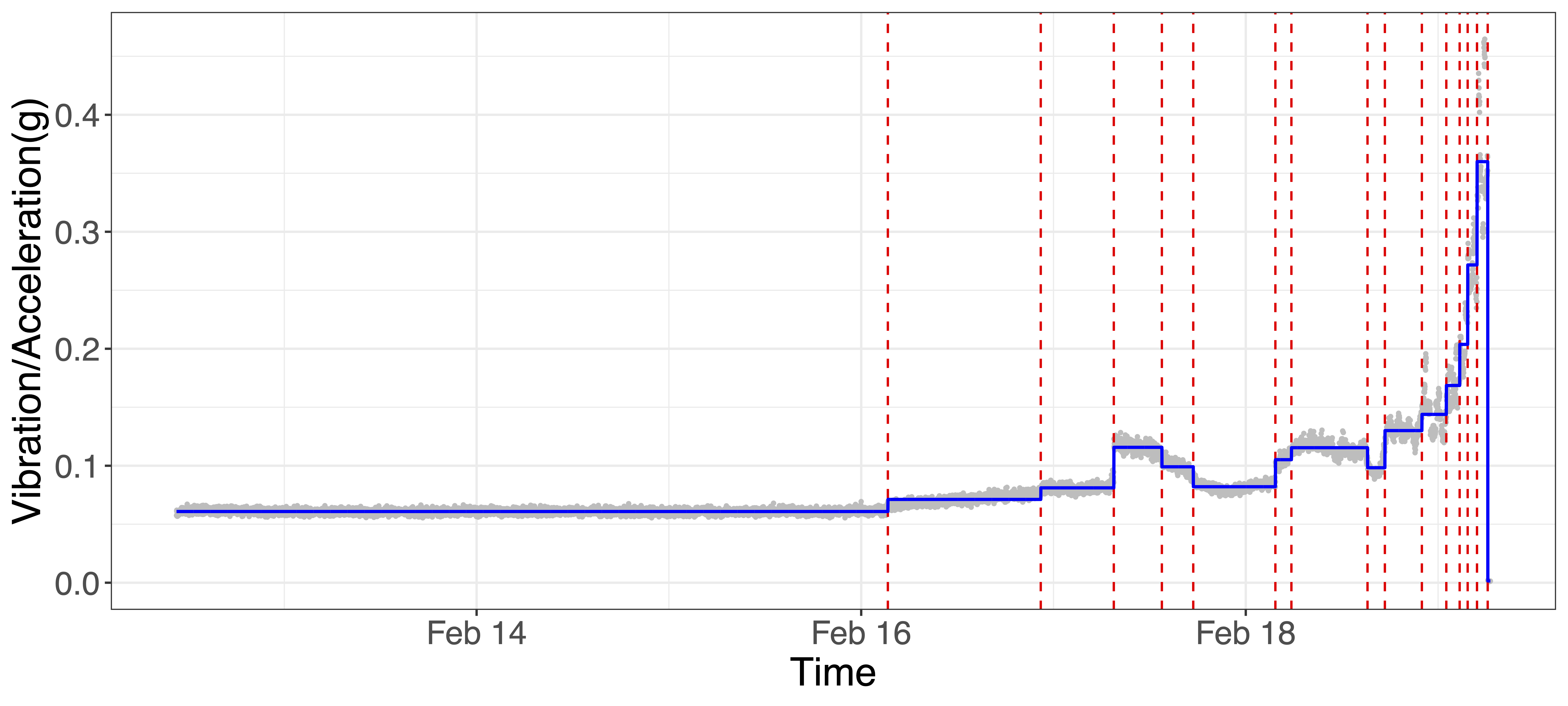}
\\
\includegraphics[width=0.5\textwidth]{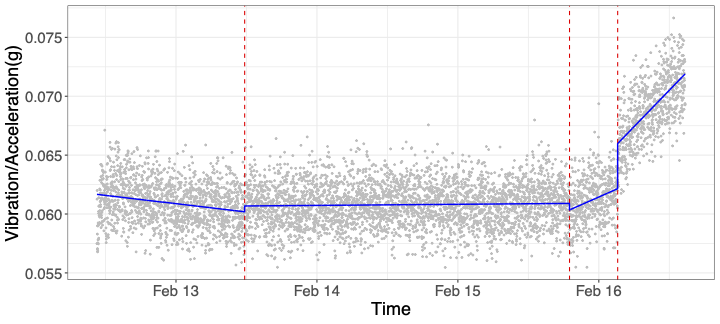}
\hspace{-3mm}
\includegraphics[width=0.5\textwidth]{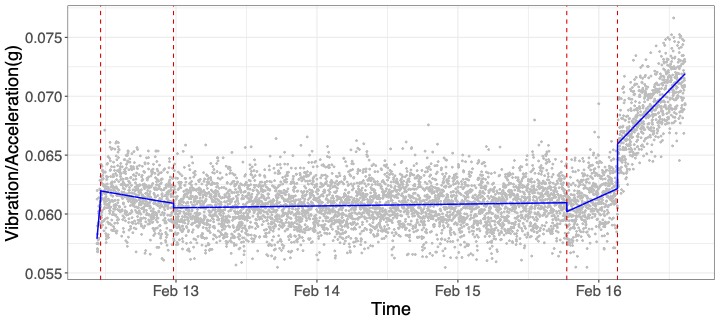}
\vspace{-7mm}
\caption{\small Bearing dataset superimposed with the change points detected by different methods (vertical broken lines) and the estimated piecewise linear signal (solid). We visualize the results from MOSUM (top left), MOSUM.pwConst (top right), NOT.pwLin (middle left) and NOT.pwConst (middle right). At bottom, change points detected from the truncated bearing dataset by the proposed MOSUM procedure (left) and NOT.pwLin (right) are displayed.}
\label{bearing1.compared}
\end{figure}


For comparison, we apply NOT.pwLin and the methods which detect change points under piecewise constant modeling, such as NOT.pwConst, the variant of NOT (see Table~\ref{sim:tab:m4}), and MOSUM.pwConst, the multiscale MOSUM procedure combined with the bottom-up merging as implemented in \cite{cho2019}; we apply the latter with $\mc G$ as the set of bandwidths. Figure~\ref{bearing1.compared} shows that MOSUM.pwConst detects more change points than our method, possibly as it tries to approximate a signal with trends using a piecewise constant signal. Both NOT.pwLin and NOT.pwConst do not detect any change point before February 16th while detecting more frequent change points post-February 17th. These differences are attributed to the variability increasing significantly towards the end of the signal. NOT-based methods use a single constant as an estimator of the noise level, which may have been chosen too large to detect the subtle changes before February 16th, while too small to prevent these methods from possibly over-fitting the anomalous behavior post-February 17th. On the other hand, by adopting the local variance estimator in~\eqref{sigma_k}, 
the proposed MOSUM procedure better captures local variability. We further verify this by considering the truncated dataset which runs up to about 3 PM on February 16th ($n = 6000$), see bottom panels of Figure~\ref{bearing1.compared}. As desired, the truncation of the data does not alter the results reported by our method whereas NOT.pwLin returns estimators previously undetected when applied to the entire dataset. This also demonstrates the applicability of the MOSUM procedure for real-time monitoring of change points.
 

\section{Concluding remarks}
\label{sec:conc}

This paper introduces a moving sum-based methodology for detecting multiple changes in both intercept and slope parameters under a piecewise linear model. We derive the asymptotic null distribution of the proposed Wald-type MOSUM test statistic, which provides a principled way of calibrating the methodology while controlling the family-wise error. We also establish its theoretical consistency in multiple change point detection and, when the additional continuity is imposed, derive the rate of localization. In doing so, we make a mild assumption on the errors, which is considerably weaker than independence and (sub-)Gaussianity assumptions found in the relevant literature. The competitiveness of the proposed methodology is further demonstrated empirically on both simulated and real datasets compared to the existing methods, where it shows promising performance thanks to the adoption of moving windows that enables efficient computation with $O(n)$ complexity. We envision that the proposed methodology and tools for theoretical analysis can be generalized to detect multiple change points under piecewise polynomial models.

%

\bibliographystyle{apalike}
\bibliography{references_mosumlin.bib}

\clearpage

\appendix

\numberwithin{equation}{section}
\numberwithin{figure}{section}
\numberwithin{table}{section}

\section{Supplements for the theoretical results}

\subsection{Estimation of long-run variance} 
\label{supp:sec:est:var}

We proposed a MOSUM-based local estimator of variance $\sigma^2$ under independence in the manuscript, see Remark~\ref{rem:est:var}. In this section, we propose a local long-run variance (LRV) estimator that extends the estimator of \cite{eichinger2018} to our setting, and establish its properties under serial dependence.

Denote by $\Gamma(h) = \Cov(\epsilon_0, \epsilon_h)$ the autocovariance (ACV) function of $\{\epsilon_i\}$ at lag $h$, such that $\tau^2 = \sigma^2 + 2 \sum_{h = 1}^\infty \Gamma(h)$.
Then we propose a MOSUM-based local ACV estimator $\wh{\Gamma}_k (h) = (\wh{\Gamma}_{k, +} (h) + \wh{\Gamma}_{k, -} (h))/2$, where
\begin{align*}
\wh{\Gamma}_{k, +} (h) &= \frac{1}{G-2} \sum_{i = k+1}^{k+G-h} \big( X_i - \bx_{i,k}^{\top} \wh{\bbeta}^{+}(k)\big) \big( X_{i+h} - \bx_{i+h, k}^{\top} \wh{\bbeta}^{+}(k)\big) \text{ \ and \ } 
\\
\wh{\Gamma}_{k,-} (h) &= \frac{1}{G-2} \sum_{i=k-G+1}^{k-h} \big( X_i - \bx_{i,k}^{\top} \wh{\bbeta}^{-}(k)\big) \big( X_{i+h} - \bx_{i+h,k}^{\top} \wh{\bbeta}^{-}(k)\big),
\end{align*}
with estimating $\tau^2$ by
$\wh{\tau}_k^2 = \wh{\Gamma}_k(0) + 2 \sum_{h=1}^{S_n} \cK(h/S_n) \wh{\Gamma}_k(h)$.
Here $\cK(\cdot)$ is a bounded kernel such as the flat-top \citep{politis1995bias} or Bartlett kernels, and $S_n$ denotes the kernel bandwidth.
If we additionally assume that $\{\epsilon_i\}$ is a sequence of i.i.d.\ random variables, it suffices to estimate the variance $\sigma^2$ by
$
\wh{\sigma}_k^2 = \wh{\Gamma}_{k}(0)
$, which becomes \eqref{sigma_k} in the manuscript. 

Theorem~\ref{thm_sigma} establishes that under weak dependence, $\wh\tau_k^2$ fulfils the first condition in~\ref{c:one}.
Some kernel-based LRV estimators are not guaranteed to be positive, in which case a standard solution is to truncate the estimator from the below, and we may similarly enforce the boundedness from the above.
Under independence, $\wh\sigma_k^2$ satisfies both~\ref{c:one} and~\ref{c:two} without such truncation.

\begin{thm} 
\label{thm_sigma}
Assume that $\{\epsilon_i\}_{i \in \mathbb{Z}}$ satisfies~\ref{a:two} and $\E(\vert \epsilon_1 \vert^{4}) < \infty$, and the bandwidth $G$ satisfies~\eqref{eq:cond:G}.
Furthermore, assume that 
$\sup_{h \in \bbZ} \sum_{k \in \bbZ} \sum_{\ell \in \bbZ} |\omega(h, k, \ell)| < \infty$, where $\omega(h, k, \ell) = \Cov(\epsilon_1 \epsilon_{1+h}, \epsilon_{1+k} \epsilon_{1+\ell} )$.
\begin{enumerate}[label = (\roman*)]
\item \label{thm_sigma:tau} If $0 \le \cK(x) \le C$ for some constant $C \in (0, \infty)$ and $G^{-2} S_n^2 n = O(1)$, uniformly over all $k$ satisfying $\min_{1 \le j \le J_n} \vert k - k_j \vert \ge G$, we have
\begin{align*}
\big| \wh{\tau}_{k}^{2} - \tau^{2} \big| = O_P \left( \frac{\sqrt{n} S_n}{G} + \sum_{h \in \bbZ} \left| \cK \left( \frac{h}{S_n} \right) - 1 \right| {|\Gamma(h)|} \right).
\end{align*}

\item \label{thm_sigma:sigma} Additionally, suppose that $\{\epsilon_i\}_{i \in \mathbb{Z}}$ is a sequence of i.i.d.\ random variables.
\begin{enumerate}[label = (\alph*)]
\item \label{thm_sigma:one} Uniformly over all $k$ satisfying $\min_{1 \le j \le J_n} \vert k - k_j \vert \ge G$, we have
\begin{align*}
\l\vert \wh{\sigma}_k^{2} - \sigma^{2} \r\vert = o_P\left( \frac{1}{\log(n/G)} \right)
\quad \text{and} \quad 
\wh{\sigma}_k^{- 2}  = O_P(1).
\end{align*}
\item \label{thm_sigma:two} $\max_{G \le k \le n - G} \wh{\sigma}_k^{2} = O_P(1)$, provided that $\max_{1 \le j \le J_n} \|\bDelta_j\| = O(1)$.
\end{enumerate}

\end{enumerate}
\end{thm}

Proof of the theorem is in Section \ref{sec:proof}. 

\subsection{When $f_i$ is piecewise linear and discontinuous}
\label{sec:discont}

In this section, we explore the behavior of the Wald-type MOSUM statistic $W_{k, n}(G)$ in~\eqref{waldstat} when $f_i$ is not necessarily continuous. 

For $k \in \{k_j - G + 1, \ldots, k_j +  G\}$, we can approximate
$\E[\wh{\bbeta}^{+}(k) - \wh{\bbeta}^{-}(k)]$ by
\begin{align*}
\b\delta_j(\kappa) = (1 - \vert \kappa \vert) 
\bmx 1 - 3 \vert \kappa \vert &  \kappa (1 - \vert \kappa \vert) \\ 
-6 \kappa & -2\vert \kappa \vert^2 + \vert \kappa \vert +1 \emx \b\Delta_j, 
\text{ \ where \ } \kappa = \frac{k-k_j}{G}.
\end{align*}
See Lemma~\ref{lemA.4} in Section~\ref{sec:proof}. For $\kappa \in (-1, 1]$, we define
\begin{align*}
g_{j}(\kappa) =  \sqrt{ \l(\b\delta_j(\kappa) \r)^\top \bmx 1 & 0 \\ 0 & \frac{1}{3} \emx \b\delta_j(\kappa) },
\end{align*}
where the weight matrix is proportional to $\b\Sigma^{-1}$.
When $f_i$ is piecewise linear and continuous such that $\Delta^{(0)}_j = 0$, it further simplifies to
$g_j(\kappa) = (1 - \vert \kappa \vert)^2 \sqrt{ \vert \kappa \vert^2 + (2 \vert \kappa \vert + 1)^2/3} \vert \Delta^{(1)}_j \vert$,
which has its maximum attained at $\kappa = 0$.
In short, when $f_i$ is continuous, the `signal' part of $W_{k, n}(G)$ forms local maxima at the true change points under~\ref{b:one}.

\begin{figure}[!htb]
\centering
\includegraphics[width=1.0\textwidth]{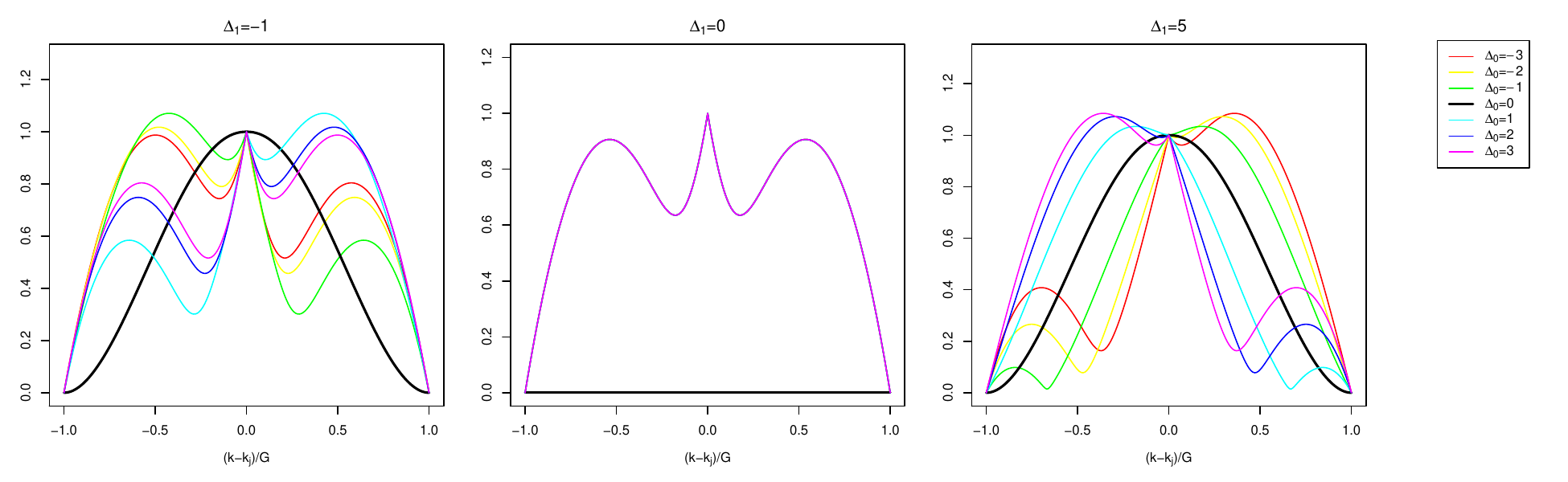}
\vspace{-8mm}
\caption{Behavior of $g_j(\kappa)/g_j(0), \, \kappa \in (-1, 1]$, when $\Delta^{(1)}_j = -1$ (left), $\Delta^{(1)}_j = 0$ (middle) and $\Delta^{(1)}_j = 5$ (right)
and $\Delta^{(0)}_j\in \{-3, -2, \ldots, 3\}$.
In each plot, the curve corresponding to $\Delta^{(0)}_j= 0$ is highlighted in bold.
We adopt the convention $0/0 = 0$ in the case where $g_{j} (\kappa) \equiv 0$ for all $\kappa$.}
\label{fig.theoretical1}
\end{figure}

The behavior of $g_j(\kappa)$ becomes more complex when $f_i$ contains discontinuities as illustrated in Figure~\ref{fig.theoretical1}, which plots the ratio $g_j(\kappa)/g_j(0)$ for varying combinations of $\b\Delta_j = (\Delta_j^{(0)}, \Delta^{(1)}_j)$. 
When either $\Delta^{(0)}_j= 0$ or $\Delta^{(1)}_j = 0$, the ratio is maximized at $\kappa = 0$ and in the latter case, $g_j(\kappa)$ is unimodal. However, this no longer holds when $\Delta^{(0)}_j\ne 0$ and in fact, $g_j(\kappa)$ has its maximum anywhere within $\kappa \in (-1, 1]$.
Due to this, in the absence of continuity, we cannot derive the refined rate of localization for the change point estimators as in Theorem~\ref{5.1}~\ref{thm:est:two}, although the proposed MOSUM procedure achieves consistency in detecting the presence of multiple change points regardless of whether $\Delta^{(0)}_j = 0$ or not.
To the best of our knowledge, this characteristic of the Wald-type MOSUM statistic under piecewise linearity has previously been unobserved. We envision that once a change point is detected, its location can further be refined as in \cite{chen2021jump} or \cite{yu2022localising} where in particular, the former proposes an estimator that adapts to the unknown regime.

In practice, we benefit from using the proposed local variance estimator in~\eqref{sigma_k}, see Figure~\ref{fig.theoretical2}.
Here, we generate $X_i$ under~\eqref{eq:model} with $n = 3500$, $\Delta t = 0.01$ and $\epsilon_i \sim_{\text{i.i.d}} \mc N(0, 0.1^2)$; for the purpose of illustration, we choose the small error variance so that the data sequence is close to being noiseless. The underlying signal $f_i$ is piecewise linear and undergoes three discontinuous jumps at $k_1 = 1000$, $k_2 = 2000$ and $k_3 = 2500$. The middle panel plots $W_{k,n}(G)$ obtained with $G = 200$ and the local variance estimator~$\wh\sigma_k$, and the last panel displays $\wh\sigma_k/\sigma \cdot W_{k,n}(G)$ with the same $G$, i.e.,\ it is normalized using the true $\sigma$ rather than its estimator. We observe that the former attains clear local maxima at the true change points, while the latter displays the behavior described in Figure~\ref{fig.theoretical1}. This discrepancy is thanks to the upward bias in $\wh\sigma_k^2$ for $k \in \{k_j - G + 1, \ldots, k_j + G\} \setminus \{k_j\}$ due to the presence of changes, while it approximates the true variance well at $k = k_j$. This phenomenon, also observed in \cite{eichinger2018}, leads to the desirable result of 
preventing the spurious peaks in $g_j(\kappa)$ held at $\kappa \ne 0$ from appearing as local maxima in $W_{k, n}(G)$ at $k \ne k_j$ empirically.

\begin{figure}[h!t!]
\centering
\includegraphics[width=1\textwidth]{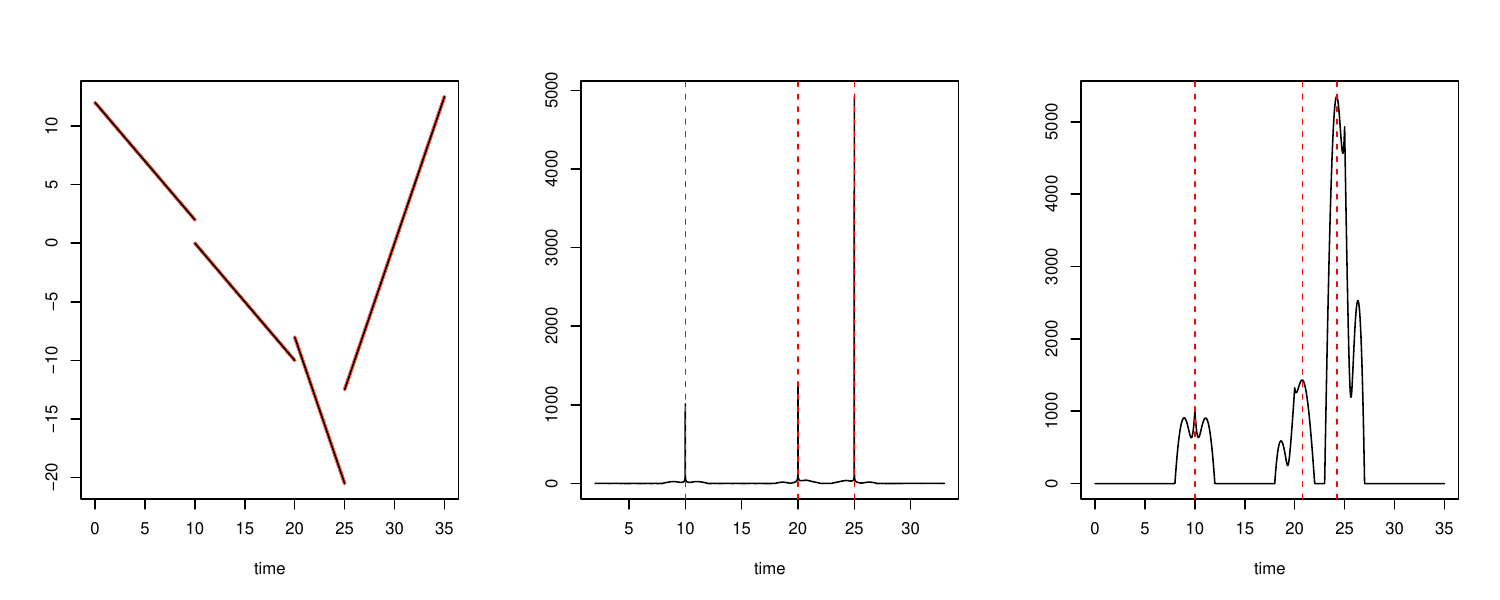}
\vspace{-7mm}
\caption{Left: $X_i$ (red) and $f_i$ (black). Middle: $W_{k, n}(G)$ with $G = 200$ and the local variance estimator~$\wh\sigma_k$. Right: $W_{k, n}(G)$ with $\sigma$ used in place of $\wh\sigma_k$.}
\label{fig.theoretical2}
\end{figure}

\section{Practical implementation}
\label{sec:supp:practical}

\subsection{Choice of $H$ for the critical value}
\label{sec:supp:critical}

Based on Theorem~\ref{thm3.1}, we select the critical value $C_n(G, \alpha)$ as $C_n(G, \alpha) = a_{G}^{-1}\big(b_{G} - \log(- \log (1-\alpha)/2)\big)$ where $a_G = \sqrt{2\log(n/G)}$ and $b_G = 2 \log(n/G) + \log\log(n/G) + \log(H)$. The constant $H$ is found via following simulations. 

Randomly generating a sequence of independent random variables 
$\{X_i\}_{i = 1}^n$ from $\mc N(0, 1)$,
we compare the empirical median of $a_G W_n(G) - 2\log(n/G) - \log\log(n/G)$
(setting $\wh\tau_k = 1$) with the median of $\Gamma_2$ to approximate $\log(H)$.
Table~\ref{logH} lists the approximated values of $\log(H)$ with varying $(n, G)$.
Although not reported here, we obtain similar results from the data simulated as $X_i \sim \cN(a_0 + a_1 i, 1)$ with varying $(a_0, a_1)$.
Based on this, in all the numerical experiments, we use $\log(H) \approx 0.7284$.
Note that the variation in $\log(H)$ with respect to $(n, G)$ is much smaller than that in $2\log(n/G)$, and hence choice of $\log(H)$ tends not to affect the performance much, see Table~\ref{table:logGn}.

\begin{table}[!htb]
\centering
\caption{$\log(H)$ values obtained via simulation with varying $(n, G)$.}
\label{logH}
{\small
\begin{tabular}{ccc ccc}
\toprule
$n$ & $G$ & $\log(H)$ & $n$ & $G$ & $\log(H)$   \\
\cmidrule(lr){1-3} \cmidrule(lr){4-6}
$10^5$ &     500 & 0.6375 & $10^6$ &   5000 & 0.7433 \\
$10^5$ &   1000 & 0.6544 & $10^6$ & 10000 & 0.7284 \\
$10^5$ &   2000 & 0.6383 & $10^6$ & 20000 & 0.7099 \\
\bottomrule
\end{tabular}}
\end{table}

\begin{table}[!htb]
\centering
\caption{$2 \log (n/G)$ values obtained via simulation with varying $(n, G)$.}
{\small \begin{tabular}{ c  c c  c  c  c }
\toprule
$n$ & $G$ & $2\log(n/G)$ & $n$ & $G$ & $2\log(n/G)$   \\
\cmidrule(lr){1-3} \cmidrule(lr){4-6}
$3500$ &   100 & 7.1107 & $500$ &  20 & 6.4378 \\
$3500$ &   250 & 5.2781 & $500$ &  50 & 4.6052 \\
$3500$ &   400 & 4.3381 & $500$ &  80 & 3.6652 \\
\bottomrule
\end{tabular}}
\label{table:logGn}
\end{table}

\subsection{Efficient computation}
\label{sec:supp:comp}

The computational complexity of the proposed MOSUM procedure is $O(n)$.
To see this, note that 
the coefficient estimators $\wh{\b\beta}^{\pm}(k)$ can be updated as
\begin{align*}
\wh{\beta}_1^{\pm}(k + 1) =& \,  \wh{\beta}_1^{\pm}(k) + B^{\pm}(k), 
\quad \text{where} \quad A = \sum_{i = 1}^G \l(i - \frac{G + 1}{2}\r)^2 \Delta t,
\\
B^+(k) =&  \frac{1}{A} \l( \frac{G - 1}{2} X_{k + G + 1} - \sum_{i = 2}^G X_{k + i} + \frac{G - 1}{2} X_{k + 1} \r)  \quad \text{and}
\\
B^-(k) =&  \frac{1}{A} \l( \frac{G - 1}{2} X_{k + 1} - \sum_{i = 2}^G X_{k - G + i} + \frac{G - 1}{2} X_{k - G + 1} \r), \quad \text{and}
\\
\wh{\beta}_0^+(k + 1) - & \wh{\beta}_0^-(k + 1)
 = \wh{\beta}_0^+(k) -  \wh{\beta}_0^-(k)
 \\
& + \frac{1}{G}  (X_{k + G + 1} - 2 X_{k + 1} + X_{k - G + 1})
- \l(\frac{G+1}{2} B^+(k) + \frac{G - 1}{2} B^-(k) \r) \Delta t,
\end{align*}
and the computation of $B^{\pm}(k)$ for all $k$ is of complexity $O(n)$.
Similarly, noting that the local variance estimator in~\eqref{sigma_k} can be written as $\wh\sigma_k^2 = (\wh\sigma_{k, +}^2 + \wh\sigma_{k, -}^2)/2$ with
\begin{align*}
\wh\sigma_{k, +}^2 & = \frac{1}{G - 2} \l( \sum_{i = 1}^G X_{k + i}^2 - \frac{1}{G} \l(\sum_{i = 1}^G X_{k + i} \r)^2 - (\wh\beta_1^+(k))^2A \Delta t \r) \quad \text{and} \\
\wh\sigma_{k, -}^2 & = \frac{1}{G - 2} \l( \sum_{i = 1}^G X_{k - G + i}^2 - \frac{1}{G} \l(\sum_{i = 1}^G X_{k  - G + i} \r)^2 - (\wh\beta_1^-(k))^2 A \Delta t \r), 
\end{align*}
it can also be updated sequentially.

\section{Supplements for the simulation studies} 
\label{sec:comp:sim}

This section complements Section~\ref{sec:MOSUM_simul} in the main text by providing the complete descriptions of the data generating processes and providing additional simulation results. 

\subsection{Settings}
\label{sec:comp:sim:models}

The piecewise linear signal $f_i$ is generated as follows. Throughout, we set $t_i = 0.01 i$.
\begin{enumerate}[label = (M\arabic*)]
\setcounter{enumi}{-1}
\item \label{m:zero} No change point ($J = 0$) with $f_i = \beta t_i + \epsilon_i$, $1 \le i \le n = 3500$ where $\beta \sim \mc N(-1, 0.2^2)$ is randomly drawn for each realization. 

\item \label{m:one} Piecewise linear and continuous with $J = 3$ and
\begin{align} 
f_i &= \left\{ \begin{array}{ll} \beta_1 (t_i - 10) + 10 & ( 0 < i \le k_1 = 1000) 
\\  
\beta_2 (t_i-10) & ( 1000 < i\le k_2 = 2000) 
\\ 
10(1+\beta_2)+ \beta_3 (t_i-20) & (2000 < i \le k_3 = 2500 ) 
\\ 
10(1+\beta_2) + 5 \beta_3 + \beta_4 (t_i - 25) & ( 2500 < i \le n = 3500)
\end{array} \right. \nn
\text{ \ when $n = 3500$,}
\\
f_i &= \left\{ \begin{array}{ll} 10\beta_1 (t_i - 1) + 10 & ( 0 < i \le k_1 = 100) 
\\  
10\beta_2 (t_i-1) & ( 100 < i\le k_2 = 200) 
\\ 
10(1+\beta_2)+ \frac{10}{3}\beta_3 (t_i-2) & (200 < i \le k_3 = 350 ) 
\\ 
10(1+\beta_2) + 5 \beta_3 + \frac{20}{3}\beta_4 (t_i - 3.5) & ( 350 < i \le n = 500)
\end{array} \right. 
\text{ \ when $n = 500$.} \nn
\end{align}
The vector $\b\beta = (\beta_i, \beta_2, \beta_3, \beta_4)^\top \sim \cN (\b\mu, \sigma_{\beta}^{2} \mbf I)$ is randomly drawn for each realization
with $\b\mu = (-1, -1, -2.5, 2.5)^{\top}$ and $\sigma_{\beta} = 0.2$.
Parameters are chosen differently when $n = 500$ and $n = 3500$ to ensure that $f_i$ takes the same values at the three change points regardless of $n$ for given $\b\beta$.
At $i = k_1$, the change in the intercept brings a discontinuity to the signal,
at $i = k_2$, the slope changes as well as there being a discontinuous jump
and at $i = k_3$, the slope changes while the signal remains continuous.
A realization is shown in the top left panel of Figure~\ref{fig.simul1}. 

\item \label{m:two} Piecewise linear and continuous with $J = 3$ and
\begin{align} 
f_i &= \left\{ \begin{array}{ll} 
\beta_1 (t_i - 10) & ( 0 < i \le k_1 = 1000) \\  
\beta_2 (t_i - 10) & ( 1000 < i \le k_2 = 2000) \\ 
10\beta_2 + \beta_3 (t_i-20) & (2000 < i \le k_3 = 2500 ) \\ 
10\beta_2 + 5 \beta_3 + \beta_4 (t_i-25) & ( 2500 < i \le n = 3500)
\end{array} \right.  \nn
\text{ \ when $n = 3500$,}
\\
f_i &= \left\{ \begin{array}{ll} 10\beta_1 (t_i - 1) & ( 0 < i \le k_1 = 100) 
\\  
10\beta_2 (t_i-1) & ( 100 < i\le k_2 = 200) 
\\ 
10\beta_2+ \frac{10}{3}\beta_3 (t_i-2) & (200 < i \le k_3 = 350 ) 
\\ 
10\beta_2 + 5 \beta_3 + \frac{20}{3}\beta_4 (t_i - 3.5) & ( 350 < i \le n = 500)
\end{array} \right. 
\text{ \ when $n = 500$.}
\nn 
\end{align}
We generate $(\beta_1, \beta_2, \beta_3, \beta_4)^{\top}$ as in \ref{m:one}.
The top right panel of Figure~\ref{fig.simul1} displays a realization.

\item \label{m:three} Piecewise linear with frequent changes ($J = 6$) and $n = 2500$, where
\begin{align*}
f_i = \left\{ \begin{array}{ll} 
    \beta_1 (t_i-5) & ( 0 < t \le k_1 = 500) \\  
    \beta_2 (t_i-5) - 10 & ( 500 < t \le k_2 = 800) \\ 
    3\beta_2+ \beta_3 (t_i-12) & (800 < t \le k_3 = 1200 ) \\ 
    5 & (1200 < t_i \le k_4 = 1300 ) \\ 
    3\beta_2 + 4\beta_3 + \beta_4 (t_i-12) & ( 1300 < t \le k_5 = 1700) \\ 
    3\beta_2 + 4\beta_3 + 5 \beta_4 & ( 1700 < t \le k_6 = 2100) \\ 
    3\beta_2 + 4\beta_3 + 5 \beta_4 + \beta_5 (t_i-21) & ( 2100 < t \le n = 2500).
\end{array} \right.  
\end{align*}
Here, $\b\beta = (\beta_i, \beta_2, \beta_3, \beta_4, \beta_5)^\top \sim \cN (\b\mu, \sigma_{\beta}^{2} \mbf I)$ is randomly drawn for each realization
with $\b\mu = (-1, -1, -2.5, 2.5, -2.5)^{\top}$ and $\sigma_{\beta} = 0.2$.
The bottom left panel of Figure~\ref{fig.simul1} displays one realization from such a model.

\item \label{m:four} Piecewise constant with $J = 3$ and
\begin{align} 
f_i &= \left\{ \begin{array}{ll} 
\beta_1  & ( 0 < i \le k_1 = 1000) \\  
\beta_2 & ( 1000 < i \le k_2 = 2000) \\ 
\beta_3 & (2000 < i \le k_3 = 2500 ) \\  
\beta_4   & ( 2500 < i \le n = 3500)
\end{array} \right. 
\text{ \ when $n = 3500$,} 
\nn 
\\
f_i &= \left\{ \begin{array}{ll} 
3\beta_1  & ( 0 < i \le k_1 = 100) \\  
3\beta_2 & ( 100 < i \le k_2 = 200) \\ 
3\beta_3 & (200 < i \le k_3 = 350 ) \\  
3\beta_4   & ( 350 < i \le n = 500)
\end{array} \right. 
\text{ \ when $n = 500$.}
\nn
\end{align}
We draw $\b\beta = (\beta_1, \beta_2, \beta_3, \beta_4)^\top$ from
$\mc N(\b\mu, \sigma_\beta^2 \mbf I)$ with 
$\b\mu = (-2, 2, -5, 5)^\top$ and $\sigma_\beta = 0.2$. 
The bottom right panel of Figure~\ref{fig.simul1} shows that $f_i$ is piecewise constant.
\end{enumerate}


\begin{figure}[!htb]
\centering
\begin{tabular}{cc}
\includegraphics[width=0.45\textwidth]{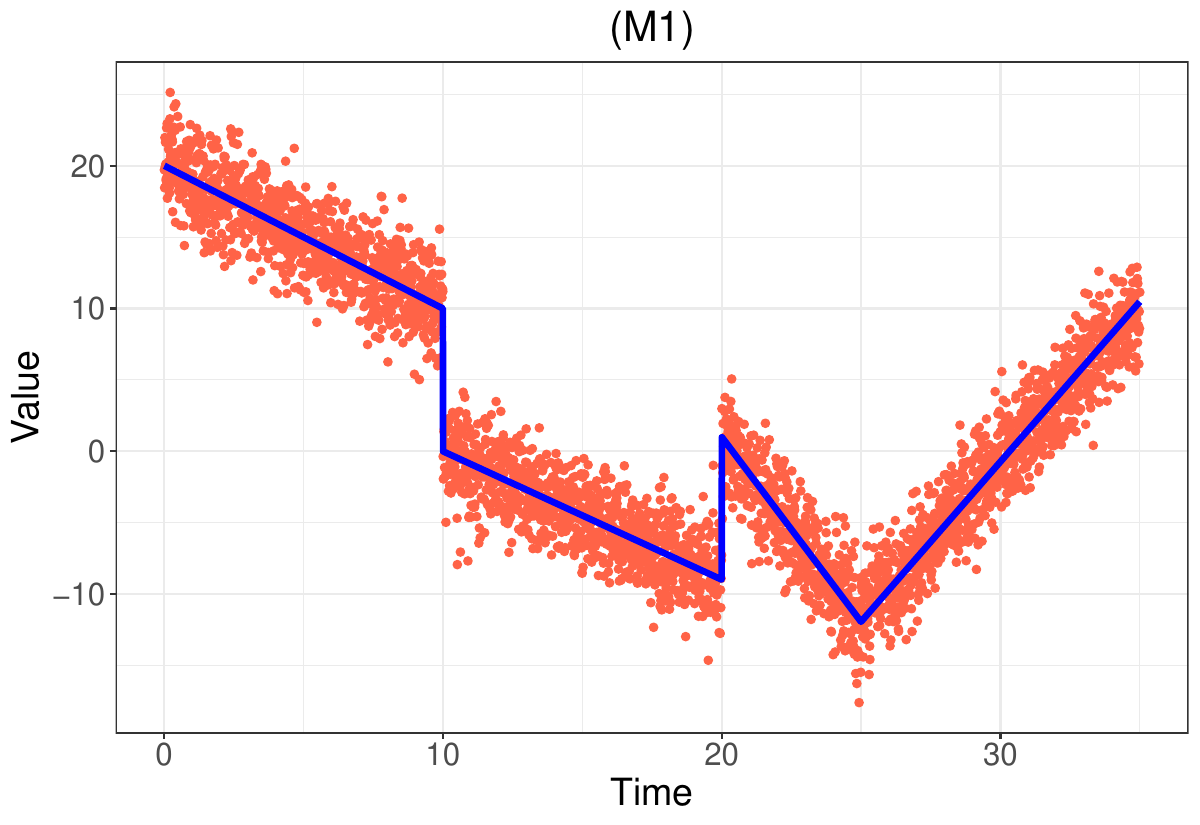} &
\includegraphics[width=0.45\textwidth]{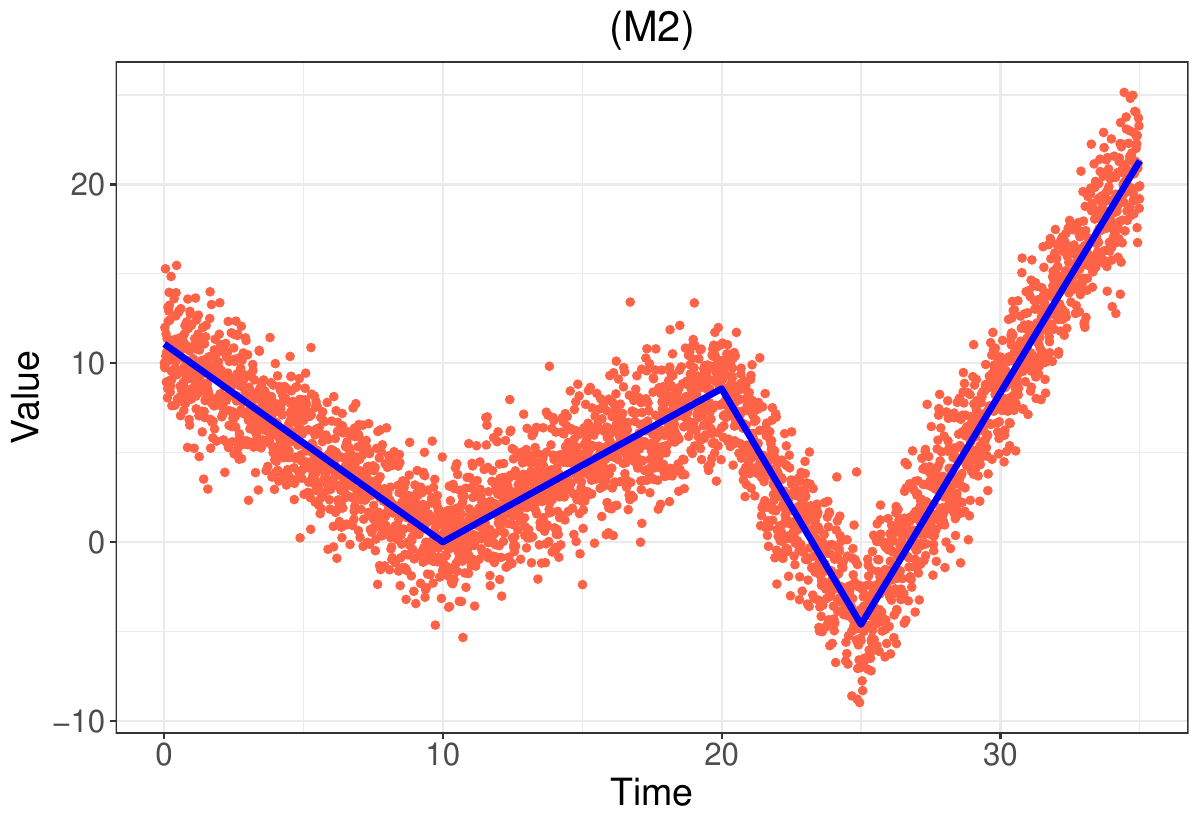} \\
\includegraphics[width=0.45\textwidth]{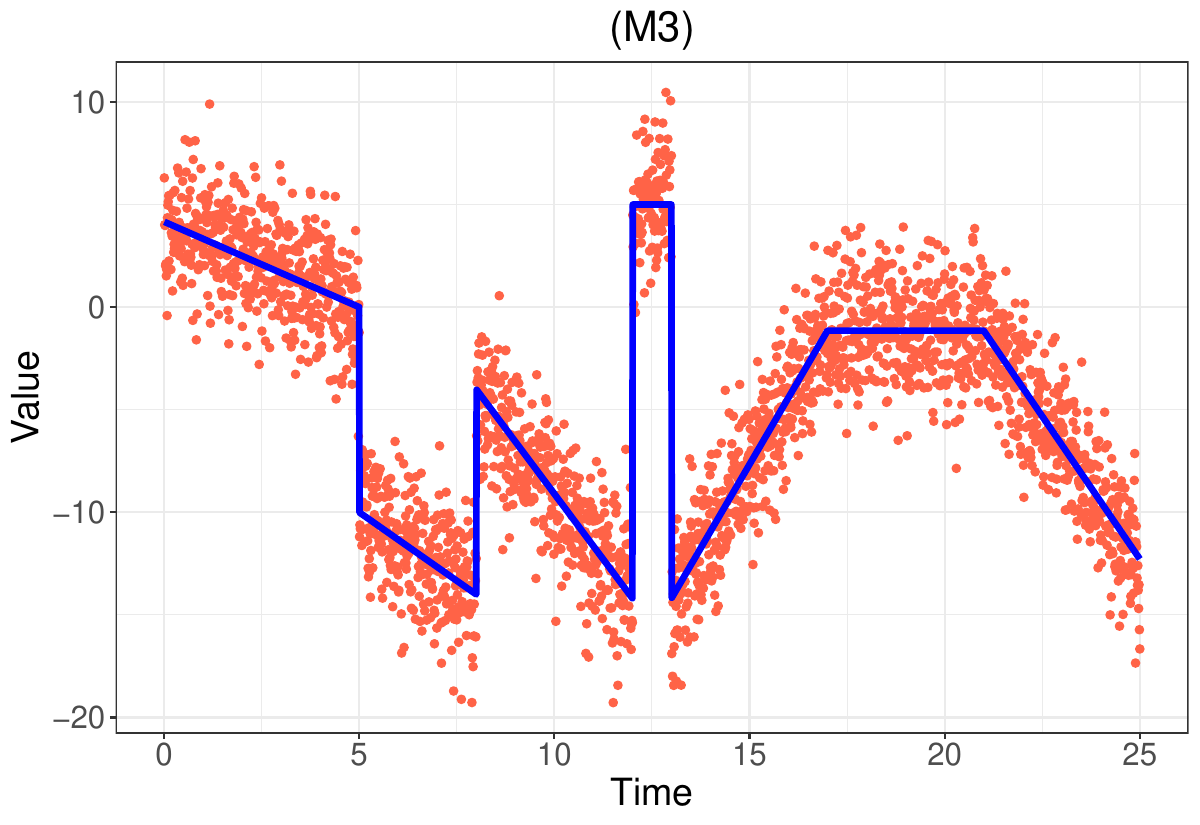} &
\includegraphics[width=0.45\textwidth]{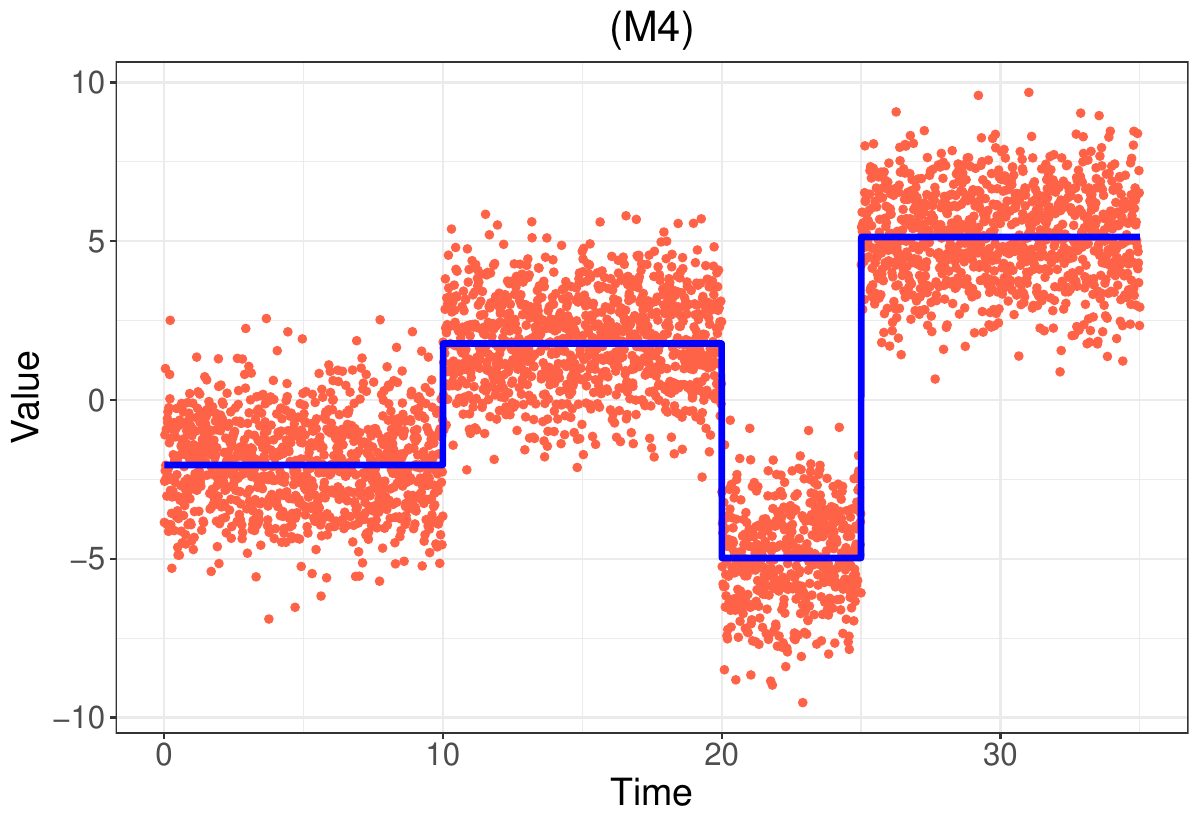}
\end{tabular}
\vspace{-3mm}
\caption{A realization from \ref{m:one}, \ref{m:two}, \ref{m:three} and \ref{m:four} with Gaussian errors generated as in~\ref{e:one} with $\sigma = 2$. The true signal $f_i$ is super-imposed (solid line).}
\label{fig.simul1}
\end{figure}

For the generation of $\{\epsilon_i\}_{i = 1}^n$, we consider 
\begin{enumerate}[label = (E\arabic*)]
\item \label{e:one} A sequence of i.i.d.\ random variables with a normal distribution.
\item \label{e:two} A sequence of i.i.d.\ random variables with a scaled $t_5$-distribution.
\item \label{e:three} A sequence of i.i.d.\ random variables with a scaled Laplace distribution whose distribution takes the form with some $s > 0$:
\begin{align} 
f(x; s) = \frac{1}{4s} \exp \left( -\frac{|x|}{2s} \right). \nn
\end{align}
\item \label{e:four} An AR($1$) process $\epsilon_i = \rho \epsilon_{i - 1} +  \sqrt{1 - \rho^2} \sigma Z_i$ with $\rho \in \{ 0.3, 0.7 \}$ and $Z_i \sim_{\text{i.i.d.}} \mc N(0, 1)$.
\end{enumerate}
In all above, we vary $\Var(\epsilon_i) = \sigma^2$ with $\sigma \in \{ 0.5, 1, 1.5, 2 \}$.


\subsection{Results}
\label{sec:comp:sim:res}

Tables~\ref{supp:simul1.table.bic}--\ref{supp:simul3.table.bic} summarize the simulation results obtained from $1000$ realizations for each setting.
As observed in Section~\ref{sec:time} in the main text, CPOP tends to be much slower with worse performance than the rest of the methods, and thus is included in Tables~\ref{supp:simul1.table.bic}, \ref{supp:simul2.table.bic} and~\ref{simul2-laplace-t5.table.bic} only.
See Table~\ref{sim:tab:m0} in the main text for the results under~\ref{m:zero}.

Firstly, we consider the results obtained under independence and piecewise linearity, see Tables~\ref{supp:simul1.table.bic}, \ref{simul1-laplace-t5.table.bic}, \ref{supp:simul2.table.bic}, \ref{simul2-laplace-t5.table.bic} and~\ref{simul9.table.bic}.
Overall, the proposed MOSUM procedure estimates the total number and the locations of the change points with high accuracy, and it shows comparable performance as NOT.pwLin or even outperforms latter by a small margin.
In particular, MOSUM tends to attain better localization accuracy measured by {\sf MAXscore1} and {\sf MAXscore2}, and the good performance does not depend on the length of the signals, the error generating processes (Gaussianity under \ref{e:one} or heavy-tailed under \ref{e:two}--\ref{e:three}) or the types of changes (\ref{m:one}, \ref{m:two} or \ref{m:three}).
One exception is when $\sigma = 2$ and the signal is short ($n = 500$) where we observe slight deterioration in the detection accuracy of MOSUM measured by \textsf{COUNTscore}, see Table~\ref{supp:simul1.table.bic}.
TGUW is generally outperformed by both MOSUM and NOT.pwLin in all scenarios under consideration and in particular, it is prone to return many false positives under non-Gaussianity, as evidenced by large \textsf{COUNTscore} and \textsf{MAXscore2} (Tables~\ref{simul1-laplace-t5.table.bic} and~\ref{simul2-laplace-t5.table.bic}).

NOT.pwLinCont and CPOP pre-suppose that the signals are piecewise linear and continuous and as such, in the presence of discontinuities as in~\ref{m:one}, they tend to over-estimate the number of change points by detecting spurious estimators close to the true change points (as evidenced by the small \textsf{MAXscore1} values), in order to approximate the discontinuous signal by introducing more segments, see Figure~\ref{fig_contapprox} for an example.
For piecewise linear and continuous signals under \ref{m:two}, as expected, NOT.pwLinCont and CPOP perform well under~\ref{e:one}; in particular, the former detects the correct number of change points in almost all realizations, see Table~\ref{supp:simul2.table.bic}. 
CPOP has its performance deteriorate greatly in the presence of heavy tails, with larger than $3$ average \textsf{COUNTscore}.
The MOSUM procedure performs comparably to NOT.pwLinCont in terms of detection accuracy and even outperforms CPOP in this respect (\textsf{COUNTscore}), but its estimation accuracy is slightly worse than these methods.
This may be explained by that our method does not use the extra information that the signal is continuous, i.e.\ no constraint is imposed when estimating $\wh{\bbeta}^{\pm}(k)$. 
Compared to~\ref{m:one}, we remark that most methods attain worse localization performance under~\ref{m:two} in the presence of heavy-tailedness noise, see Table \ref{simul2-laplace-t5.table.bic}.

\begin{figure}[!htb]
\centering
\includegraphics[width=0.48\textwidth]{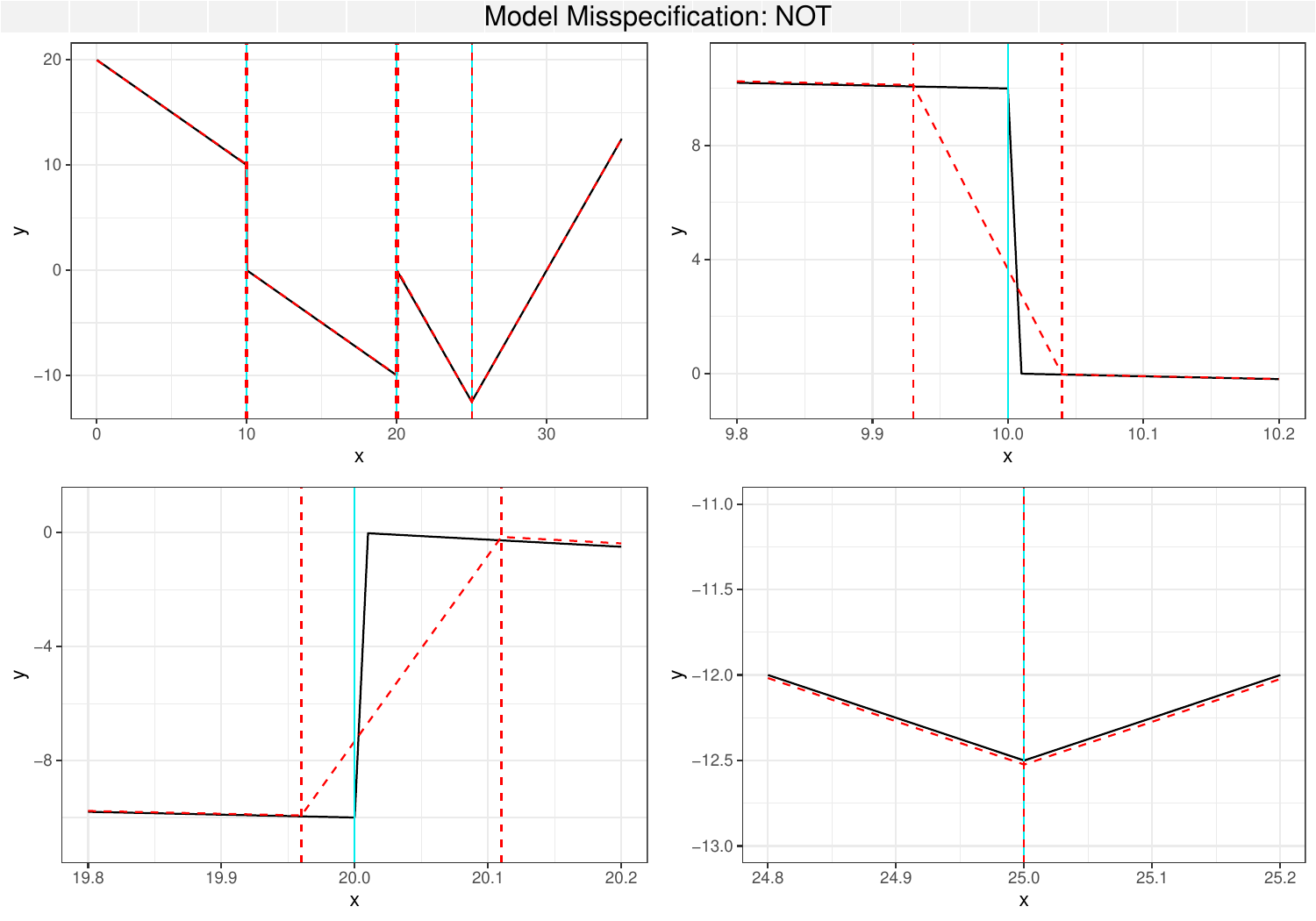}
\includegraphics[width=0.48\textwidth]{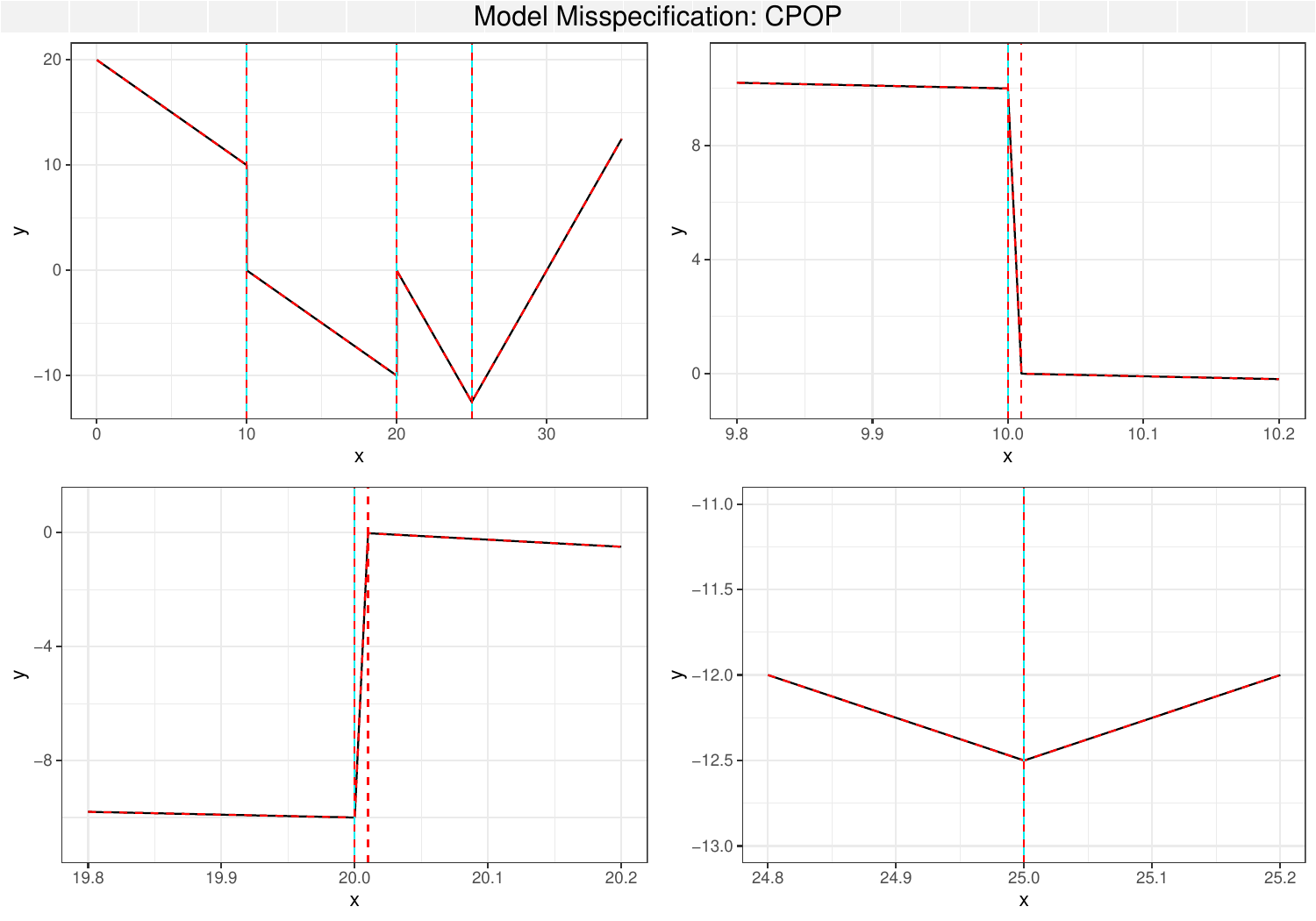}
\vspace{-2mm}
\caption{\ref{m:one}: Left: NOT.pwLinCont approximates the discontinuous signal $f_i$ (top left) around $k_1$ and $k_2$ by introducing additional change point estimators (top right and bottom left) while estimating the continuous kink at $k_3$ well (bottom right). True signal and change point locations are denoted by solid lines, while the estimated signal and change point locations are given by dashed lines. The right panel similarly illustrates the performance of CPOP.}
\label{fig_contapprox}
\end{figure}

For the case with the serially correlated errors under~\ref{e:four}, see Table~\ref{simul1.table.ar1} and also Table~\ref{sim:tab:e4} in the main text and its description.

Under~\ref{m:four}, the signal is piecewise constant with $\alpha_{1, j} = 0$ in~\eqref{eq:model}.
Hence, we additionally consider NOT.pwConst (`piecewise constant') as proposed in \cite{baranowski2019} besides MOSUM, NOT.pwLin and TGUW, see Table~\ref{supp:simul3.table.bic} for the summary of the results. 
The MOSUM procedure shows comparable or better performance than NOT.pwConst regardless of $n$ when the noise level is small ($\sigma \le 1.5$), without pre-supposing that the signal is piecewise constant.
Its performance is relatively worse when $\sigma = 2$ and $n = 3500$.
Upon close inspection, most inaccuracy stems from that MOSUM occasionally approximates the piecewise constant signal with two constant pieces around $k_1$, via a piecewise linear fit with three linear segments, see Figure~\ref{fig:M4:MOSUM}) for an example of such an instance.

\begin{figure}[!htb]
\centering
\includegraphics[width=0.7\textwidth]{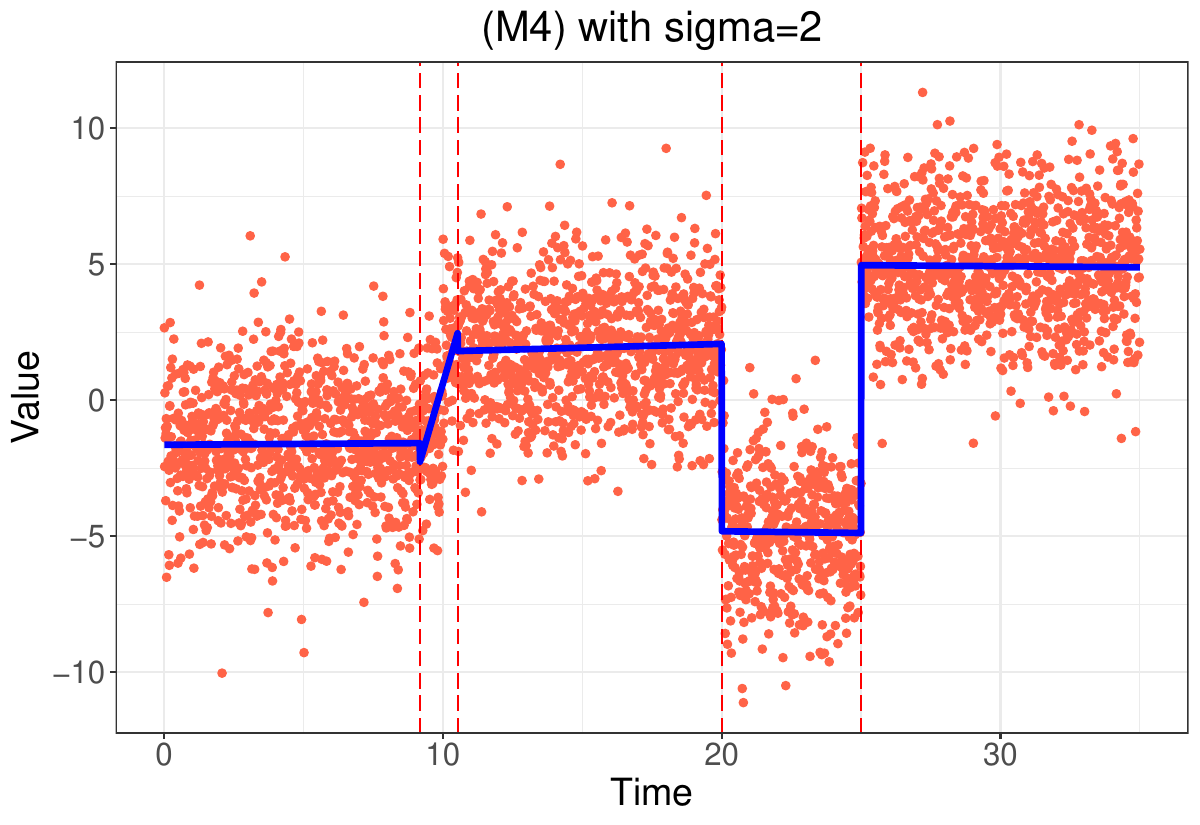}
\vspace{-2mm}
\caption{One realization under~\ref{m:four} with $n = 3500$ with Gaussian errors as in~\eqref{e:one} when $\sigma = 2$. Vertical broken lines denote the estimated location of changes and the solid line represents the fitted piecewise linear signal based on the change point estimators returned by MOSUM.}
\label{fig:M4:MOSUM}
\end{figure}


\begin{table}[!htb]
\centering
\caption{\ref{m:one}: Results from MOSUM, NOT.pwLin, NOT.pwLinCont, CPOP and TGUW when the errors are generated as in \ref{e:one} with $\sigma \in \{0.5, 1, 1.5, 2\}$. We report the average and standard error (in parentheses) of the performance metrics over $1000$ realizations.}
\label{supp:simul1.table.bic}
\resizebox{\columnwidth}{!}{\small 
\begin{tabular}{| cc | ccccc |}\hline
$(n = 500)$       & $\sigma$     & MOSUM & NOT.pwLin & NOT.pwLinCont & CPOP & TGUW  \\\hline
\multirow{4}{*}{\textsf{COUNTscore}} & 
 			   $0.5$  & 0.011 (0.1135) & 0.007 (0.0947) & 2.165 (1.6169) & 2.022 (0.1719) & 0.021 (0.1435)\\
 			 & $1  $ &  0.006 (0.0773) & 0.004 (0.0632) & 2.419 (2.4152) & 2.015 (0.1216) & 0.019 (0.1437)\\
 			 & $1.5$ & 0.007 (0.0834) & 0.007 (0.0834) & 2.562 (2.6354) & 2.019 (0.1366) & 0.039 (0.1988)\\
 			 & $2  $ &  0.046 (0.2448) & 0.004 (0.0632) & 2.354 (1.8466) & 2.018 (0.1403) & 0.095 (0.3034)\\
                            \hline
\multirow{4}{*}{\textsf{MAXscore1}} &
 			   $0.5$  & 0.019 (0.0139) & 0.022 (0.014) & 0.016 (0.008) & 0.005 (0.0057) & 0.028 (0.0217)\\
 			 & $1  $  & 0.029 (0.0208) & 0.039 (0.0237) & 0.025 (0.0304) & 0.008 (0.0085) & 0.046 (0.0344)\\
 			 & $1.5$ & 0.04 (0.0292) & 0.053 (0.0313) & 0.034 (0.0205) & 0.012 (0.0099) & 0.062 (0.0465)\\
 			 & $2  $ & 0.052 (0.0421) & 0.065 (0.0376) & 0.049 (0.0748) & 0.017 (0.013) & 0.077 (0.0558)\\
                            \hline
\multirow{4}{*}{\textsf{MAXscore2}} & 
 			   $0.5$  & 0.024 (0.0705) & 0.025 (0.0555) & 0.047 (0.1088) & 0.021 (0.1033) & 0.029 (0.0288)\\
 			 & $1  $  & 0.034 (0.0602) & 0.041 (0.0379) & 0.081 (0.2161) & 0.019 (0.078) & 0.048 (0.0449)\\
 			 & $1.5$ & 0.042 (0.052) & 0.057 (0.0678) & 0.125 (0.2893) & 0.021 (0.0594) & 0.065 (0.057)\\
 			 & $2  $  & 0.057 (0.0562) & 0.066 (0.0388) & 0.139 (0.2756) & 0.026 (0.075) & 0.081 (0.0717)\\
                            \hline\hline
$(n = 3500)$      & $\sigma$     & MOSUM & NOT.pwLin & NOT.pwLinCont & CPOP & TGUW  \\\hline
\multirow{4}{*}{\textsf{COUNTscore}} & 
 			   $0.5$ & 0 (0) & 0.002 (0.0447) & 2.004 (0.1265) & 2.003 (0.0547) & 0.037 (0.2041)\\
 			 & $1  $ & 0.001 (0.0316) & 0.003 (0.0547) & 2.029 (0.3692) & 2.007 (0.0947) & 0.029 (0.1679)\\
 			 & $1.5$ & 0 (0) & 0.003 (0.0547) & 2.01 (0.1262) & 2.004 (0.0894) & 0.043 (0.203)\\
 			 & $2  $ & 0 (0) & 0.003 (0.0547) & 2.017 (0.1508) & 2.004 (0.0632) & 0.09 (0.2863)\\
                            \hline
\multirow{4}{*}{\textsf{MAXscore1}} &
 			   $0.5$ & 0.06 (0.0391) & 0.073 (0.0406) & 0.033 (0.0134) & 0.012 (0.0101) & 0.093 (0.068)\\
 			 & $1  $ & 0.088 (0.0601) & 0.123 (0.0653) & 0.046 (0.0259) & 0.02 (0.0159) & 0.152 (0.1142)\\
 			 & $1.5$ & 0.105 (0.0724) & 0.159 (0.085) & 0.062 (0.041) & 0.027 (0.0206) & 0.203 (0.1475)\\
 			 & $2  $ & 0.126 (0.0888) & 0.193 (0.1057) & 0.083 (0.063) & 0.037 (0.0278) & 0.259 (0.1849)\\
                            \hline
\multirow{4}{*}{\textsf{MAXscore2}} & 
 			   $0.5$ & 0.06 (0.0391) & 0.076 (0.0831) & 0.079 (0.2573) & 0.025 (0.316) & 0.119 (0.4068)\\
 			 & $1  $ & 0.093 (0.1545) & 0.131 (0.2561) & 0.121 (0.5554) & 0.047 (0.4254) & 0.158 (0.1386)\\
 			 & $1.5$ & 0.105 (0.0724) & 0.161 (0.0961) & 0.113 (0.4401) & 0.041 (0.3216) & 0.208 (0.1633)\\
 			 & $2  $ & 0.126 (0.0888) & 0.196 (0.1252) & 0.132 (0.3214) & 0.048 (0.2331) & 0.266 (0.2093)\\
                            \hline
\end{tabular}
}
\end{table}

\begin{table}[!htb]
\centering
\caption{\ref{m:one}: Results from MOSUM, NOT.pwLin, NOT.pwLinCont, CPOP and TGUW when $n = 3500$ and the errors are generated as in \ref{e:two}--\ref{e:three} with $\sigma \in \{0.5, 1, 1.5, 2\}$. We report the average and standard error (in parentheses) of the performance metrics over $1000$ realizations.}
\label{simul1-laplace-t5.table.bic}
{\small
\begin{tabular}{| cc | cccc |}\hline
\ref{e:two}              & $\sigma$     & MOSUM & NOT.pwLin & NOT.pwLinCont & TGUW \\\hline
\multirow{4}{*}{\textsf{COUNTscore}} & 
			      $0.5$ & 0 (0) & 0.016 (0.1783) & 2.038 (0.4083) & 0.464 (1.1472)\\
			    & $1  $ & 0 (0) & 0.023 (0.1962) & 2.059 (0.4751) & 0.472 (1.1741)\\
			    & $1.5$ & 0 (0) & 0.02 (0.1888) & 2.096 (0.3673) & 0.459 (1.1134)\\
			    & $2  $ & 0.001 (0.0316) & 0.018 (0.1664) & 2.14 (0.4434) & 0.55 (1.1777)\\
                            \hline
\multirow{4}{*}{\textsf{MAXscore1}} &
			      $0.5$ & 0.057 (0.0408) & 0.071 (0.0418) & 0.033 (0.0136) & 0.098 (0.0729)\\
			    & $1  $ & 0.083 (0.0574) & 0.117 (0.0657) & 0.046 (0.0285) & 0.172 (0.1219)\\
			    & $1.5$ & 0.106 (0.0757) & 0.157 (0.0881) & 0.064 (0.0433) & 0.235 (0.1644)\\
			    & $2  $ & 0.126 (0.0909) & 0.196 (0.1109) & 0.081 (0.0649) & 0.299 (0.2174)\\
                            \hline
\multirow{4}{*}{\textsf{MAXscore2}} & 
			      $0.5$ & 0.057 (0.0408) & 0.129 (0.6975) & 0.148 (0.7339) & 0.848 (2.1083)\\
			    & $1  $ & 0.083 (0.0574) & 0.184 (0.776) & 0.23 (1.0301) & 0.84 (1.9763)\\
			    & $1.5$ & 0.106 (0.0757) & 0.22 (0.7302) & 0.472 (1.5789) & 0.931 (1.9843)\\
			    & $2  $ & 0.126 (0.0926) & 0.235 (0.5624) & 0.538 (1.515) & 1.009 (1.9956)\\
                            \hline\hline
\ref{e:three}                 & $\sigma$     & MOSUM & NOT.pwLin & NOT.pwLinCont & TGUW \\\hline
\multirow{4}{*}{\textsf{COUNTscore}} & 
			      $0.5$ & 0 (0) & 0.013 (0.1298) & 2.025 (0.3471) & 0.603 (1.2966)\\
			    & $1  $ & 0 (0) & 0.002 (0.0632) & 2.017 (0.1917) & 0.64 (1.3506)\\
			    & $1.5$ & 0 (0) & 0.015 (0.1575) & 2.035 (0.3064) & 0.672 (1.3814)\\
			    & $2  $ & 0 (0) & 0.006 (0.0893) & 2.039 (0.2399) & 0.712 (1.3345)\\
                            \hline
\multirow{4}{*}{\textsf{MAXscore1}} &
			      $0.5$ & 0.059 (0.0391) & 0.071 (0.0431) & 0.032 (0.0134) & 0.103 (0.0741)\\
			    & $1  $ & 0.083 (0.0582) & 0.117 (0.0664) & 0.047 (0.0281) & 0.176 (0.1259)\\
			    & $1.5$ & 0.108 (0.0776) & 0.153 (0.0873) & 0.062 (0.0439) & 0.237 (0.1632)\\
			    & $2  $ & 0.125 (0.0893) & 0.187 (0.1047) & 0.08 (0.0575) & 0.296 (0.2087)\\
                            \hline
\multirow{4}{*}{\textsf{MAXscore2}} & 
			      $0.5$ & 0.059 (0.0391) & 0.091 (0.3739) & 0.127 (0.6998) & 1.055 (2.3195)\\
			    & $1  $ & 0.083 (0.0582) & 0.124 (0.2406) & 0.11 (0.4325) & 1.155 (2.3102)\\
			    & $1.5$ & 0.108 (0.0776) & 0.178 (0.4126) & 0.15 (0.596) & 1.185 (2.2676)\\
			    & $2  $ & 0.125 (0.0893) & 0.191 (0.1571) & 0.224 (0.8226) & 1.239 (2.1685)\\
                            \hline
\end{tabular}
}
\end{table}

\begin{table}[!htb]
\centering
\caption{\ref{m:one}: Results from MOSUM, MOSUM.dlrv, NOT.pwLin, NOT.pwLinCont, CPOP and TGUW when $n = 3500$ and the errors are generated as in \ref{e:four} with $\sigma \in \{0.5, 1, 1.5, 2\}$. We report the average and standard error (in parentheses) of the performance metrics over $1000$ realizations.}
\label{simul1.table.ar1}
{\small 
\begin{tabular}{| cc | cccc |}\hline
($\rho = 0.3$)      & $\sigma$     & MOSUM & MOSUM.dlrv & NOT.pwLin & TGUW  \\\hline
\multirow{4}{*}{\textsf{COUNTscore}} & 
 			   $0.5$ & 0.036 (0.2162) & 0.423 (0.5003) & 0.12 (0.4688) & 2.366 (2.1947)\\
 			 & $1  $ & 0.044 (0.2369) & 0.492 (0.5237) & 0.123 (0.4517) & 2.236 (2.1272)\\
 			 & $1.5$ & 0.033 (0.1842) & 0.693 (0.5649) & 0.188 (0.5941) & 2.331 (2.1877)\\
 			 & $2  $ & 0.033 (0.1842) & 0.87 (0.5453) & 0.17 (0.5342) & 2.386 (2.1141)\\
			 \hline
\multirow{4}{*}{\textsf{MAXscore1}} &
 			   $0.5$  & 0.071 (0.0465) & 0.392 (0.3878) & 0.087 (0.0532) & 0.104 (0.078)\\
 			 & $1  $ & 0.101 (0.0686) & 0.44 (0.3841) & 0.142 (0.0795) & 0.171 (0.1247)\\
 			 & $1.5$& 0.13 (0.09) & 0.552 (0.394) & 0.203 (0.125) & 0.229 (0.1711)\\
 			 & $2  $ & 0.159 (0.1151) & 0.7 (0.4083) & 0.246 (0.1392) & 0.287 (0.2122)\\
                            \hline
\multirow{4}{*}{\textsf{MAXscore2}} & 
 			   $0.5$  & 0.205 (0.8384) & 0.552 (0.572) & 0.366 (1.2819) & 3.47 (3.2545)\\
 			 & $1  $ & 0.241 (0.8473) & 0.658 (0.5796) & 0.442 (1.3325) & 3.403 (3.1662)\\
 			 & $1.5$& 0.25 (0.7988) & 0.907 (0.6003) & 0.642 (1.5826) & 3.566 (3.2606)\\
 			 & $2  $ & 0.29 (0.8461) & 1.156 (0.6117) & 0.621 (1.3958) & 3.539 (3.1179)\\
                            \hline
\hline
($\rho = 0.7$)      & $\sigma$     & MOSUM & MOSUM.dlrv & NOT.pwLin & TGUW  \\\hline
\multirow{4}{*}{\textsf{COUNTscore}} & 
 			   $0.5$ & 8.181 (2.9303) & 0.437 (0.5023) & 10.966 (7.2253) & 123.015 (11.0449)\\
 			 & $1  $ & 8.16 (2.8497) & 0.599 (0.5836) & 12.368 (6.1324) & 123.405 (10.9049)\\
 			 & $1.5$ &8.097 (2.9567) & 0.706 (0.6588) & 13.483 (5.4558) & 123.218 (10.7549)\\
 			 & $2  $ & 8.259 (2.9807) & 0.987 (0.7399) & 13.079 (5.7503) & 122.899 (11.013)\\
			 \hline
\multirow{4}{*}{\textsf{MAXscore1}} &
 			   $0.5$  & 0.092 (0.0687) & 0.245 (0.3147) & 0.131 (0.1028) & 0.088 (0.0737)\\
 			 & $1  $ & 0.147 (0.1183) & 0.43 (0.3839) & 0.219 (0.1727) & 0.104 (0.095)\\
 			 & $1.5$& 0.206 (0.1662) & 0.715 (0.4734) & 0.275 (0.2185) & 0.112 (0.1083)\\
 			 & $2  $ & 0.245 (0.1944) & 0.979 (0.5241) & 0.33 (0.2583) & 0.115 (0.1117)\\
                            \hline
\multirow{4}{*}{\textsf{MAXscore2}} & 
 			   $0.5$  &7.585 (1.5808) & 0.588 (0.5667) & 6.764 (3.4284) & 9.809 (0.1327)\\
 			 & $1  $ & 7.585 (1.5782) & 0.8 (0.6272) & 7.784 (2.2462) & 9.806 (0.1351)\\
 			 & $1.5$& 7.543 (1.5966) & 1.03 (0.8024) & 8.302 (1.5665) & 9.803 (0.1361)\\
 			 & $2  $ & 7.577 (1.581) & 1.455 (0.9964) & 8.151 (1.688) & 9.797 (0.1525)\\
                            \hline
\end{tabular}
}
\end{table}


\begin{table}[!htb]
\centering
\caption{\ref{m:two}: Results from MOSUM, NOT.pwLin, NOT.pwLinCont, CPOP and TGUW when the errors are generated as in \ref{e:one} with $\sigma \in \{0.5, 1, 1.5, 2\}$. We report the average and standard error (in parentheses) of the performance metrics over $1000$ realizations.}
\label{supp:simul2.table.bic}
\resizebox{\columnwidth}{!}{\small
\begin{tabular}{| cc | ccccc |}\hline
($n=500$)      & $\sigma$     & MOSUM & NOT.pwLin & NOT.pwLinCont & CPOP & TGUW  \\\hline
\multirow{4}{*}{\textsf{COUNTscore}} & 
			      $0.5$  & 0.012 (0.1178) & 0.007 (0.0947) & 0 (0) & 0.028 (0.193) & 0.076 (0.2726)\\
			    & $1  $  & 0.007 (0.0834) & 0.004 (0.0632) & 0.001 (0.0316) & 0.022 (0.1534) & 0.043 (0.2078)\\
			    & $1.5$  & 0.005 (0.0706) & 0.014 (0.1258) & 0.001 (0.0316) & 0.018 (0.1403) & 0.047 (0.221)\\
			    & $2  $  & 0.003 (0.0547) & 0.006 (0.0773) & 0.003 (0.0547) & 0.02 (0.1601) & 0.04 (0.2108)\\
                            \hline
\multirow{4}{*}{\textsf{MAXscore1}} &
			      $0.5$  & 0.034 (0.0147) & 0.041 (0.0164) & 0.011 (0.0052) & 0.013 (0.007) & 0.055 (0.026)\\
			    & $1  $  & 0.053 (0.0237) & 0.07 (0.0261) & 0.02 (0.0097) & 0.021 (0.013) & 0.09 (0.0411)\\
			    & $1.5$  & 0.073 (0.0344) & 0.096 (0.0372) & 0.028 (0.0158) & 0.029 (0.018) & 0.124 (0.0579)\\
			    & $2  $  & 0.089 (0.0389) & 0.12 (0.0447) & 0.038 (0.0213) & 0.038 (0.0225) & 0.152 (0.0674)\\
                            \hline
\multirow{4}{*}{\textsf{MAXscore2}} & 
			      $0.5$  & 0.04 (0.0695) & 0.043 (0.0342) & 0.011 (0.0052) & 0.025 (0.1117) & 0.061 (0.0551)\\
			    & $1  $  & 0.057 (0.0597) & 0.071 (0.035) & 0.02 (0.01) & 0.028 (0.0848) & 0.094 (0.0531)\\
			    & $1.5$  & 0.075 (0.0537) & 0.099 (0.0485) & 0.028 (0.0195) & 0.035 (0.0796) & 0.128 (0.0682)\\
			    & $2  $  & 0.09 (0.0435) & 0.122 (0.0512) & 0.039 (0.0271) & 0.046 (0.0898) & 0.157 (0.0826)\\
                            \hline\hline
($n=3500$)      & $\sigma$     & MOSUM & NOT.pwLin & NOT.pwLinCont & CPOP & TGUW  \\\hline 
\multirow{4}{*}{\textsf{COUNTscore}} & 
 			      $0.5$  & 0.002 (0.0447) & 0.004 (0.0632) & 0 (0) & 0.007 (0.0834) & 0.093 (0.3073)\\
 			    & $1  $  & 0 (0) & 0 (0) & 0 (0) & 0.003 (0.0547) & 0.069 (0.2652)\\
 			    & $1.5$  & 0 (0) & 0.001 (0.0316) & 0 (0) & 0.006 (0.0773) & 0.073 (0.2641)\\
 			    & $2  $  & 0 (0) & 0.003 (0.0547) & 0 (0) & 0.005 (0.0706) & 0.065 (0.2466)\\
                            \hline
\multirow{4}{*}{\textsf{MAXscore1}} &
 			      $0.5$  & 0.122 (0.0552) & 0.16 (0.0616) & 0.023 (0.0121) & 0.028 (0.0207) & 0.229 (0.1176)\\
 			    & $1  $  & 0.186 (0.0883) & 0.262 (0.1016) & 0.047 (0.0253) & 0.05 (0.0261) & 0.372 (0.1817)\\
 			    & $1.5$  & 0.254 (0.1327) & 0.359 (0.1431) & 0.073 (0.0409) & 0.076 (0.0528) & 0.506 (0.2401)\\
 			    & $2  $  & 0.32 (0.1639) & 0.44 (0.1728) & 0.099 (0.0565) & 0.1 (0.0634) & 0.623 (0.2983)\\
                             \hline
\multirow{4}{*}{\textsf{MAXscore2}} & 
 			      $0.5$  & 0.130 (0.2598) & 0.168 (0.2147) & 0.023 (0.0121) & 0.046 (0.4038) & 0.246 (0.1674)\\
 			    & $1  $  & 0.186 (0.0883) & 0.262 (0.1016) & 0.047 (0.0253) & 0.054 (0.0969) & 0.393 (0.2488)\\
 			    & $1.5$  & 0.254 (0.1327) & 0.359 (0.1431) & 0.073 (0.0409) & 0.079 (0.1024) & 0.531 (0.309)\\
 			    & $2  $  & 0.32 (0.1639) & 0.442 (0.193) & 0.099 (0.0565) & 0.109 (0.2848) & 0.648 (0.368)\\
                             \hline
\end{tabular}
}
\end{table}

\begin{table}[!htb]
\centering
\caption{\ref{m:two}: Results from MOSUM, NOT.pwLin, NOT.pwLinCont, CPOP and TGUW when $n = 3500$ and the errors are generated as in \ref{e:two}--\ref{e:three} with $\sigma \in \{0.5, 1, 1.5, 2\}$. We report the average and standard error (in parentheses) of the performance metrics over $1000$ realizations.}
\label{simul2-laplace-t5.table.bic}
\resizebox{\columnwidth}{!}
{\small
\begin{tabular}{| cc | ccccc |}\hline
\ref{e:two}                 & $\sigma$     & MOSUM & NOT.pwLin & NOT.pwLinCont & CPOP & TGUW \\\hline
\multirow{4}{*}{\textsf{COUNTscore}} & 
			      $0.5$ & 0 (0) & 0.005 (0.0836) & 0 (0) & 5.903 (4.2374) & 0.596 (1.2459)\\
			    & $1  $ & 0.003 (0.0547) & 0.019 (0.1633) & 0.004 (0.0632) & 5.989 (4.1427) & 0.659 (1.237)\\
			    & $1.5$ & 0.001 (0.0316) & 0.018 (0.1473) & 0.005 (0.0948) & 6.402 (4.2131) & 0.847 (1.3869)\\
			    & $2  $ & 0.007 (0.0834) & 0.023 (0.2203) & 0.006 (0.0773) & 6.051 (4.1983) & 0.744 (1.276)\\
                            \hline
\multirow{4}{*}{\textsf{MAXscore1}} &
			      $0.5$  & 0.178 (0.103) & 0.252 (0.1352) & 0.045 (0.0332) & 0.09 (0.1164) & 0.414 (0.2632)\\
			    & $1  $  & 0.336 (0.5639) & 0.454 (0.5608) & 0.123 (0.5478) & 0.204 (0.3175) & 0.751 (0.6129)\\
			    & $1.5$ & 0.418 (0.3909) & 0.577 (0.308) & 0.145 (0.1045) & 0.298 (0.3109) & 1 (0.6024)\\
			    & $2  $  & 0.53 (0.7981) & 0.714 (0.6252) & 0.225 (0.5635) & 0.421 (0.4754) & 1.238 (0.8628)\\
                            \hline
\multirow{4}{*}{\textsf{MAXscore2}} & 
			      $0.5$  &  0.178 (0.103) & 0.274 (0.4542) & 0.045 (0.0332) & 4.518 (3.2269) & 1.168 (2.0277)\\
			    & $1  $  & 0.306 (0.1857) & 0.482 (0.728) & 0.102 (0.311) & 4.634 (3.2025) & 1.501 (1.9696)\\
			    & $1.5$ & 0.408 (0.244) & 0.642 (0.7988) & 0.159 (0.3044) & 4.884 (3.1902) & 1.965 (2.1865)\\
			    & $2  $  & 0.471 (0.295) & 0.731 (0.6933) & 0.201 (0.2448) & 4.53 (3.196) & 2.024 (2.034)\\
                            \hline\hline
\ref{e:three}             & $\sigma$     & MOSUM & NOT.pwLin & NOT.pwLinCont & CPOP & TGUW  \\\hline
\multirow{4}{*}{\textsf{COUNTscore}} & 
			      $0.5$ & 0 (0) & 0.003 (0.0547) & 0 (0) & 3.532 (3.4054) & 0.815 (1.401)\\
			    & $1  $ & 0.002 (0.0447) & 0.016 (0.1542) & 0.002 (0.0447) & 3.606 (3.3289) & 0.866 (1.443)\\
			    & $1.5$ & 0.001 (0.0316) & 0.008 (0.0891) & 0 (0) & 3.53 (3.2711) & 0.855 (1.4177)\\
			    & $2  $ & 0.005 (0.0706) & 0.013 (0.1444) & 0.002 (0.0632) & 3.8 (3.1957) & 0.841 (1.337)\\
                            \hline
\multirow{4}{*}{\textsf{MAXscore1}} &
			      $0.5$ & 0.189 (0.1125) & 0.251 (0.1279) & 0.046 (0.033) & 0.079 (0.084) & 0.425 (0.264)\\
			    & $1  $ & 0.314 (0.4752) & 0.424 (0.4657) & 0.111 (0.4489) & 0.168 (0.3693) & 0.747 (0.5885)\\
			    & $1.5$ & 0.39 (0.3798) & 0.584 (0.3116) & 0.151 (0.1275) & 0.274 (0.297) & 0.943 (0.5587)\\
			    & $2  $ & 0.528 (0.7272) & 0.703 (0.3865) & 0.195 (0.1642) & 0.36 (0.3768) & 1.204 (0.8215)\\
                            \hline
\multirow{4}{*}{\textsf{MAXscore2}} & 
			      $0.5$ & 0.189 (0.1125) & 0.27 (0.443) & 0.046 (0.033) & 3.17 (3.3651) & 1.471 (2.2634)\\
			    & $1  $ & 0.294 (0.182) & 0.433 (0.4659) & 0.092 (0.0677) & 3.221 (3.3163) & 1.761 (2.2033)\\
			    & $1.5$ & 0.381 (0.2392) & 0.615 (0.6003) & 0.151 (0.1275) & 3.117 (3.2616) & 1.875 (2.0448)\\
			    & $2  $ & 0.48 (0.3102) & 0.709 (0.4006) & 0.198 (0.2067) & 3.332 (3.2156) & 2.137 (2.0815)\\
                            \hline
\end{tabular}
}
\end{table}

\begin{table}[!htb]
\centering
\caption{\ref{m:three}: Results from MOSUM, NOT.pwLin and TGUW when the errors are generated as in \ref{e:one}--\ref{e:three} with $\sigma \in \{0.5, 1, 1.5, 2\}$. We report the average and standard error (in parentheses) of the performance metrics over $1000$ realizations.}
\label{simul9.table.bic}
{\small 
\begin{tabular}{| cc | ccc |}\hline
 \ref{e:one}      & $\sigma$     & MOSUM & NOT.pwLin & TGUW  \\\hline
\multirow{4}{*}{\textsf{COUNTscore}} & 
 			   $0.5$ & 0 (0) & 0.001 (0.0316) & 0.083 (0.3003)\\
 			 & $1  $ & 0 (0) & 0.003 (0.0547) & 0.196 (0.5328)\\
 			 & $1.5$& 0.008 (0.0891) & 0.001 (0.0316) & 0.274 (0.6598)\\
 			 & $2  $ & 	0.031 (0.1734) & 0.003 (0.0547) & 0.37 (0.8086)\\
			 \hline
\multirow{4}{*}{\textsf{MAXscore1}} &
 			   $0.5$  & 0.111 (0.0592) & 0.144 (0.0574) & 0.198 (0.1066)\\
 			 & $1  $ & 0.182 (0.0943) & 0.244 (0.0965) & 0.319 (0.1676)\\
 			 & $1.5$& 0.272 (0.3497) & 0.331 (0.1353) & 0.455 (0.2347)\\
 			 & $2  $ & 0.313 (0.3155) & 0.412 (0.1764) & 0.564 (0.2768)\\
                            \hline
\multirow{4}{*}{\textsf{MAXscore2}} & 
 			   $0.5$  & 0.111 (0.0592) & 0.145 (0.058) & 0.217 (0.1597)\\
 			 & $1  $ & 0.182 (0.0943) & 0.248 (0.1671) & 0.433 (0.4911)\\
 			 & $1.5$& 0.243 (0.131) & 0.331 (0.1353) & 0.604 (0.5761)\\
 			 & $2  $ & 0.293 (0.1698) & 0.414 (0.1789) & 0.731 (0.6109)\\
\hline\hline
 \ref{e:two}      & $\sigma$     & MOSUM & NOT.pwLin & TGUW  \\\hline

\multirow{4}{*}{\textsf{COUNTscore}} & 
 			   $0.5$ & 0 (0) & 0.024 (0.1716) & 0.383 (0.8793)\\
 			 & $1  $ & 0 (0) & 0.025 (0.1685) & 0.485 (0.9406)\\
 			 & $1.5$& 0.013 (0.1133) & 0.031 (0.2147) & 0.472 (0.9198)\\
 			 & $2  $ & 0.035 (0.1892) & 0.03 (0.1927) & 0.505 (0.9321)\\
			 \hline
\multirow{4}{*}{\textsf{MAXscore1}} &
 			   $0.5$ & 0.116 (0.0636) & 0.147 (0.0623) & 0.215 (0.1071)\\
 			 & $1  $ & 0.18 (0.0917) & 0.239 (0.1029) & 0.384 (0.1921)\\
 			 & $1.5$& 0.281 (0.4169) & 0.328 (0.1402) & 0.539 (0.2566)\\
 			 & $2  $ & 0.303 (0.2415) & 0.402 (0.1803) & 0.664 (0.3307)\\
                            \hline
\multirow{4}{*}{\textsf{MAXscore2}} & 
 			   $0.5$ & 0.116 (0.0636) & 0.157 (0.1742) & 0.412 (0.6519)\\
 			 & $1  $ & 0.18 (0.0917) & 0.248 (0.1997) & 0.607 (0.6723)\\
 			 & $1.5$& 0.242 (0.1665) & 0.342 (0.229) & 0.768 (0.7023)\\
 			 & $2  $ & 0.296 (0.1805) & 0.442 (0.395) & 0.902 (0.737)\\
\hline\hline
 \ref{e:three}      & $\sigma$     & MOSUM & NOT.pwLin & TGUW  \\\hline
 \multirow{4}{*}{\textsf{COUNTscore}} & 
 			   $0.5$ & 0 (0) & 0.01 (0.1262) & 0.504 (1.0673)\\
 			 & $1  $ & 0 (0) & 0.007 (0.0947) & 0.518 (1.0221)\\
 			 & $1.5$& 0.01 (0.0995) & 0.015 (0.1442) & 0.498 (0.9607)\\
 			 & $2  $ & 0.058 (0.2544) & 0.009 (0.1045) & 0.615 (1.1189)\\
			 \hline
\multirow{4}{*}{\textsf{MAXscore1}} &
 			   $0.5$ & 0.112 (0.0605) & 0.15 (0.0613) & 0.222 (0.1148)\\
 			 & $1  $ & 0.177 (0.0976) & 0.244 (0.0999) & 0.392 (0.1863)\\
 			 & $1.5$& 0.279 (0.3972) & 0.333 (0.1368) & 0.526 (0.2588)\\
 			 & $2  $ & 0.311 (0.2349) & 0.412 (0.184) & 0.676 (0.3179)\\
                            \hline
\multirow{4}{*}{\textsf{MAXscore2}} & 
 			   $0.5$ & 0.112 (0.0605) & 0.156 (0.1239) & 0.504 (0.8012)\\
 			 & $1  $ & 0.177 (0.0976) & 0.251 (0.203) & 0.646 (0.7245)\\
 			 & $1.5$& 0.24 (0.1269) & 0.337 (0.1483) & 0.756 (0.7051)\\
 			 & $2  $ & 0.308 (0.1983) & 0.421 (0.2496) & 0.915 (0.7361)\\
\hline
\end{tabular}
}
\end{table}

\begin{table}[!htb]
\centering
\caption{\ref{m:four}: Results from MOSUM, NOT.pwLin, NOT.pwConst and TGUW when the errors are generated as in \ref{e:one} with $\sigma \in \{0.5, 1, 1.5, 2\}$. We report the average and standard error (in parentheses) of the performance metrics over $1000$ realizations.}
\label{supp:simul3.table.bic}
{\small 
\begin{tabular}{| cc | cccc |}\hline
$(n=500)$ & $\sigma$     & MOSUM & NOT.pwLin & NOT.pwConst & TGUW \\ \hline
\multirow{4}{*}{\textsf{COUNTscore}} & 
 			      $0.5$  & 0.002 (0.0447)& 0.001 (0.0316) & 0.039 (0.2086) & 0.007 (0.1377)\\
 			    & $1  $  & 0.005 (0.0706)& 0.005 (0.0706) & 0.034 (0.1868) & 0.003 (0.0707)\\
 			    & $1.5$  &0.005 (0.0706) & 0.001 (0.0316) & 0.039 (0.227) & 0.007 (0.1047)\\
 			    & $2  $  & 0.01 (0.0995) & 0.005 (0.0706) & 0.054 (0.2778) & 0.024 (0.1657)\\
                            \hline
\multirow{4}{*}{\textsf{MAXscore1}} &
 			      $0.5$  & 0 (0) & 0 (0) & 0 (0) & 0 (0)\\
 			    & $1  $  & 0 (0) & 0 (0) & 0 (0) & 0 (0)\\
 			    & $1.5$  & 0 (0) & 0 (0) & 0 (0) & 0 (0.0018)\\
 			    & $2  $  & 0.001 (0.0107) & 0 (0.0033) & 0 (0.001) & 0.001 (0.0048)\\
                            \hline
\multirow{4}{*}{\textsf{MAXscore2}} & 
 			      $0.5$  & 0.001 (0.0288) & 0.001 (0.0256) & 0.012 (0.0876) & 0.001 (0.0238)\\
 			    & $1  $  & 0.002 (0.0345) & 0.001 (0.0204) & 0.012 (0.0915) & 0.001 (0.0178)\\
 			    & $1.5$  & 0.002 (0.0285) & 0 (0.0073) & 0.015 (0.1092) & 0.003 (0.0511)\\
 			    & $2  $  & 0.004 (0.0465) & 0.003 (0.0504) & 0.022 (0.1319) & 0.005 (0.0563)\\
                            \hline \hline
$(n=3500)$ & $\sigma$     & MOSUM & NOT.pwLin & NOT.pwConst & TGUW \\ \hline
\multirow{4}{*}{\textsf{COUNTscore}} & 
 			      $0.5$   & 0 (0) & 0.003 (0.0547) & 0.006 (0.0773) & 0.005 (0.0706)\\
 			    & $1  $   & 0 (0) & 0.001 (0.0316) & 0.008 (0.0891) & 0.073 (0.2603)\\
 			    & $1.5$  & 0.018 (0.133) & 0.005 (0.0706) & 0.011 (0.1135) & 0.173 (0.3889)\\
 			    & $2  $   & 0.162 (0.382) & 0 (0) & 0.019 (0.1366) & 0.263 (0.4879)\\
                            \hline
\multirow{4}{*}{\textsf{MAXscore1}} &
 			      $0.5$    & 0 (0) & 0 (0) & 0 (0) & 0 (0.0018)\\
 			    & $1  $   & 0.001 (0.0026) & 0.001 (0.0026) & 0.001 (0.0026) & 0.006 (0.0131)\\
 			    & $1.5$  & 0.013 (0.0741) & 0.003 (0.0064) & 0.003 (0.0064) & 0.021 (0.0335)\\
 			    & $2  $   & 0.121 (0.2746) & 0.007 (0.0108) & 0.007 (0.011) & 0.049 (0.0761)\\
                            \hline
\multirow{4}{*}{\textsf{MAXscore2}} & 
 			      $0.5$   & 0 (0) & 0.014 (0.3351) & 0.012 (0.299) & 0.001 (0.0142)\\
 			    & $1  $  & 0.001 (0.0026) & 0.001 (0.0048) & 0.035 (0.4835) & 0.008 (0.0185)\\
 			    & $1.5$  &0.017 (0.0971) & 0.024 (0.4353) & 0.037 (0.498) & 0.028 (0.0523)\\
 			    & $2  $  & 0.155 (0.3468) & 0.007 (0.0108) & 0.062 (0.5403) & 0.068 (0.1884)\\
                            \hline
\end{tabular}
}
\end{table}

\clearpage

\subsection{Empirical size control}
\label{sec:size}

We additionally conduct numerical experiments to examine the size control of the proposed MOSUM-based test.
For this, we adopt the null model~\ref{m:zero} in Section~\ref{sec:comp:sim:models} with Gaussian errors generated as in~\ref{e:one} with varying variance to generate $1000$ realizations. 
On each realization, we compute the test statistics $W_{n}(G)$ and examine whether it exceeds a threshold.
Here, we compare two approaches to get the threshold: (i) theoretically motivated critical value $C_n(G, \alpha) = a_{G}^{-1}\big(b_{G} - \log(- \log (1-\alpha)/2)\big)$ with $\alpha = 0.05$ and $\log(H) = 0.7284$ as suggested in Section~\ref{sec:tuning}, and (ii) the empirical critical value obtained as the sample $(1-\alpha)$-quantile of the simulated values of $W_{n}(G)$ for given $n$ and $G$.
We note that while (i) does not require the knowledge of the data generating process, (ii) does.
Table~\ref{table:size} shows that as expected, the approach~(ii) results in the empirical size close to the significance level $\alpha = 0.05$, whereas (i)~tends to be more conservative. 
With both approaches, using the proposed multiscale extension with bandwidths $\cG = \{50, 100, 150, 250, 400, 650\}$, we also produce an estimator $\wh J_n$ of the number of change points ($J = 0$), the results from which are summarized in Table~\ref{table:critical:countscore}.
From this, we conclude that using either asymptotic or empirical critical values leads to satisfactory size control and estimation consistency under $\mc H_0: \, J = 0$, but the former approach is more preferable as it depends only on $n$ and $G$ and does not require the prior knowledge of the data generating process or heavy computation.
	
		\begin{table}[!h]
		\centering
		\caption{Empirical sizes when using (i) asymptotic and (ii) empirical critical values for varying $G$ and $\sigma$.} 
	\label{table:size}
        \begin{tabular}{l|c|c|c|c|c|c}
        \hline
        & $G=50$ & $G=100$ & $G=150$ & $G=250$ & $G=400$ & $G=650$\\
        \hline
        \multicolumn{7}{l}{\textbf{Asymptotic critical value}}\\
        \hline
        \hspace{1em}$\sigma=0.5$ & 0.035 & 0.015 & 0.013 & 0.012 & 0.008 & 0.002\\
        \hline
        \hspace{1em}$\sigma=1$ & 0.038 & 0.020 & 0.012 & 0.010 & 0.010 & 0.008\\
        \hline
        \hspace{1em}$\sigma=1.5$ & 0.019 & 0.020 & 0.014 & 0.008 & 0.005 & 0.003\\
        \hline
        \hspace{1em}$\sigma=2$ & 0.033 & 0.022 & 0.013 & 0.015 & 0.005 & 0.002\\
        \hline
        \multicolumn{7}{l}{\textbf{Empirical critical value}}\\
        \hline
        \hspace{1em}$\sigma=0.5$ & 0.054 & 0.053 & 0.045 & 0.041 & 0.039 & 0.053\\
        \hline
        \hspace{1em}$\sigma=1$ & 0.050 & 0.060 & 0.056 & 0.057 & 0.054 & 0.059\\
        \hline
        \hspace{1em}$\sigma=1.5$ & 0.040 & 0.051 & 0.043 & 0.036 & 0.037 & 0.052\\
        \hline
        \hspace{1em}$\sigma=2$ & 0.051 & 0.052 & 0.048 & 0.053 & 0.034 & 0.056\\
        \hline
        \end{tabular}
        \end{table}
		\begin{table}[!h]
		\centering
		\caption{Accuracy in estimating the number of change points when using (i) asymptotic and (ii) empirical critical values. We report the average and standard error (in parentheses) of \textsf{COUNTscore} over $1000$ realizations.}
	\label{table:critical:countscore}
        \begin{tabular}{l|c|c}
        \hline
          & Asymptotic & Empirical\\
        \hline
        \hspace{1em} $\sigma=0.5$ & 0 (0) & 0 (0)\\
        \hline
        \hspace{1em} $\sigma=1$ & 0 (0) & 0 (0)\\
        \hline
        \hspace{1em} $\sigma=1.5$ & 0.001 (0.0316) & 0 (0)\\
        \hline
        \hspace{1em} $\sigma=2$ & 0 (0) & 0 (0)\\
        \hline
        \end{tabular}
        \end{table}

\clearpage

\section{Proofs}
\label{sec:proof}

This section gives the proofs for the theorems. We define $\mbf x_{i, k} = (1, (i - k)/G)^\top$ and let 
\begin{align}
& \bC_{G, +} = 
\bmx G & \frac{G+1}{2} \\ \frac{G+1}{2} & \frac{(G+1)(2G+1)}{6G} \emx
\quad \text{and} \quad
\bC_{G, -} = 
\begin{bmatrix} G & -\frac{G-1}{2} \\ -\frac{G-1}{2} & \frac{(G-1)(2G-1)}{6G} \end{bmatrix}.
\nn 
\end{align}
Further, we define
\begin{align}
& \bC_{+} = \begin{bmatrix} 1& \frac{1}{2} \\ \frac{1}{2} & \frac{1}{3} \end{bmatrix}, 
\quad 
\bD_{G,+} = \begin{bmatrix} 0 & \frac{1}{2G} \\ \frac{1}{2G} & \frac{1}{2G} + \frac{1}{6G^2} \end{bmatrix}, \quad \text{and}
\label{def_C+}
\\
& \bC_{-} = \begin{bmatrix} 1& -\frac{1}{2} \\ -\frac{1}{2} & \frac{1}{3} \end{bmatrix}, \quad 
\bD_{G,-} = \begin{bmatrix} 0 & \frac{1}{2G} \\ \frac{1}{2G} & -\frac{1}{2G} + \frac{1}{6G^2} \end{bmatrix}, \label{def_C-}
\end{align} 
so that $G^{-1} \bC_{G, +} = \bC_{+} + \bD_{G, +}$ and 
$G^{-1} \bC_{G, -} = \bC_{-} + \bD_{G,-}$.
Throughout, for a matrix $\mbf A \in \mathbb{R}^{p \times q}$, we denote by $\Vert \mbf A \Vert_r$ the matrix norms induced by vector $r$-norms for a given $r$.
We denote by $\mbf 0$ and $\mbf I$ the vector of zeros and the identity matrix, respectively, whose dimensions are determined in the context.
For a finite-dimensional matrix $\mbf A$, we write $\mbf A = O(G^{-1})$ to denote that all its elements are bounded as $O(G^{-1})$.

\subsection{Preliminary lemmas}

\begin{lem}
\label{lem:invC}
\begin{align*}
G \mbf C_{G, +}^{-1} = \mbf C_+^{-1} + \frac{1}{G - 1} \bmx 6 & - 6 \\ - 6 & \frac{12}{G + 1} \emx 
\quad \text{and} \quad
G \mbf C_{G, -}^{-1} = \mbf C_-^{-1} + \frac{1}{G + 1} \bmx - 6 & - 6 \\ - 6 & \frac{12}{G - 1} \emx.
\end{align*}
\end{lem}

\begin{lem} \label{lemA.1} 
Assume that~\ref{a:two}--\ref{a:three} hold and
let $\wt{W}(t) = G^{-1/2} W (Gt)$. Then,
\begin{align}
& \sup_{1 \le t \le \frac{n}{G} - 1} \left\vert  G^{-1/2}\sum_{i= \lfloor tG \rfloor +1}^{\lfloor tG \rfloor +G} \mbf x_{i, \lfloor tG \rfloor} \epsilon_i - \tau \int_{t}^{t+1} \begin{bmatrix} 1 \\ (s-t)\end{bmatrix} \diff \wt{W}(s)\right\vert 
=  o_P \left( \frac{1}{\sqrt{\log (n/G)}}\right),
\label{A2}
\\
& \sup_{1 \le t \le \frac{n}{G} - 1} \left\vert G^{-1/2}\sum_{i=\lfloor tG \rfloor - G+1}^{\lfloor tG \rfloor} \mbf x_{i, \lfloor tG \rfloor} \epsilon_i - \tau \int_{t-1}^{t} \begin{bmatrix} 1 \\ (s-t)\end{bmatrix} \diff \wt{W}(s)\right\vert  = o_P \left( \frac{1}{\sqrt{\log (n/G)}}\right).
\label{A3}
\end{align}
\end{lem}
\begin{proof}
Assumption~\ref{a:two} implies that 
\begin{align*}
\max_{G \le k \le n-G} \l\vert  \sum_{i=k+1}^{k+G} \epsilon_i - \tau \l( W(k+G) - W(k) \r)\right\vert = O\l( n^{\frac{1}{2+\nu}} \r) \quad \text{ a.s.}
\end{align*}
It also implies that, for $\alpha = 1, 2, \ldots$, \cite[page 237]{csorgo1997},
\begin{align*}
& \frac{1}{n^{\alpha}} \left\vert  \sum_{i=1}^{n} i^{\alpha} \epsilon_i - \tau \int_{0}^{n} s^{\alpha} \diff W(s) \right\vert = O \l(n^{\frac{1}{2+\nu}} \r) \quad \text{ {a.s.}} \quad \text{where} \\
& \left\{ \int_{0}^{t} s^{\alpha} \diff W(s), \, 0 \le t < \infty \right\} \overset{d}{\equiv} \left\{ W \left( \frac{1}{2\alpha+1} t^{2\alpha+1} \right) , ~ 0 \le t < \infty \right\}.
\end{align*}
Thus we have 
\begin{align*}
\max_{G \le k \le n-G} \left\vert \sum_{i=k+1}^{k+G} \frac{i-k}{G}  \epsilon_i - \frac{\tau}{G} \int_{k}^{k+G} (s-k) \diff W(s) \right\vert  = O \left( \frac{ n^{1+\frac{1}{2+\nu}}}{G} \right) \quad \text{a.s.}
\end{align*} 
As $W(t)$ satisfies 
\begin{align*} 
\sup_{1 \le t \le N - 1} \sup_{0 \le s < 1} \vert W(t+s) - W(t) \vert = O\left( \sqrt{\log(N)} \right) \quad \text{a.s.,} 
\end{align*}
and by Theorem~1.2.1 of \cite{csorgo1979}, the Gaussian process 
$Y(t) = \int_{0}^{t} s \diff W(s)$ satisfies 
\begin{align*}
\sup _{1 \le t \le N - 1} \sup_{0 \le s < 1} \vert Y(t+s) - Y(t) \vert = O \left( N\sqrt{\log(N)} \right) \quad \text{a.s.,} 
\end{align*}
we obtain
\begin{align*}
& \sup_{1 \le t \le \frac{n}{G} - 1 } \left\vert  \int_{\lfloor tG \rfloor}^{\lfloor tG \rfloor + G} \diff W(s) - \int_{tG}^{(t+1)G} \diff W(s) \right\vert 
\\&= \sup_{1 \le t \le \frac{n}{G} - 1} \left\vert \l(W\big(\lfloor (t+1)G \rfloor \big) - W\big( (t+1) G \big)\r) - \l(W\big(\lfloor tG \rfloor \big) - W\big( t G \big)\r) \right\vert  
= O\l( \sqrt{\log(n)} \r) \text{ \ a.s.,}
\\
& \sup_{1 \le t \le \frac{n}{G} - 1 } \left\vert  \int_{\lfloor tG \rfloor}^{\lfloor tG \rfloor + G} ( s - \lfloor tG \rfloor)  \diff W(s) - \int_{tG}^{(t+1)G} (s - tG) \diff W(s) \right\vert 
\\
& \le \sup_{1 \le t \le \frac{n}{G} - 1 } \bigg\{ \Big\vert Y \big( tG \big) - Y \big( \lfloor tG \rfloor \big) \Big\vert + \Big\vert Y \big( (t+1)G \big) - Y \big( \lfloor (t+1)G \rfloor \big) \Big\vert 
\\
& \quad +  \lfloor tG \rfloor \Big\vert W \big( t G \big) -  W \big( \lfloor t G \rfloor \big) \Big\vert + \lfloor tG \rfloor \Big\vert W \big( (t+1) G \big) -  W \big( \lfloor (t+1) G \rfloor \big) \Big\vert
\\
& \quad +  \sqrt{G} \big( t G - \lfloor tG \rfloor\big) \Big\vert \wt{W} \big( t+1 \big) - \wt{W} \big( t  \big) \Big\vert \bigg\} = O \left( n \sqrt{\log(n)} \right) \text{ \ a.s.}
\end{align*}
Then for $k = \lfloor tG \rfloor$, we have
\begin{align} 
\sup_{1 \le t \le \frac{n}{G} - 1} \left\vert \sum_{i= \lfloor tG \rfloor +1}^{\lfloor tG \rfloor + G} \mbf x_{i,k} \epsilon_i - \tau \int_{tG}^{(t+1) G} \begin{bmatrix} 1 \\ (s-tG) / G \end{bmatrix} \diff W(s) \right\vert  = O \left( \frac{n^{1 + \frac{1}{2+\nu}}}{G} + \frac{n { \sqrt{\log(n)}}}{G} \right) \text{ \ a.s.} \nn
\end{align}
Together with~\eqref{eq:cond:G}, we have~\eqref{A2} and~\eqref{A3} is similarly derived.
\end{proof}

\begin{lem} \label{lemA.3}
Assume that~\ref{a:two}--\ref{a:three} hold.
\begin{align}
\frac{1}{\sqrt{G}} \max_{G \le k \le n-G} \left\vert \sum_{i=k+1}^{k+G} \mbf x_{i,k} \epsilon_i\right\vert 
&= O_P\left( \sqrt{\log(n/G)} \r), \label{eq:lem:a3:one} \\
\frac{1}{\sqrt{G}}  \max_{G \le k \le n-G} \left\vert \sum_{i=k-G+1}^{k} \mbf x_{i,k} \epsilon_i\right\vert 
&= O_P\left( \sqrt{\log(n/G)} \r). \label{eq:lem:a3:two} 
\end{align}
\end{lem}
\begin{proof}
By Lemma~\ref{lemA.1},
\begin{align*}
\max_{G \le k \le n-G} \ \left\vert   \sum_{i=k+1}^{k+G} \mbf x_{i,k} \epsilon_i \right\vert 
& = \sup_{1 \le t \le n/G - 1} \left\vert  \tau \int_{t}^{t+1} \begin{bmatrix} 1 \\ (s-t)\end{bmatrix} \diff \wt{W}(s) \right\vert   + o_P \left( \frac{1}{\sqrt{\log (n/G)}} \right) 
\end{align*}
and as for $\wt{Y} (t) = \int_{0}^{t} s \diff \wt{W}(s)$,  elements of 
\begin{align*}  
\int_{t}^{t+1} \begin{bmatrix} 1 \\ (s-t)\end{bmatrix} \diff \wt{W}(s)  = \begin{bmatrix} \wt{W}(t+1) - \wt{W}(t) \\ \wt{Y}(t+1) - \wt{Y}(t) - t \big( \wt{W} (t+1) - \wt{W} (t) \big) \end{bmatrix} 
\end{align*}
are stationary Gaussian processes with zero means and constant variances.
More specifically, $\Var(\wt{W}(t+1) - \wt{W}(t)) = 1$ and
\begin{align*}
\Var\l(\wt{Y}(t+1) - \wt{Y}(t) \r) &= \Var\left( \int_{t}^{t+1} s \diff \wt{W}(s) \right) = t^2 + t + \frac{1}{3}, \quad \text{and} \\
\Cov\l(\wt{Y}(t+1) - \wt{Y}(t) , \wt{W} (t+1) - \wt{W} (t) \r) 
&= \Cov\left( \int_{t}^{t+1} s \diff \wt{W}(s) , \int_{t}^{t+1} \diff \wt{W}(s) \right) = t + \frac{1}{2}
\end{align*}
such that the variance of the second element is $1/3$.
From this,~\eqref{eq:lem:a3:one} follows and~\eqref{eq:lem:a3:two} can similarly be shown.
\end{proof}

\subsection{Proof of Theorem~\ref{thm3.1}}

For the proof of Theorem~\ref{thm3.1}, we first establish that 
$\sqrt{G}/\tau \cdot \b\Sigma^{-1/2}(\wh{\b\beta}^+(k) - \wh{\b\beta}^-(k))$ is well-approximated by a stationary bivariate Gaussian process $\mbf Z(t) = (Z_0(t), Z_1(t))^\top$ with a known covariance structure,
where $k = \lfloor t G \rfloor$ (Proposition~\ref{lem_thm3.1}).
A Wald-type MOSUM statistic has been considered in \cite{kirch2022data} for detecting changes in linear regression where such Gaussian approximation plays a similar role in deriving the asymptotic null distribution.
However, their results are not applicable to our setting due to the presence of the trend under~\eqref{eq:model}, which leads to $Z_0(t)$ and $Z_1(t)$ with non-zero cross-(auto)covariance.
This makes the investigation into the asymptotic null distribution more challenging and introduces the undetermined constant $\log(H)$.
We carefully address this issue and derive the asymptotic null distribution of 
$\sup_{1 \le t \le n/G - 1} \Vert \mbf Z(t) \Vert$ (Proposition~\ref{lem_thm3.2}).


\begin{prop} \label{lem_thm3.1}
Assume that~\ref{a:two}--\ref{a:three} hold and that $G$ fulfils~\eqref{eq:cond:G}. 
Then under $\mc H_0$, 
\begin{align} 
\sup_{1 \le t \le n/G - 1 } \left\vert  \l\Vert \frac{\sqrt{G}}{\tau} \b\Sigma^{-1/2} \l(\wh{\bbeta}^{+}(\lfloor tG \rfloor)  - \wh{\bbeta}^{-}(\lfloor tG \rfloor)\r) \right\Vert -  \left\Vert    \begin{bmatrix} Z_0 (t) \\ Z_1 (t) \end{bmatrix} \right\Vert \right\vert = o_{P} \left( \frac{1}{\sqrt{\log(n/G)}} \right), \nn
\end{align}
where $\mbf Z(t) = (Z_0(t), Z_1(t))^\top$ is a stationary, bivariate Gaussian process with $\E(\mbf Z(t)) = \mbf 0$ and $\Cov(\mbf Z(t)) = \mbf I$ at each $t$, and
\begin{align} 
\Cov (Z_0 (t) , Z_0 (t+h)) &= \left\{\begin{array}{ll} 1 - \frac{9}{2} \vert h\vert + 3\vert h\vert^2 + \frac{3}{4} \vert h\vert^3 & (\vert h\vert < 1) \\ -1 + \frac{7}{2}\vert h\vert - 3\vert h\vert^2 + \frac{3}{4} \vert h\vert^3 & (1 \le \vert h\vert < 2) \\ 0 & (\text{otherwise}), \end{array} \right. \label{cov1}
\\
\Cov (Z_1 (t) , Z_1 (t+h)) &= \left\{\begin{array}{ll} 1 - \frac{3}{2} \vert h\vert - 3\vert h\vert^2 + 3 \vert h\vert^3 & (\vert h\vert < 1) \\ -1 - \frac{3}{2}\vert h\vert + 3\vert h\vert^2 - \vert h\vert^3 & (1 \le \vert h\vert < 2) \\ 0 & (\text{otherwise}), \end{array} \right. 
\quad \text{and}
\label{cov2}
\\
\Cov (Z_0 (t) , Z_1 (t+h)) &= \left\{\begin{array}{ll} - \frac{3\sqrt{3}}{2} h + \frac{9\sqrt{3}}{4} h^2 - \frac{\sqrt{3}}{2} h^3  & (0 \le h < 1) \\ \frac{3\sqrt{3}}{2} h - \frac{7\sqrt{3}}{4} h^2 + \frac{\sqrt{3}}{2} h^3 & (1 \le h < 2)\\ - \frac{3\sqrt{3}}{2} h - \frac{9\sqrt{3}}{4}h^2 - \frac{\sqrt{3}}{2} h^3 & (-1 < h < 0) \\  \frac{3\sqrt{3}}{2} h + \frac{7\sqrt{3}}{4} h^2 + \frac{\sqrt{3}}{2} h^3  & (-2 < h \le -1) \\ 0 & (\text{otherwise}). \end{array} \right. 
\label{cov3}
\end{align} 
\end{prop}

\begin{prop} \label{lem_thm3.2}
Let $M(T) = \sup_{1 \le t \le T} \sqrt{ Z_0^{2} (t) + Z_1^{2} (t)}$
with $Z_0(t)$ and $Z_1(t)$ defined in Proposition~\ref{lem_thm3.1}. 
Then, there exists a constant $H>0$ such that 
we have 
\begin{align*} 
& a_T M(T) - b_T \xrightarrow[n\to\infty]{d} \Gamma_2, \quad \text{where}
\\
& a_T = \sqrt{2\log(T)} , \quad b_T = 2\log(T) + \log\log(T) + \log(H)
\end{align*}
and $\Gamma_2$ is as in Theorem~\ref{thm3.1}.
\end{prop}

\begin{proof}[Proof of Theorem~\ref{thm3.1}]
Combining Propositions~\ref{lem_thm3.1} and~\ref{lem_thm3.2},
we have
\begin{align}
a_G \max_{G \le k \le n - G} W^*_{k, n}(G) - b_G \xrightarrow[n\to\infty]{d} \Gamma_2,
\quad \text{with} \quad 
W^*_{k, n}(G) = \ \frac{\sqrt{G}}{\tau} \l\Vert \b\Sigma^{-1/2} \l( \wh{\b\beta}^+(k) - \wh{\b\beta}^-(k) \r) \r\Vert.
\label{eq:pop:w}
\end{align}
Noting that $W_{k, n}(G) = \tau W_{k, n}^{*}(G) / \wh\tau_k$, 
from~\ref{c:one}, we have
\begin{align*}
\max_{G \le k \le n-G} \l\vert W_{k, n}(G) - W^*_{k, n} (G) \r\vert \le 
\frac{\max_{k} \vert \wh{\tau}_k - \tau\vert }{\sqrt{\tau^2 - \max_{k} \vert\wh{\tau}_k^2 - \tau^2\vert}}
\max_k W_{k,n} (G)  = o_P \left( \frac{1}{\sqrt{\log(n/G)}} \r).
\end{align*}
\end{proof}

\subsubsection{Proof of Proposition~\ref{lem_thm3.1}}

Since $J_n = 0$, we have
\begin{align}
\wh{\bbeta}^{+}(k) - \wh{\bbeta}^{-}(k) = \bC_{G, +}^{-1} \sum_{i=k+1}^{k+G} \mbf x_{i,k} \epsilon_i  - \bC_{G, -}^{-1} \sum_{i=k-G+1}^{k} \mbf x_{i,k} \epsilon_i.
\nn 
\end{align}
Then by Lemmas~\ref{lem:invC},~\ref{lemA.1} and~\ref{lemA.3}, 
for $k = \lfloor tG \rfloor$,
\begin{align*}
\frac{\sqrt{G}}{\tau} \l( \wh{\bbeta}^{+}(k)  - \wh{\bbeta}^{-} (k) \r)
&= \frac{1}{\tau \sqrt{G}} 
\left\{ \l( \bC_{+}^{-1} + O(G^{-1}) \r) \sum_{i=k+1}^{k+G} \mbf x_{i,k}\epsilon_i -  
\l( \bC_{-}^{-1} + O(G^{-1}) \r) \sum_{i=k-G+1}^{k} \mbf x_{i,k}\epsilon_i \right\}  
\\
&= \bC_{+}^{-1} \int_{t}^{t+1} \begin{bmatrix} 1 \\ s - t \end{bmatrix} \diff \wt{W} (s) -
\bC_{-}^{-1} \int_{t-1}^{t} \begin{bmatrix} 1 \\ s - t \end{bmatrix} \diff \wt{W}(s) 
+ R_k
\end{align*}
with $\max_{G \le k \le n-G} \Vert R_k\Vert_2 = O_P(\sqrt{\log(n/G)}/G) = o_P(1/\sqrt{\log(n/G)})$. 
Recalling that $\wt{Y} (t) = \int_{0}^{t} s \diff \wt{W} (s)$, we have
\begin{align*}
& \b\Sigma^{-1/2} \left\{ \bC_{+}^{-1} \int_{t}^{t+1} \begin{bmatrix} 1 \\ s - t \end{bmatrix} \diff \wt{W} (s) - \bC_{-}^{-1} \int_{t-1}^{t} \begin{bmatrix} 1 \\ s - t \end{bmatrix} \diff \wt{W}(s) \right\}
\\
=&   \,
\b\Sigma^{-1/2} \left\{ \begin{bmatrix} 4 & -6 \\ -6 & 12 \end{bmatrix} \begin{bmatrix} \wt{W}(t+1) - \wt{W} (t) \\ \wt{Y} (t+1) - \wt{Y} (t) - t \big( \wt{W} (t+1) - \wt{W} (t) \big)  \end{bmatrix}\right. 
\\& \left. \qquad \qquad \qquad 
- \begin{bmatrix} 4 & 6 \\ 6 & 12 \end{bmatrix}  \begin{bmatrix} \wt{W}(t) - \wt{W} (t-1) \\ \wt{Y} (t) - \wt{Y} (t-1) - t \big( \wt{W} (t) - \wt{W} (t-1) \big)  \end{bmatrix} \right\}
\\
=&   \,
\b\Sigma^{-1/2} \left\{ \begin{bmatrix} 4+6t & -6 \\ -(6+12t) & 12 \end{bmatrix} \begin{bmatrix} \wt{W}(t+1) - \wt{W} (t) \\ \wt{Y} (t+1) - \wt{Y} (t)  \end{bmatrix} -\begin{bmatrix} 4-6t & 6 \\ 6-12t & 6 \end{bmatrix}  \begin{bmatrix} \wt{W}(t) - \wt{W} (t-1) \\ \wt{Y} (t) - \wt{Y} (t-1) \end{bmatrix} \right\}
\\
=: & \, \mbf Z(t) = \begin{bmatrix} Z_0 (t) \\ Z_1 (t) \end{bmatrix},
\end{align*}
it holds that 
\begin{align*} 
\sup_{1 \le t \le \frac{n}{G} - 1 } \left\Vert  
\frac{\sqrt{G}}{\tau} \b\Sigma^{-1/2} \l( \wh{\bbeta}^{+}(\lfloor tG \rfloor)  - \wh{\bbeta}^{-}(\lfloor tG \rfloor) \r) - \begin{bmatrix} Z_0 (t) \\ Z_1 (t) \end{bmatrix} \right\Vert = o_{P} \left( \frac{1}{\sqrt{\log(n/G)}} \right),
\end{align*}
which gives the assertion. 
What remains is to establish~\eqref{cov1}--\eqref{cov3}, which are shown in the following Lemmas \ref{lem_cov1}--\ref{lem_cov_fin}.

\begin{lem} \label{lem_cov1}
Let
\begin{align*}
\Omega(t,h) := \Cov\left( \begin{bmatrix} \wt{W}(t+1+h) - \wt{W} (t+h) \\ \wt{Y} (t+1+h) - \wt{Y} (t+h) \end{bmatrix} , \begin{bmatrix} \wt{W}(t+1) - \wt{W} (t) \\ \wt{Y} (t+1) - \wt{Y} (t) \end{bmatrix}  \right) .
\end{align*}
Then 
\begin{align*}
\Omega(t,h) = \left\{ \begin{array}{ll} \begin{bmatrix} 1-h & (1-h)t + \frac{1-h^2}{2} \\ (1-h)t + \frac{1-h^2}{2} & (1-h)t^2 + (1-h^2) t + \frac{1-h^3}{3} \end{bmatrix} & (0 \le h < 1) \\\begin{bmatrix} 1+h & (1+h)t + \frac{(1+h)^2}{2} \\ (1+h)t + \frac{(1+h)^2}{2} & (1+h)t^2 + (1+h)^2 t + \frac{(1+h)^3}{3} \end{bmatrix} & (-1 < h < 0) \\ \mbf 0 & (|h| \ge 1). \end{array} \right. 
\end{align*}
\end{lem}

\begin{proof}
Let $f(t)$ and $g(t)$ be continuous functions. For $t,s > 0$, simple calculations give
\begin{align} 
&\Cov\left( \int_{0}^{t+1} f(u) \diff \wt{W}(u) - \int_{0}^{t} f(u) \diff \wt{W} (u) , \int_{0}^{s+1} g(u) \diff \wt{W}(u) - \int_{0}^{s} g(u) \diff \wt{W}(u) \right) 
\nn \\
& = \int_{\min( \max(t, s) , \min(t,s)+1 )}^{\min(t,s) +1} f(u)g(u) \diff u.
\nn
\end{align}
Now denote $\Omega (t,h) = \begin{bmatrix} \Omega_{11} (t,h) & \Omega_{12} (t,h) \\ \Omega_{21} (t,h) & \Omega_{22} (t,h) \end{bmatrix}$. Then,
\begin{align*} 
\Omega_{11}(t,h) &= \left\{ \begin{array}{ll} \int_{t+h}^{t+1}  \diff u & ( 0 \le h < 1) \\ \int_{t}^{t+1+h}  \diff u & (-1 < h < 0) \\ 0 & (|h| \ge 1) \end{array} \right. = \left\{ \begin{array}{ll} 1-h & ( 0 \le h < 1) \\ 1+h & (-1 < h < 0) \\ 0 & (|h| \ge 1) ,\end{array} \right.
\\
\Omega_{12}(t,h) &= \left\{ \begin{array}{ll} \int_{t+h}^{t+1} u \diff u & ( 0 \le h < 1) \\ \int_{t}^{t+1+h} u \diff u & (-1 < h < 0) \\ 0 & (|h| \ge 1) \end{array} \right. = \left\{ \begin{array}{ll} (1-h)t + \frac{1-h^2}{2} & ( 0 \le h < 1) \\ (1+h)t + \frac{(1+h)^2}{2} & (-1 < h < 0) \\ 0 & (|h| \ge 1) ,\end{array} \right.
\\
\Omega_{21}(t,h) &= \left\{ \begin{array}{ll} \int_{t+h}^{t+1} u \diff u & ( 0 \le h < 1) \\ \int_{t}^{t+1+h} u \diff u & (-1 < h < 0) \\ 0 & (|h| \ge 1) \end{array} \right. = \left\{ \begin{array}{ll} (1-h)t + \frac{1-h^2}{2} & ( 0 \le h < 1) \\ (1+h)t + \frac{(1+h)^2}{2} & (-1 < h < 0) \\ 0 & (|h| \ge 1) ,\end{array} \right.
\\
\Omega_{22}(t,h) &= \left\{ \begin{array}{ll} \int_{t+h}^{t+1} u^2 \diff u & ( 0 \le h < 1) \\ \int_{t}^{t+1+h} u^2 \diff u & (-1 < h < 0) \\ 0 & (|h| \ge 1) \end{array} \right. = \left\{ \begin{array}{ll} (1-h)t^2 + (1-h^2)t + \frac{1-h^3}{3} & ( 0 \le h < 1) \\ (1+h)t^2 + (1+h)^2 t + \frac{(1+h)^3}{3} & (-1 < h < 0) \\ 0 & (|h| \ge 1) .\end{array} \right.\end{align*}
\end{proof}

\begin{lem} \label{lem_cov_fin}
The Gaussian process $\mbf Z(t)$ satisfies~\eqref{cov1}--\eqref{cov3}.
\end{lem}

\begin{proof}
For notational simplicity, we define
\begin{align*} 
\bA(t) = \b\Sigma^{-1/2} \begin{bmatrix} 4+6t & -6 \\ -(6+12t) & 12 \end{bmatrix}, ~ \bB(t) = \b\Sigma^{-1/2} \begin{bmatrix} 4-6t & 6 \\ 6-12t & 12 \end{bmatrix}
\text{ and } \bV(t) = \begin{bmatrix} \tilde{W}(t+1) - \tilde{W}(t) \\ \tilde{Y}(t+1) - \tilde{Y}(t) \end{bmatrix}.
\end{align*}
Then $\mbf Z(t) = \bA(t) \bV(t) - \bB(t) \bV(t-1)$ and 
by Lemma~\ref{lem_cov1}, we have $\Cov( \bV(t+h), \bV(t) ) = \Omega(t,h)$.
Therefore,
\begin{align} 
& \Cov\left( \begin{bmatrix} Z_0 (t+h) \\ Z_1 (t+h) \end{bmatrix}, \begin{bmatrix} Z_0 (t) \\ Z_1 (t) \end{bmatrix}\right) 
= \begin{bmatrix} \Cov(Z_0 (t+h), Z_0 (t) ) & \Cov(Z_0 (t+h), Z_1(t) ) \\ \Cov(Z_1(t+h), Z_0(t) ) & \Cov(Z_1(t+h), Z_1(t) ) \end{bmatrix}
\nn
\\
&= \Cov\big(\bA(t+h) \bV(t+h) - \bB(t+h) \bV(t-1+h),\bA(t) \bV(t) - \bB(t) \bV(t-1) \big) 
\nn
\\
& = \bA(t+h) \Omega(t,h) \bA(t)^{\top}  - \bB(t+h) \Omega(t, -1+h) \bA(t)^{\top} 
\nn
\\
& - \bA(t+h) \Omega (t-1, 1+h) \bB(t)^{\top} + \bB(t+h) \Omega(t-1, h) \bB(t)^{\top}.
\label{eq_cov_main}
\end{align}

\paragraph{Case 1: $0 \le h <1$.} 

Then as $-1 \le -1+h < 0$ and $1 \le 1+h < 2$, 
\begin{align*}\Omega (t, -1+h) = \begin{bmatrix} h & ht + \frac{h^2}{2} \\ ht + \frac{h^2}{2} & ht^2 + h^2 t + \frac{h^3}{3} \end{bmatrix}, ~ \Omega (t-1, 1+h) = \begin{bmatrix} 0 & 0 \\ 0 & 0 \end{bmatrix} , \end{align*}
and hence 
\begin{align} \label{eq_cov_case1}
 \Cov\left( \begin{bmatrix} Z_0 (t+h) \\ Z_1 (t+h) \end{bmatrix}, \begin{bmatrix} Z_0 (t) \\ Z_1 (t) \end{bmatrix}\right) = \begin{bmatrix} \frac{6h^3 + 24h^2 - 36h + 8 }{8} & \frac{6h (2h^2 - 9h + 6)}{8 \sqrt{3}} \\ \frac{-6h (2h^2 - 9h + 6)}{8 \sqrt{3}} & \frac{12 (6h^3 - 6h^2 - 3h + 2)}{24} \end{bmatrix} .
\end{align}

\paragraph{Case 2: $1 \le h <2$.}

Then as $0 \le -1 + h < 1$ and $1+h \ge 2$, 
\begin{align*}\Omega (t, -1+h) = \begin{bmatrix} 2-h & (2-h)t + h - \frac{h^2}{2} \\ (2-h)t + h - \frac{h^2}{2} & (2-h)t^2 - (2h - h^2 ) t + \frac{2 - h^3}{3} -h+h^2 \end{bmatrix} \end{align*} 
and $\Omega (t, h)$, $\Omega (t-1, 1+h)$, $\Omega (t-1, h)$ are all the zero matrices. Therefore 
\begin{align} \label{eq_cov_case2}
 \Cov\left( \begin{bmatrix} Z_0 (t+h) \\ Z_1 (t+h) \end{bmatrix}, \begin{bmatrix} Z_0 (t) \\ Z_1 (t) \end{bmatrix}\right) = \begin{bmatrix} \frac{6h^3 - 24h^2 + 28h - 8 }{8} & \frac{-2h (6h^2 - 21h + 18)}{8 \sqrt{3}} \\ \frac{2h (6h^2 - 21h + 18)}{8 \sqrt{3}} & \frac{-12 (2h^3 - 6h^2 + 3h + 2)}{24} \end{bmatrix} .
\end{align}

\paragraph{Case 3: $-1 \le h <0$.}

Then as $0 \le 1+h < 1$ and $-1+h < 1$, $\Omega (t, -1+h)$ is a zero matrix and 
\begin{align*} \Omega (t-1, 1+h) = \begin{bmatrix} -h & -ht - \frac{h^2}{2} \\ -ht - \frac{h^2}{2} & -ht^2 - h^2 t - \frac{h^3}{3} \end{bmatrix}.\end{align*} 
Therefore we get 
\begin{align} \label{eq_cov_case3}
 \Cov\left( \begin{bmatrix} Z_0 (t+h) \\ Z_1 (t+h) \end{bmatrix}, \begin{bmatrix} Z_0 (t) \\ Z_1 (t) \end{bmatrix}\right) = \begin{bmatrix} \frac{-6h^3 + 24h^2 + 36h + 8 }{8} & \frac{6h (2h^2 + 9h + 6)}{8 \sqrt{3}} \\ \frac{-6h (2h^2 + 9h + 6)}{8 \sqrt{3}} & \frac{12 (-6h^3 - 6h^2 + 3h + 2)}{24} \end{bmatrix} .
\end{align}

\paragraph{Case 4: $-2 \le h < -1$.}

Then as $-1 + h < -2$ and $-1 \le 1+h < 0$, 
\begin{align*}\Omega (t-1, 1+h) = \begin{bmatrix} 2+h & (2+h)t + h + \frac{h^2}{2} \\ (2+h)t + h + \frac{h^2}{2} & (2+h)t^2 + (2h + h^2 ) t + \frac{2 + h^3}{3} + h+h^2 \end{bmatrix} \end{align*} 
and $\Omega (t, h)$, $\Omega (t-1, 1+h)$, $\Omega (t, -1+h)$ are all the zero matrices. Therefore 
\begin{align} \label{eq_cov_case4}
 \Cov\left( \begin{bmatrix} Z_0 (t+h) \\ Z_1 (t+h) \end{bmatrix}, \begin{bmatrix} Z_0 (t) \\ Z_1 (t) \end{bmatrix}\right) = \begin{bmatrix} \frac{-6h^3 - 24h^2 - 28h - 8 }{8} & \frac{-2h (6h^2 + 21h + 18)}{8 \sqrt{3}} \\ \frac{2h (6h^2 + 21h + 18)}{8 \sqrt{3}} & \frac{12 (2h^3 + 6h^2 + 3h - 2)}{24} \end{bmatrix} .
\end{align}

\paragraph{Case 5: $|h| \ge 2$.}

In this case, $\Omega(t,h)$, $\Omega(t-1, 1+h)$, $\Omega(t, -1+h)$, $\Omega(t-1,h)$ are all zero matrices. Hence 
\begin{align} \label{eq_cov_case5}
 \Cov\left( \begin{bmatrix} Z_0 (t+h) \\ Z_1 (t+h) \end{bmatrix}, \begin{bmatrix} Z_0 (t) \\ Z_1 (t) \end{bmatrix}\right) = \begin{bmatrix}0 & 0 \\ 0 & 0 \end{bmatrix} .
\end{align}
Combining \eqref{eq_cov_case1}--\eqref{eq_cov_case5} 
with~\eqref{eq_cov_main} gives~\eqref{cov1}--\eqref{cov3}.
\end{proof}

\subsubsection{Proof of Proposition~\ref{lem_thm3.2}}

Theorem~10 of \cite{albin1990} gives that their Conditions A($\{0\}$), B, C$^0 (\{0\})$, D(0), D$^{\prime}$ and Equation~(2.15) imply
\begin{align}
\label{eq:albin}
\lim_{u \to \infty} \Pr \left( w(u)^{-1} \left(\sup_{t \in [0, T(u)]} \Vert \mbf Z(t) \Vert - u \right) \le x \right) = \exp \big( - (1-F(x))^{1-c} \big) 
\end{align}
where $T(u) \sim \frac{q(u)}{H (1 - G(u))}$ as $u \to \infty$
(hereafter, $a(u) \sim b(u)$ means $a(u)/b(u) \to 1$ as $u \to \infty$), 
$q(u) = u^{-2}$, $G(u)$ is a marginal distribution function of $\Vert \mbf Z(t) \Vert$, and $w(u)$, $F$, $c \in [0,1]$ are defined in the conditions. 
We first show that all the required conditions hold with $w(u) = u^{-1}$, $F(x) = 1- e^{-x}$ and $c=0$, and then apply the arguments adopted in \cite{steinebach1996} to complete the proof.

Condition~A$(\{0\})$ is implied by Equations~(2.1), (2.12) and (2.13) of \cite{albin1990} by Theorem~3 therein. As marginal density $g$ of $\Vert \mbf Z(t) \Vert$ is $g(x) = x e^{-x^2/2}$, $G$ satisfies their~(2.1) and (2.12) with $w(u) = u^{-1}$ and $F(x) = 1-e^{-x}$ (also refer to Theorem~9 of \cite{albin1990}). 
Next, assume that Condition~B has been shown. Then Condition~C$^{0}(\{0\})$ is implied by Condition~A$(\{0\})$, B, and C by Theorem~2~(c) of \cite{albin1990}, and as Condition~C is achieved by proving (2.23) by Theorem~6 of \cite{albin1990}, proving (2.23) with Condition~B gives Condition~C$^{0}(\{0\})$.
Condition~D(0) and D$^{\prime}$ follow immediately as auto- and cross-covariance functions of $\mbf Z(t)$ are locally supported. Finally, we can show that~(2.15) of \cite{albin1990} holds
from that
\begin{align*} 
\frac{q(u+x/u)}{q(u)} = \frac{u}{u + x/u} \to 1 \text{ as } u \to \infty .
\end{align*}
Therefore, combining all above, it remains to verify Equation~(2.13), Condition~B, and Equation~(2.23).

\paragraph{Equation~(2.13):}
Let $\omega(t) = \Vert \mbf Z(t) \Vert$.
We show that there exist a random sequence $\{\eta_{a,x} (k)\}_{k=1}^{\infty}$ and 
such that, for all $a>0$ and all~$N$, 
\begin{align*}
\left( u\l(\omega(aq)-u\r), \ldots, u\l(\omega(aqN)-u\r) \left\vert \r. u\l(\omega(0) - u\r) = x \right) \xrightarrow[u \to \infty]{d} (\eta_{a,x} (1), \ldots, \eta_{a,x} (N))
\end{align*}
for almost all $x>0$, where $q = q(u) = u^{-2}$.
Let $r_i (h) = \Cov(Z_i(t), Z_i (t+h))$ for $i=0,1$ and
\begin{align*}
\b\Sigma (qt) = \Cov (\mbf Z (qt), \mbf Z (0) ) = \bmx 1 - \frac{9}{2} qt & - \frac{3\sqrt{3}}{2} qt \\ \frac{3\sqrt{3}}{2} qt & 1-\frac{3}{2} qt  \emx + o(qt) =:  \bmx 1 - C_0 qt & - C_{01} qt \\ - C_{10} qt & 1-C_1 qt  \emx + o(qt).
\end{align*}
Also define $\mbf Z^{u} (t) = u ( \mbf Z (qt) - \b\Sigma(qt) \mbf Z (0))$ and $\Delta_i (t) = u \big( u + \frac{x}{u} \big) (1-r_i (qt))$, and denote by $\kappa^{2}$  the one-dimensional Hausdorff measure over $\bbR^{2}$. 
If we write $u(\omega(0) - u) = x$ as
$\mbf Z(0) = \left( u + \frac{x}{u} \right) {\mbf x}$ with some ${\mbf x} = (x_0, x_1)$ satisfying $\Vert \mbf x \Vert = 1$, we have $Z_i (qt) = Z_i (0) + o(1)$ as $u \to \infty$ (and hence $q \to 0$) and 
\begin{align*}
\omega(aqk) - \omega(0)
& = \frac{\sum_{i=0}^{1} ( Z_i^{2} (aqk) - Z_i^{2} (0))}{\omega(aqk) + \omega(0)}
=  \frac{\sum_{i=0}^{1} (Z_i (aqk) - Z_i (0))(2Z_i (0) + o(1))}{2 \omega(0) + o(1)} 
\\& = \frac{\sum_{i=0}^{1} \left( u + \frac{x}{u} \right) x_i \big( Z_i (aqk) - Z_i (0) \big) + o(1)}{u + \frac{x}{u} + o(1) } 
= \sum_{i=0}^{1} x_i \big( Z_i (aqk) - Z_i (0) \big) + o(u^{-1}).
\end{align*}
With $\cdot$ denoting the inner product, we have
\begin{align*}
&\Pr \left(\left. \bigcap_{k=1}^{N} \left\{ u\big(\omega(aqk) - u \big) \le z_k \right\} \right\vert u \big( \omega(0) - u \big) = x\right) 
\\
=& \int_{\Vert{\bx}\Vert=1} \Pr \left(\left. \bigcap_{k=1}^{N} \left\{ u\big(\omega(aqk) - \omega(0) \big) + x \le z_k \right\} 
\right\vert  \mbf Z(0) = \left( u + \frac{x}{u} \right) {\bx}\right) \frac{(2\pi)^{-1} \left( u + \frac{x}{u} \right)}{g \left( u + \frac{x}{u} \right) \exp \left( \frac{1}{2} \left( u + \frac{x}{u} \right)^{2} \right) } \diff \kappa^{2} ({\bx})
\\
=& \frac{1 }{2\pi} \int_{\Vert{\bx}\Vert=1} \Pr \left(\left. \bigcap_{k=1}^{N} \left\{ \sum_{i=0}^{1} u x_i \big( Z_i (aqk) - Z_i (0) \big) + x + o(1) \le z_k \right\} 
\right\vert  \mbf Z(0) = \left( u + \frac{x}{u} \right) {\bx}\right) \diff \kappa^{2} ({\bx})
\\
=&   \frac{1 }{2\pi} \int_{\Vert{\bx}\Vert=1} \Pr \left(\left. \bigcap_{k=1}^{N} \left\{ {\bx} \cdot \mbf Z^{u} (aqk) - u {\bx} \cdot (\mbf I - \b\Sigma (aqk)) \mbf Z(0) + x + o(1) \le z_k \right\} 
\right\vert  \mbf Z(0) = \left( u + \frac{x}{u} \right) {\bx}\right) \diff \kappa^{2} ({\bx})
\\
=&   \frac{1 }{2\pi} \int_{\Vert{\bx}\Vert=1} \Pr \left(\left. \bigcap_{k=1}^{N} \left\{ {\bx} \cdot \mbf Z^{u} (aqk) - (u^2 + x)  {\bx} \cdot (\mbf I - \b\Sigma (aqk))\bx  + x + o(1) \le z_k \right\} 
\right\vert  \mbf Z(0) = \left( u + \frac{x}{u} \right) {\bx}\right) \diff \kappa^{2} ({\bx})
\\
=&   \frac{1 }{2\pi} \int_{\Vert{\bx}\Vert=1} \Pr \left(\bigcap_{k=1}^{N} \left\{ {\bx} \cdot \mbf Z^{u} (aqk) - (u^2 + x) \sum_{i=0}^{1} C_i aqk x_i^2 + x + o(1) \le z_k \right\} \right) \diff \kappa^{2} ({\bx})
\end{align*}
from the independence between $(\mbf Z^{u} (t_k))_{k=1}^{N}$ and $\mbf Z(0)$. 
Noting that
\begin{align*}
\Cov (\mbf Z^{u} (t), \mbf Z^{u} (s)) 
& = u^2 \Cov \big( \mbf Z(qt) - \b\Sigma (qt) \mbf Z(0), \mbf Z(qs) - \b\Sigma (qs) \mbf Z(0) \big)
\\& = u^2 ( \b\Sigma (q(t-s)) - \b\Sigma (qt) \b\Sigma(qs)^{\top}) 
= (t+s-\vert t-s\vert) \begin{bmatrix} \frac{9}{2} & 0 \\ 0 & \frac{3}{2}  \end{bmatrix}  + o(1),
\end{align*}
the normally distributed $( \mbf Z^{u} (ak))_{k=1}^{N}$ converges to
multivariate normal random vectors $\big( \bzeta (t_k)\big)_{k=1}^{N}$
with $\b\zeta(t_k) = (\zeta_0 (t_k), \zeta_1 (t_k))^{\top}$,
where $\zeta_0 (t)$ and $\zeta_1(t)$ are independent zero-mean Gaussian processes with $\Cov (\zeta_i (s) , \zeta_i (t) ) = C_i (t + s - \vert t-s\vert )$.
Finally, noting that $(u^2+x) C_i aqk  \to C_i ak$,
we have the finite dimensional distributions of $\left\{ \big( u(\omega(qt) - u)\vert u(\omega(0) - u)=x\big)\right\}_{t > 0}$ converge weakly to those of $\{\eta(t)\}_{t > 0}$, where 
\begin{align*}
\Pr \left( \bigcap_{k=1}^{N} \l\{ \eta(aqk) \le z_k \r\} \right) =  \frac{1}{2\pi} \int_{\vert\mbf x\vert=1} \Pr \left(\bigcap_{k=1}^{N} \left\{ \mbf x \cdot {\bzeta} (ak) - \sum_{i=1}^{m} C_ix_i^2 ak+ x  \le z_k \right\} \right)\diff \kappa^{2} (\mbf x).
\end{align*}

\paragraph{Condition~B:}
We show that
\begin{align*}
\displaystyle \limsup_{u \to \infty} \sum_{k=N}^{[h/(aq)]} \Pr \big( \omega(aqk)>u \vert \omega(0) >u \big) \to 0
\end{align*}
when $N \to \infty$, for all fixed $a>0$ and some $h>0$. 
Let $\b\Sigma (t) $ have singular values $\lambda_1 (t) \ge \lambda_2 (t) > 0$. 
Choose constants $A, B, \epsilon>0$ such that $At \le 1 - \lambda_i (t) \le Bt$ for $0 < t \le \epsilon$, and let $R = \lambda_1(qt) $, $r = \lambda_2(qt)$ and $\mbf Z_{\bar{r}} (qt) = \b\Sigma (qt) \mbf Z (0)$.
Then by the triangle inequality, for $\omega(qt) > \omega (0) > u$, 
\begin{align*}
u(1-R) < \omega(qt) - R \omega(0) \le \omega(qt) - \Vert\mbf Z_{\bar{r}}(qt)\Vert \le \Vert\mbf Z(qt) - \mbf Z_{\bar{r}}(qt)\Vert.
\end{align*}
Also note that for any vector $\bv$, 
we have $\bv^{\top} (\mbf I - \b\Sigma (qt) \b\Sigma(qt)^{\top} ) \bv {\le (1 - r^2) \bv^{\top} \bv}$. Hence using symmetry and Boole's inequality, 
\begin{align*}
& \Pr \big( \omega(qt) > u, \, \omega(0) > u \big)
= 2 \Pr \big( \omega(qt) > \omega(0) > u \big)
\\
& \le 2\Pr \big( \omega(0) > u , \Vert\mbf Z(qt) - \mbf Z_{\bar{r}}(qt)\Vert > u(1-R)\big)
\\
& =  2\Pr \big( \omega(0) > u \big) \Pr \big( \Vert\mbf Z(qt) - \mbf Z_{\bar{r}}(qt)\Vert > u(1-R) \, \big\vert \, \omega(0) > u \big)
\\
& \le 2 (1-G(u)) \Pr \left( \sqrt{\chi^2_2} > \frac{(1-R)u}{\sqrt{1-r^2}} \right) 
\le 2 (1-G(u)) \sum_{i=0}^{1} \Pr \left( \vert Z_i\vert > {\frac{(1-R)u}{\sqrt{2(1 - r^2)}}} \right) 
\\
& \le 8 (1-G(u)) \left( 1 - \Phi \left( \frac{(1-R)u}{\sqrt{2 (1-r^{2})}} \right) \right)
\end{align*}
where with a slight abuse of notation, we denote by $\chi^2_2$ a random variable following the chi-squared distribution, 
and $Z_i, \, i = 0, 1$, denote independent standard normal random variables. 
Then, 
\begin{align*}
& \frac{1}{1-G(u)} \Pr \big( \omega(qt) > u, \omega(0) > u \big) \le 8 \big( 1 - \Phi ( A(4B)^{-1/2} t^{1/2}) \big) \text{ for } 0 < qt \le \epsilon,
\\
& \frac{1}{1-G(u)} \Pr \big( \omega(qt) > u, \omega(0) > u \big) \le 8 \big( 1 - \Phi ( \lambda 2^{-1/2} u) \big) \text{ for } \epsilon < qt \le h,
\end{align*}
where $\lambda = 1- \sup_{s \in (\epsilon, h]} \max (\lambda_1 (s), \lambda_2(s))$.
From this, Condition~B follows since
\begin{align*}
& \sum_{k=N}^{\l\lfloor \frac{h}{aq} \r\rfloor} \Pr \big( \omega(aqk) > u \big\vert  \omega(0) > u \big) 
\\
& \le 8 \left\{ \sum_{k=N}^{\l\lfloor \frac{\epsilon}{aq} \r\rfloor} \left( 1 - \Phi \left( \frac{A}{\sqrt{4B}} (ak)^{1/2} \right) \right) + \left( \l\lfloor \frac{h}{aq} \r\rfloor - \l\lfloor \frac{\epsilon}{aq} \r\rfloor \right) \left( 1 - \Phi \left( \frac{\lambda}{\sqrt{2}} u \right) \right)   \right\} 
\\& \le 8 \left( \int_{N}^{\infty} \frac{\frac{1}{\sqrt{2\pi}} \exp \left( - \frac{A^2}{8B} (ax) \right)}{\frac{A}{\sqrt{4B}} (ax)^{1/2}} \diff x +\left( \l\lfloor \frac{h}{aq} \r\rfloor - \l\lfloor \frac{\epsilon}{aq} \r\rfloor \right) \frac{\frac{1}{\sqrt{2\pi}} \exp \left( - \frac{\lambda^2}{4} u^2\right)}{\frac{\lambda}{\sqrt{2}} u} \right)
\xrightarrow[N \to \infty]{} 0.
\end{align*}

\paragraph{Equation~(2.23):}
We show that there exist constants $\lambda_0$, $\rho$, $e$, $C$,  $\delta_0 >0$, $u_0 > 0$ and $d>1$ such that 
\begin{align*}
\Pr \big( \omega(qt) - \omega(0) > \lambda w(u) ,  \omega(0) \le u + \delta_0 w(u) \big\vert \omega(qt)  > u \big) \le Ct^{d} \lambda^{-e}
\end{align*}
for all $0 < t^{\rho} < \lambda < \lambda_0$ and $u > u_0 $.

Set $\delta \in (0,1)$ and the constants $B$, $\epsilon>0$ such that $\lambda_i (t) > \frac{1}{2}$ and $1-\lambda_i (t) \le Bt$ for $0 < t \le \epsilon$, for $i=1,2$, and let $\lambda_0 = \min ( \epsilon^{1/2}, 1/(8B))$. 
Further, let $u_0=1$ and $\mbf Z_{\bar{r}}^{\prime} (qt) = \b\Sigma(qt)^{-1} \mbf Z(0)$. 
Then we have, for $0<t^{1/2} < \lambda < \lambda_0$, $u>u_0$ and $\omega(0) \le u + \delta / u$, both $t^{1/2} < \epsilon^{1/2}$ and $t^{1/2} < (8B)^{-1}$ hold, which implies $qt < \epsilon$ and therefore 
\begin{align*}
\Vert\mbf Z_{\bar{r}}^{\prime} (qt)\Vert - \omega(0) \le (r(qt)^{-1} -1) \omega(0) \le 2B(qt) u (1+u^{-2} \delta) \le 4B \frac{1}{u} t  < \frac{4B}{u} \frac{\lambda}{8B} = \frac{\lambda}{2u}. 
\end{align*}
Hence we obtain, by the triangle inequality and by arguments adopted in the verification of Condition~B, 
\begin{align*}
&\Pr \left( \omega(qt) - \omega(0) > \frac{\lambda}{u} , \, \omega(0) \le u + \frac{\delta}{u} , \, \omega(qt) > u \right) 
\le \Pr \left( \omega(qt) > u, \, \Vert\mbf Z(qt)  - \mbf Z_{\bar{r}} ^{\prime}(qt)\Vert > \frac{\lambda}{2u}  \right)  
\\& \le 4 (1-G(u)) \left( 1 - \Phi \left( \frac{\lambda}{8 \sqrt{Bt} } \right) \right). 
\end{align*}
Hence for each constant $p \ge 1$, there exists a corresponding constant $K_p >0$ such that 
\begin{align*} 
4 (1-\Phi(x)) \le \frac{4}{\sqrt{ 2\pi}} \frac{1}{x} \exp \left( - \frac{x^2}{2} \right) \le K_p x^{-p} , ~ x>0 ,  
\end{align*}
and the claim follows.

Now we show the main claim of the proposition using~\eqref{eq:albin}. 
The remainder of the proof proceeds analogously as in \cite{steinebach1996}. 
As $u \to \infty$, 
\begin{align*} 
\frac{q(u)}{H(1-G(u))} \sim \frac{1}{2H}  \left( \frac{u^2}{2} \right)^{-1} \exp \left( \frac{u^2}{2} \right).
\end{align*}
We set $T(u) = K (u^2/2)^{-1} \exp(u^2/2)$, where $K = (2 H)^{-1}$. For $v$ sufficiently large, if $T(u) = v$, we have
$\log(v) = u^2 / 2 - \log (u^2/2) + \log(K)$, i.e.\
$u^2  = 2\log(v) + 2  \log (u^2/2) - \log(K)$,
which implies 
\begin{align*}
w(u)^{-1} = u \sim (2 \log(v) )^{1/2} = a_v 
\quad \text{and} \quad
\log (u^2/2) = \log \log(v) + o(1) \text{ \ as \ } v \to \infty.
\end{align*}
Furthermore, these statements imply 
\begin{align*}
u &= \big( 2\log(v) + 2  \log (u^2/2) - 2 \log(K) \big)^{1/2} 
\\&= \big( 2\log(v) + 2 \log \log(v) - 2\log(K) + o(1) \big)^{1/2} 
\\& = \big( 2 \log(v) \big)^{1/2} + \big(2 \log(v) \big)^{-1/2} (\log \log(v) - \log(K) + o(1) ), \quad \text{so that}
\\
w(u)^{-1} u = u^2 &= 2 \log(v) +  \log \log(v) - \log(K) + o(1) 
= b_{v} + \log(2) + o(1),
\end{align*}
and therefore we have for $x \in \bbR$,
\begin{align*}
\lim_{v \to \infty} \Pr \left( a_v M(v) - b_v - \log(2) \le x - \log(2) \right) = \exp \big( -e^{-(x- \log(2))} \big) = \exp (-2e^{-x}). 
\end{align*}

 \subsection{Proof of Theorem~\ref{5.1}} 
\label{appendA.3}

For ease of presentation, we introduce the following notations: For given $k$, 
let $\beta_{0, j}(k) = \alpha_{0, j} + \alpha_{1, j} t_k$ and 
define $\beta_{1, j} = G \Delta t \alpha_{1, j}$
and $\b\beta_j(k) = (\beta_{0, j}(k), \beta_{1, j})^\top$.
Then, $\Delta^{(0)}_j= \beta_{1, j+1}(k_j) - \beta_{1, j}(k_j)$ and
$\Delta^{(1)}_j = \beta_{1, j+1} - \beta_{1, j}$.

\subsubsection{Proof of Theorem~\ref{5.1}~\ref{thm:est:one}}

We need the following lemma for the proof of Theorem~\ref{5.1}~\ref{thm:est:one}.

\begin{lem} 
\label{lemA.4}
Assume that~\ref{b:one} holds.
Let $k \in \{k_j - G + 1, \ldots, k_j + G\}$.
Then,
\begin{align}
& \wh{\bbeta}^{+}(k) - \wh{\bbeta}^{-}(k) = \big(\mbf A(\kappa) + \mbf O_G (\kappa) \big) \bDelta_j  + \bzeta_{k}, \quad \text{where}
\nn \\
& \bzeta_{k, +} = \mbf C_{G, +}^{-1} \sum_{i = k + 1}^{k + G} \mbf x_{i, k} \epsilon_i,
\
\bzeta_{k,-} = \mbf C_{G, -}^{-1} \sum_{i = k - G + 1}^{k} \mbf x_{i, k} \epsilon_i
\text{ \ and \ }
\bzeta_{k} = \bzeta_{k,+} - \bzeta_{k,-}
\label{eq:zeta}
\end{align}
where $\kappa = \vert k - k_j\vert / G$, 
and when $k \le k_j$,
\begin{align*}
& \mbf A(\kappa) = (1-\kappa) \bmx 1  - 3\kappa & \kappa^2 - \kappa \\ 6 \kappa & - 2\kappa^2 + \kappa + 1 \emx
\text{ \ and \ }
\mbf O_G(\kappa) = -\frac{\kappa(1-\kappa)}{G-1} \bmx 3 & 2 -\kappa \\ - \frac{6}{G+1} & \frac{2\kappa - 1 - 3G}{G+1} \emx,
\end{align*}
while when $k > k_j$,
\begin{align*}
& \mbf A(\kappa) = (1-\kappa) \bmx 1  - 3\kappa & \kappa - \kappa^2 \\ -6 \kappa & - 2\kappa^2 + \kappa + 1  \emx
\text{ \ and \ }
\mbf O_G(\kappa) =  - \frac{\kappa(1-\kappa)}{G+1} \bmx -3 & 2-\kappa \\ \frac{6}{G-1} & \frac{2\kappa - 1 + 3G}{G-1} \emx.
\end{align*}
\end{lem}

\begin{proof} Note that
\begin{align}
\label{eq:reparam}
X_i = \left\{ \begin{array}{ll} 
\alpha_{0, j} + \alpha_{1, j} t_i + \epsilon_i = \mbf x_{i,k_j}^{\top} \bbeta_{j}(k_j)+ \epsilon_i & (k - G +1 \le i \le k_j) 
\\ 
\alpha_{0,j+1} + \alpha_{1,j+1} t_i + \epsilon_i = \mbf x_{i, k_j}^{\top} \bbeta_{j+1}(k_j)+ \epsilon_i& ( k_j + 1 \le i \le k+G).
\end{array} 
\right. 
\end{align}
We consider the two cases, when $k \le k_j$ and $k > k_j$.
\paragraph{Case 1: $k \le k_j < k + G$.}
From~\eqref{eq:reparam},
\begin{align*}
\wh{\bbeta}^{-}(k) = & ( \mbf C_{G, -} )^{-1} 
\sum_{i=-G+1}^{0} \bmx 1 \\ \frac{i}{G} \emx \left( \beta_{0,j}(k_j)+ \frac{k-k_j+i}{G} \beta_{1,j}\right) + \bzeta_{k,-}
\\ = &   ( \mbf C_{G, -} )^{-1} \sum_{i=-G+1}^{0} \bmx 1 \\ \frac{i}{G} \emx\bmx 1 & \frac{i}{G} \emx \bmx \beta_{0,j}(k_j)+ \frac{k-k_j}{G}  \beta_{1,j} \\  \beta_{1,j}\emx + \bzeta_{k,-} 
\\= & \bmx \beta_{0,j}(k_j)+ \frac{k-k_j}{G}  \beta_{1,j}\\  \beta_{1,j}\emx + \bzeta_{k,-}, \quad \text{and}
\\
\wh{\bbeta}^{+}(k) = & 
(\mbf C_{G, +})^{-1} \left( 
\sum_{i=1}^{k_j-k} \bmx 1 \\ \frac{i}{G} \emx  \bmx 1 & \frac{i}{G} \emx \bmx \beta_{0,j}(k_j)+ \frac{k-k_j}{G}  \beta_{1,j} \\  \beta_{1,j}\emx   \right. 
\\
& \left. + \sum_{i=k_j-k+1}^{G} \bmx 1 \\ \frac{i}{G} \emx \bmx 1 & \frac{i}{G} \emx \bmx \beta_{0,j+1}(k_j)+ \frac{k-k_j}{G}  \beta_{1,j+1} \\  \beta_{1,j+1}\emx  \right)
+ \bzeta_{k,+}
\\
=& \bmx \beta_{0,j}(k_j)+ \frac{k-k_j}{G}  \beta_{1,j}\\  \beta_{1,j}\emx + \bzeta_{k,+}
\\
& +  (\mbf C_{G, +})^{-1} \left( \sum_{i=k_j-k+1}^{G} \bmx 1 \\ \frac{i}{G} \emx \bmx 1 & \frac{i}{G} \emx \bmx \Delta^{(0)}_j+ \frac{k-k_j}{G}  \Delta^{(1)}_j  \\  \Delta^{(1)}_j \emx  \right)
\\
= & \bmx \beta_{0,j}(k_j)+ \frac{k-k_j}{G}  \beta_{1,j}\\  \beta_{1,j}\emx + \bzeta_{k,+} 
\\
& + \l\{ \mbf I - \frac{\vert k-k_j\vert}{G} \l( \bC_{+} + \bD_{G,+}\r)^{-1} \l( \bC_{\kappa,+} + \bD_{\kappa,G,+} \r) \r\} \bmx \Delta^{(0)}_j+ \frac{k-k_j}{G}  \Delta^{(1)}_j  \\  \Delta^{(1)}_j \emx ,
\end{align*}
where $\bC_{+}$ and $\bD_{G,+}$ are defined as in~\eqref{def_C+}, and 
\begin{align*} 
& \bC_{\kappa,+} 
= \bmx 1&  \frac{\kappa}{2}  \\  \frac{\kappa}{2} & \frac{\kappa^2  }{3}  \emx 
\quad \text{and} \quad
\bD_{\kappa, G, +} =  \bmx 0 &  \frac{1}{2G} \\  \frac{1}{2G} & \frac{\kappa}{2G} + \frac{1}{6G^2} \emx, \quad \text{such that}
\\
& \frac{1}{\vert k - k_j \vert}\sum_{i=1}^{k_j - k} \bmx 1 \\ \frac{i}{G} \emx \bmx 1 & \frac{i}{G} \emx  = \bC_{\kappa,+}  + \bD_{\kappa, G, +}.
\end{align*}
Hence, we obtain
\begin{align*} 
\wh{\bbeta}^{+}(k) - \wh{\bbeta}^{-}(k) = \l( \mbf I - \kappa \big( \bC_{+} + \bD_{G,+}\big)^{-1} \big( \bC_{\kappa,+} + \bD_{\kappa,G,+} \big) \r) \bmx 1 & - \kappa \\ 0 & 1 \emx   \bmx \Delta^{(0)}_j\\ \Delta^{(1)}_j \emx + \bzeta_{k}. 
\end{align*}
With simple calculations, we get 
\begin{align}
& \l( \mbf I - \kappa \big( \bC_{+} + \bD_{G,+}\big)^{-1} \big( \bC_{\kappa,+} + \bD_{\kappa,G,+} \big) \r)\bmx 1 & - \kappa \\ 0 & 1 \emx  
\nn \\ 
& = \left[ \mbf I - \kappa \bmx 4 - 3 \kappa + \frac{3(1-\kappa)}{G-1} & 2 \kappa (1-\kappa) + \frac{2 (1-\kappa)(1+\kappa)}{G-1} \\ 6 (\kappa - 1) + \frac{6(\kappa  -1 )}{G^2 -1} & \kappa (4\kappa  - 3) + \frac{\kappa (4\kappa -3) + 3 (\kappa - 1)G - 1}{G^2 -1} \emx \right] \bmx 1 & - \kappa \\ 0 & 1 \emx  
\nn \\
& = \bmx 1 - \kappa (4-3\kappa) & - 2\kappa^2 (1-\kappa) \\ 6 \kappa (1-\kappa) & 1 - \kappa^2 (4\kappa -3) \emx\bmx 1 & - \kappa \\ 0 & 1 \emx  - \kappa \bmx \frac{3(1-\kappa)}{G-1} & \frac{2 (1-\kappa)(1+\kappa)}{G-1} \\  \frac{6(\kappa  -1 )}{G^2 -1} &  \frac{\kappa (4\kappa -3) + 3 (\kappa - 1)G - 1}{G^2 -1} \emx \bmx 1 & - \kappa \\ 0 & 1 \emx  
\nn \\
& = (1-\kappa) \bmx 1- 3\kappa & \kappa^2 - \kappa \\ 6\kappa & -2\kappa^2 + \kappa +1 \emx - \frac{\kappa(1-\kappa)}{G-1} \bmx 3 & 2- \kappa \\ - \frac{6}{G+1} & \frac{2\kappa - 1 - 3G}{G+1}\emx
\nn \\
& = (1-\kappa) \left[  \bmx 1- 3\kappa & \kappa^2 - \kappa \\ 6\kappa & -2\kappa^2 + \kappa +1 \emx - \frac{\kappa}{G-1} \bmx 3 & 2- \kappa \\ - \frac{6}{G+1} & \frac{2\kappa - 1 - 3G}{G+1}\emx \right] 
\nn \\
& = \mbf A(\kappa) + \mbf O_G (\kappa).
\label{eq:ao:kappa}
\end{align}

\paragraph{Case 2: $k - G < k_j < k$.}
Similarly as in Case~1, we have
\begin{align*}
\wh{\bbeta}^{+}(k) = &
\bmx \beta_{0,j+1}(k_j)+ \frac{k-k_j}{G}  \beta_{1,j+1}\\  \beta_{1,j+1}\emx + \bzeta_{k,+}, \quad \text{and}
\\
\wh{\bbeta}^{-}(k) 
= & \bmx \beta_{0,j+1}(k_j)+ \frac{k-k_j}{G}  \beta_{1,j+1}\\  \beta_{1,j+1}\emx  + \bzeta_{k,-} 
\\
& -  \left( \mbf C_{G, -} \right)^{-1}  \left( \sum_{i=-G+1}^{k_j-k} \bmx 1 \\ \frac{i}{G} \emx \bmx 1 & \frac{i}{G} \emx \right)\bmx \Delta^{(0)}_j+ \frac{k-k_j}{G} \Delta^{(1)}_j \\ \Delta^{(1)}_j \emx 
\\
=& \bmx \beta_{0,j+1}(k_j)+ \frac{k-k_j}{G}  \beta_{1,j+1}\\  \beta_{1,j+1}\emx  + \bzeta_{k,-} 
\\& - \left[  \mbf I - \frac{\vert k-k_j\vert}{G} \big( \bC_{-} + \bD_{G,-}\big)^{-1} \big( \bC_{\kappa,-} + \bD_{\kappa, G , -} \big)  \right]\bmx \Delta^{(0)}_j+ \frac{k-k_j}{G} \Delta^{(1)}_j \\ \Delta^{(1)}_j \emx 
\end{align*}
where $\bC_{-}$ and $\bD_{G,-}$ are defined as in~\eqref{def_C-}, and 
\begin{align*} 
& \bC_{\kappa, -} = \bmx 1 & - \frac{\kappa}{2} \\ -\frac{\kappa}{2} & \frac{\kappa^2}{3} \emx 
\quad \text{and} \quad
\bD_{\kappa, G, -} = \bmx 0 & \frac{1}{2G} \\ \frac{1}{2G} & -\frac{\kappa}{2G} + \frac{1}{6G^2} \emx, \quad \text{such that}
\\
& \frac{1}{\vert k - k_j \vert} \sum_{i=k_j - k+1}^{0} \bmx 1 \\ \frac{i}{G} \emx \bmx 1 & \frac{i}{G} \emx = \bC_{\kappa, -} + \bD_{\kappa, G, -}.
\end{align*}
Hence, we obtain
\begin{align*} 
\wh{\bbeta}^{+}(k) - \wh{\bbeta}^{-}(k) =
\l(  \mbf I - \kappa \big( \bC_{-} + \bD_{G,-}\big)^{-1} \big( \bC_{\kappa,-} + \bD_{\kappa, G , -} \big)  \r) 
\bmx 1 & \kappa \\ 0 & 1 \emx \bmx \Delta^{(0)}_j\\ \Delta^{(1)}_j \emx + \bzeta_{k}. 
\end{align*}
With simple calculations, we get 
\begin{align*}
& \l( \mbf I - \kappa \big( \bC_{-} + \bD_{G,-}\big)^{-1} \big( \bC_{\kappa,-} + \bD_{\kappa,G,-} \big) \r)\bmx 1 & \kappa \\ 0 & 1 \emx  
\\ & = \left[ \mbf I - \kappa \bmx 4 - 3 \kappa + \frac{3(\kappa-1)}{G+1} & 2 \kappa (\kappa-1) + \frac{2 (1-\kappa)(1+\kappa)}{G+1} \\ -6 (\kappa - 1) - \frac{6(\kappa  -1 )}{G^2 -1} & \kappa (4\kappa  - 3) + \frac{(\kappa-1)(4\kappa+1-3G)}{G^2 -1} \emx \right] \bmx 1 &  \kappa \\ 0 & 1 \emx  
\\& = (1-\kappa) \left[  \bmx 1- 3\kappa &2 \kappa^2  \\ -6\kappa & 4\kappa^2 + \kappa +1 \emx + \frac{\kappa}{G+1} \bmx 3 & -2 (1+\kappa) \\ - \frac{6}{G-1} & \frac{4\kappa + 1 - 3G}{G-1}\emx \right]  \bmx 1 &  \kappa \\ 0 & 1 \emx
\\& = (1-\kappa) \left[  \bmx 1- 3\kappa & \kappa - \kappa^2 \\ -6\kappa & -2\kappa^2 + \kappa +1 \emx - \frac{\kappa}{G+1} \bmx -3 & 2-\kappa \\ \frac{6}{G-1} & \frac{2\kappa - 1 + 3G}{G-1}\emx \right] 
\\& = \mbf A(\kappa) + \mbf O_G (\kappa).
\end{align*}
\end{proof}

\begin{proof}[Proof of Theorem~\ref{5.1}~\ref{thm:est:one}.] 
Denote by
\begin{align*}
\mc A_n =& \left\{ \max_{k: \, \min_{1 \le j \le J_n} \vert k - k_j\vert \ge G} W_{k,n}(G) \le C_n(G, \alpha_n) \text{ \ and \ } \r.
\\
& \qquad \l.  \min_{k: \, \min_{1 \le j \le J_n} \vert k - k_j \vert < (1 - \eta)G}  W_{k,n}(G)> C_n(G, \alpha_n) \right\},
\end{align*}
the event where every change point has an estimator detected within $\lceil (1 - \eta)G \rceil$-distance, and there exists no change point estimator outside $G$-distance of all $k_j, \, j = 1, \ldots, J_n$.
Then, we have
\begin{align*} 
\mc A_n \subset \l\{ \wh{J}_n = J_n \r\} \quad \text{and} \quad 
\mc A_n \subset \left\{ \max_{1 \le j \le J_n} \l\vert \wh{k}_j \mathbb{I}_{\{j \le \wh{J}_n\}} - k_j \r\vert < G \right\},
\end{align*}
so our claim follows if we show $\Pr(\mc A_n) \to 1$.

For $k$ satisfying $\vert k - k_j \vert \ge G$, the intervals $\cI^{+}(k)$ and $\cI^{-} (k)$ do not contain any change point so that from Theorem~\ref{thm3.1},
\begin{align} 
\label{A.type1error}
\lim_{n \to \infty} \Pr \left( \max_{k: \underset{1 \le j \le J_n}{\min} \vert k-k_j\vert \ge G} W_{k,n}(G) \ge C_n(G, \alpha_n) \right) = 0 
\end{align}
by the definition of $C_n(G, \alpha)$, the condition~\eqref{alpha_n} on $\alpha_n$
and the fact that the limiting distribution is continuous. 

Next, under~\ref{b:one}, for any $k$, there exists at most one $k_j$ in the interval $\{ k - \lceil (1-\eta)G \rceil, \ldots, k+(1-\eta)G \}$. 
Recall the definition of $W^*_{k, n}(G)$ in~\eqref{eq:pop:w}. 
By Lemma~\ref{lemA.4},
\begin{align}
& \min_{\substack{k: \, \vert k-k_j\vert < (1-\eta) G \\ 1 \le j \le J_n}} 
\frac{\sqrt{8}\tau}{\sqrt{G}} W^*_{k,n}(G) 
= \min_{\substack{k: \, \vert k-k_j\vert < (1-\eta) G \\ 1 \le j \le J_n}}
\left\Vert \bmx 1 & 0 \\ 0 & \frac{1}{\sqrt{3}} \emx \big(\wh{\bbeta}^{+}(k) - \wh{\bbeta}^{-} (k) \big)  \right\Vert
\nn \\
& \ge \min_{\substack{\kappa: \, 0 \le \kappa < 1 - \eta \\ 1 \le j \le J_n} } 
\left\Vert \bmx 1 & 0 \\ 0 & \frac{1}{\sqrt{3}} \emx
\big( \mbf A(\kappa) +\mbf O_G(\kappa) \big) \b\Delta_{j} \right\Vert
 - \max_{\substack{k: \, \vert k-k_j\vert < (1-\eta) G \\ 1 \le j \le J_n} }  
 \left\Vert  \bzeta_{k}\right\Vert.
\nn 
\end{align}
From Lemmas~\ref{lem:invC} and~\ref{lemA.3}, 
\begin{align}
& \max_{G \le k \le n - G} \left\Vert  \bzeta_{k} \right\Vert 
\le \frac{1}{G} \l\{
\max_k \left\Vert  \l(\bC_{+}^{-1} + O(G^{-1})\r) \sum_{i=k+1}^{k+G} \mbf x_{i,k} \epsilon_i \r\Vert +
\max_k \l\Vert \l(\bC_{-}^{-1} + O(G^{-1})\r) \sum_{i=k-G+1}^{k} \mbf x_{i,k} \epsilon_i \right\Vert \r\}
\nn \\
& \le \frac{1}{G} \l\{
\l( \l\Vert \bC_+^{-1} \r\Vert + O(G^{-1}) \r) \max_k \l\Vert \sum_{i=k+1}^{k+G} \mbf x_{i,k} \epsilon_i \r\Vert +
\l(\l\Vert \bC_-^{-1} \r\Vert + O(G^{-1}) \r) \max_k \l\Vert \sum_{i=k-G+1}^{k} \mbf x_{i,k} \epsilon_i \right\Vert \r\}
\nn \\
& = O_P\left( \sqrt{\frac{\log(n/G)}{G}}\right) ,
\label{eq:alt:decomp:one}
\end{align}
Also, noting that for any invertible matrix $\mbf A \in \mathbb{R}^{2 \times 2}$,
we have 
$\Vert \mbf A^{-1} \Vert_2 \le \sqrt{2} \Vert \mbf A^{-1} \Vert_1$, and that for any $\kappa \in [0, 1 - \eta)$, we have 
\begin{align}
& (1 - \kappa) \mbf A(\kappa)^{-1} = \begin{bmatrix} 1 - 3\kappa & \pm(\kappa^2 - \kappa) \\ \pm 6 \kappa & -2\kappa^2 + \kappa + 1 \end{bmatrix}^{-1}
= \begin{bmatrix} \frac{2\kappa+1}{1-\kappa} & \mp\frac{\kappa}{1-\kappa} \\ \mp\frac{-6\kappa}{(1-\kappa)^2} & \frac{1-3\kappa}{(1-\kappa)^2} \end{bmatrix} ,
\quad \text{where}
\nn \\
&
\frac{2\kappa+1}{1-\kappa} + \frac{6\kappa}{(1-\kappa)^2} = 
\frac{-2(\kappa - 7/4)^2 + 57/8}{(1-\kappa)^2}
\le \frac{6}{\eta^2} \quad \text{and}
\nn \\
& \frac{\kappa}{1-\kappa}+ \frac{\vert 1-3\kappa \vert}{(1-\kappa)^2} \le \frac{\kappa(1-\kappa)+\vert1-3\kappa\vert}{\eta^2} \le \frac{2}{\eta^2}, \quad \text{such that}
\nn \\
& \l\Vert \mbf A(\kappa)^{-1} \r\Vert_2^{-1} \ge \frac{1}{\sqrt{2}} \l\Vert \mbf A(\kappa)^{-1} \r\Vert_1^{-1} \ge \frac{\eta^3}{6\sqrt{2}}.
\label{eq:alt:decomp:two}
\end{align}
Hence, from~\eqref{eq:alt:decomp:one}--\eqref{eq:alt:decomp:two}, the condition~\eqref{alpha_n} and Assumption~\ref{b:two}, we have
\begin{align*}
& \min_{\substack{k: \, \vert k-k_j\vert < (1-\eta) G \\ 1 \le j \le J_n}} 
\frac{\sqrt{8}\tau}{\sqrt{G}} W^*_{k,n}(G) \ge
\frac{1}{\sqrt{3}} \l( \frac{\eta^3}{6 \sqrt{2}} + O(G^{-1}) \r) \min_{1 \le j \le J_n} 
\l\Vert \b\Delta_{j} \r\Vert + O_P\l(\sqrt{\frac{\log(n/G)}{G}} \r), \quad \text{and}
\\
& \Pr \left( \min_{\substack{k: \, \vert k-k_j \vert < (1-\eta) G \\ 1 \le j \le J_n} }  
W^*_{k,n}(G) \le C_{\alpha_n} (n, G)  \right)
\\ 
& \le \Pr \left(\frac{\sqrt{G}}{\sqrt{8} \tau} \left( \frac{\eta^{3}}{6 \sqrt{6}} + O(G^{-1}) \right)\min_{1 \le j \le J_n} \Vert \bDelta_j \Vert
+ O_P \left( \sqrt{\log(n/G)} \right)  \le C_{\alpha_n} (n, G)   \right) 
\\
& = \Pr \left(\frac{1}{\sqrt{8} \tau} \left( \frac{\eta^{3}}{6 \sqrt{6}} + O(G^{-1}) \right) \le \left( \sqrt{G} \min_{1 \le j \le J_n} \Vert  \bDelta_j \Vert \right)^{-1} \left(C_{\alpha_n} (n, G) + O_P \left( \sqrt{\log(n/G)} \right)  \right)  \right)
\\
& = \Pr \left(\frac{1}{\sqrt{8} \tau} \left(\frac{\eta^{3}}{6 \sqrt{6}} + O(G^{-1}) \right)
{\le} o_P (1) \right) 
\xrightarrow[G \to \infty]{} 0,
\end{align*}
The claim holds if we replace $\tau$ with $\wh\tau_k$ provided that $\wh\tau^2_k = O_P(1)$ uniformly in $k$.
Then, the assertion follows from that $\Pr (\mc A_n) \to 1$. 
\end{proof}

%

\subsubsection{Proof of Theorem~\ref{5.1}~\ref{thm:est:two}}

We need the following lemmas for the proof of Theorem~\ref{5.1}~\ref{thm:est:two}.

\begin{lem}[Theorem~9 of \cite{burkholder1966}] \label{burkholder}
If $S_i$, $i=1,\ldots,n$, is a martingale, then for $1 < p < \infty$, 
there exists constants $C', C'' > 0$, which depend only on $p$ such that
\begin{align*}
C' \E \l( \left\vert  \sum_{i=1}^{n} X_i^2 \right\vert^{p/2} \r) \le 
\E \l(\vert S_n\vert^{p}\r) 
\le C'' \E \l(\left\vert  \sum_{i=1}^{n} X_i^2 \right\vert^{p/2}\r), 
\end{align*}
where $X_i = S_i - S_{i-1}$ and $S_0 \equiv 0$. 
\end{lem}

\begin{lem} \label{convrate_lem1}
Assume that $\{\epsilon_i\}$ is a sequence of i.i.d.\ random variables with $\E(\epsilon_1) = 0$ and $\E (\vert\epsilon_1\vert^{\gamma}) < \infty$ for $\gamma > 2$. Then there exist constants $C_0, C_1 >0$, which depend only on $\gamma$ such that
\begin{align*}
\E \l(\left\vert \sum_{i = \ell + 1}^r \epsilon_i\right\vert^{\gamma} \r) \le C_{0} \vert r - \ell \vert^{\gamma/2}
\quad \text{and} \quad
\E \l(\left\vert \sum_{i = \ell + 1}^{r} i \epsilon_i \right\vert^{\gamma}\r) \le C_1 \vert  r - \ell \vert^{3\gamma/2}.
\end{align*} 
\end{lem}

\begin{proof}
Define $S_{m} = \sum_{i = \ell + 1}^{\ell + m} \epsilon_i$. 
Then $\{S_{m}\}$ is a martingale and by Lemma~\ref{burkholder},
\begin{align*}
\E \l(\left \vert \sum_{i = \ell + 1}^{\ell +m} \epsilon_i \right\vert^{\gamma} \r) 
\le C'' \E \l( \left\vert  \sum_{i = \ell + 1}^{\ell + m} \epsilon_i^2 \right\vert^{\gamma/2} \r) 
\le 
C'' \l(\sum_{i = \ell + 1}^{\ell + m} \left\{ \E\l( \vert\epsilon_i\vert^\gamma \right)\right\}^{2/\gamma} \r)^{\gamma/2}
= 
C'' m^{\gamma/2} \E \l(\vert\epsilon_1\vert^{\gamma} \r),
\end{align*}
and the first assertion  follows. 
Similarly, using that $T_{m} = \sum_{i = \ell + 1}^{\ell + m} i \epsilon_i$ is also a martingale, 
\begin{align*}
\E \l(\left\vert \sum_{i = \ell + 1}^{\ell + m} i \epsilon_i \right\vert^{\gamma} \r)
&\le 
C'' \E \l( \left\vert  \sum_{i = \ell + 1}^{\ell + m} (i \epsilon_i)^2\right\vert^{\gamma/2} \r) 
\le 
C'' \left( \sum_{i = \ell + 1}^{\ell + m} i^2  \left\{ \E\l( \vert\epsilon_i\vert^\gamma \right)\right\}^{2/\gamma}\right)^{\gamma/2}
\\
&=
C'' \left( \sum_{i = \ell + 1}^{\ell + m} i^{2} \right)^{\gamma/2} 
\E \l(\vert\epsilon_1\vert^{\gamma}\r).
\end{align*}
Noting that
\begin{align*}
\sum_{i = \ell + 1}^{r} i^{2} \le \int_{\ell + 1}^{r} x^2 \diff x = \frac{r^{3} - (\ell + 1)^{3}}{3} \le \frac{(r - \ell)^{3}}{3},
\end{align*}
the second assertion follows.
\end{proof}

\begin{lem} \label{iid_lem} 
Assume that $\{\epsilon_i\}$ is a sequence of i.i.d.\ random variables with $\E(\epsilon_1) = 0$ and $\E (\vert\epsilon_1\vert^{2+\nu}) < \infty$ for some $\nu > 0$. 
Then~\ref{a:two} and~\ref{a:three} hold.
\proof 
\ref{a:two} follows from \cite{komlos1976}. \ref{a:three} is shown to hold with $\gamma = 2 + \nu$ in Lemma~\ref{convrate_lem1}.
\endproof
\end{lem}

\begin{lem}[Lemma~B.1 of \cite{kirch2006}, Lemma~3.1 of \cite{eichinger2018}]
\label{lem:kirch}
Let $\{b_i\}_i$ be a sequence of non-increasing, positive constants.
Then under~\ref{a:three}, there exists a constant $C > 0$ which depends only on {$C_0$, $C_1$ and $\gamma$} such that {for any $m \in \mathbb{Z}$, $0 < \ell < r$ and $z > 0$,}
\begin{align*}
& \Pr\l( \max_{\ell \le k \le r} b_k \l\vert \sum_{i = m + 1}^{m + k} \epsilon_i \r\vert > z \r) 
\le \frac{C}{ z^\gamma} \l( \ell^{\gamma/2} b_\ell^{\gamma} + \sum_{i = \ell + 1}^{{ r}} b_i^{\gamma} i^{\gamma/2-1}\r), 
\\
& \Pr\l( \max_{\ell \le k \le { r}} b_{{ k}} \l\vert \sum_{i = m - k + 1}^m \epsilon_i \r\vert > z \r) 
\le \frac{C}{ z^\gamma} \l( { \ell}^{\gamma/2} b_{{\ell}}^{\gamma} + \sum_{i = { \ell} + 1}^{{ r}} b_i^{\gamma} i^{\gamma/2-1}\r).
\end{align*}
\end{lem}

\begin{lem} \label{lem:hr}
Let $k \in \{k_j - G, \ldots, k_j - 1\}$ for some $j = 1, \ldots, J_n$.
Recall the definitions of $\b\zeta_k = (\zeta^{(0)}_k, \zeta^{(1)}_k)^\top$ in~\eqref{eq:zeta}.
Then under~\ref{a:two}--\ref{a:three}, for $\xi \in \{1, \ldots, G\}$ and $p \in \{0, 1\}$:
\begin{enumerate}[label = (\roman*)]
\item \label{lem:hr:one} 
$\Pr \left( \displaystyle \max_{k: \, \xi \le \vert k - k_j \vert \le G}
G \big\vert \zeta_{k_j}^{(p)} - \zeta_{k}^{(p)} \big\vert > \beta \right) 
= O\l(\beta^{-\gamma} G^{\gamma/2}\r)$.

\item \label{lem:hr:two} 
$\Pr \left( \displaystyle \max_{k: \, \xi \le \vert k - k_j \vert \le G}
\frac{G \big\vert \zeta_{k_j}^{(p)} - \zeta_{k}^{(p)} \big\vert}{k_j-k} > \beta \right) 
= O \l(\beta^{-\gamma} (\xi^{-\gamma/2} + G^{-\gamma/2}) \r)$.

\item \label{lem:hr:three} 
$\Pr \left( \displaystyle \max_{k: \, \xi \le \vert k - k_j \vert \le G} 
\frac{G \big\vert \zeta_{k_j}^{(p)} - \zeta_{k}^{(p)} \big\vert}{(k_j-k)^2} > \beta \right) 
= O \l(\beta^{-\gamma} (\xi^{-3\gamma/2} + G^{-\gamma/2} \xi^{-\gamma} + G^{-\gamma} \xi^{-\gamma/2}) \r)$.

\item  \label{lem:hr:four} 
$\Pr \left( \displaystyle \max_{k: \, \xi \le \vert k - k_j \vert \le G} 
G \big\vert \zeta_{k_j}^{(p)} + \zeta_{k}^{(p)} \big\vert > \beta \right)
= O \l(\beta^{-\gamma} G^{\gamma/2} \r)$.

\item  \label{lem:hr:five} 
$\Pr \left( \displaystyle \max_{k: \, \xi \le \vert k - k_j \vert \le G} 
\frac{G \big\vert \zeta_{k_j}^{(p)} + \zeta_{k}^{(p)} \big\vert}{k_j-k} > \beta \right)
= O \l(\beta^{-\gamma} (\xi^{-\gamma/2} + G^{-\gamma/2} + G^{\gamma/2} \xi^{-\gamma}) \r)$.
\end{enumerate}
\end{lem}

\begin{proof} 
Note that we can write
$G \mbf C_{G, \pm}^{-1} = \mbf C_{\pm}^{-1} + \mbf D_{G, \pm}^{-1}$, where
\begin{align*}
\bD_{G,+}^{-} =  \bmx \frac{6}{G-1} & -\frac{6}{G-1} \\ -\frac{6}{G-1} & \frac{12}{G^2-1} \emx , ~ \bD_{G,-}^{-} = \bmx - \frac{6}{G+1} & - \frac{6}{G+1} \\ - \frac{6}{G+1} & \frac{12}{G^2-1} \emx. 
\end{align*}
Then we get
\begin{align}
\bzeta_{k_j} - \bzeta_{k} =& \frac{\bC_{+}^{-1} + \bD_{G,+}^{-}}{G} \left( \sum_{i=k_j+1}^{k_j+G} \bmx 1 \\ \frac{i-k_j}{G} \emx \epsilon_i - \sum_{i=k+1}^{k+G} \bmx 1 \\ \frac{i-k}{G} \emx \epsilon_i \right) 
\nn \\
& -  \frac{\bC_{-}^{-1} + \bD_{G,-}^{-}}{G} \left( \sum_{i=k_j-G+1}^{k_j} \bmx 1 \\ \frac{i-k_j}{G} \emx \epsilon_i - \sum_{i=k-G+1}^{k} \bmx 1 \\ \frac{i-k}{G} \emx \epsilon_i \right) 
\nn \\
=& \sum_{i=k-G+1}^{k_j - G}  \frac{\bC_{-}^{-1} + \bD_{G,-}^{-}}{G}  \bmx 1 \\ \frac{i-k}{G} \emx \epsilon_i 
+ \sum_{i=k_j-G+1}^{k} \frac{\bC_{-}^{-1} + \bD_{G,-}^{-}}{G} \left[ \bmx 1 \\ \frac{i-k}{G}\emx  -\bmx 1 \\ \frac{i-k_j}{G} \emx \right] \epsilon_i 
\nn \\
& + \sum_{i=k+1}^{k_j} \left[ - \frac{\bC_{+}^{-1} + \bD_{G,+}^{-}}{G}\bmx 1 \\ \frac{i-k}{G}\emx - \frac{\bC_{-}^{-1} + \bD_{G,-}^{-}}{G} \bmx 1 \\ \frac{i-k_j}{G}\emx \right]\epsilon_i
\nn \\
& + \sum_{i=k_j+1}^{k+G}  \frac{\bC_{+}^{-1} + \bD_{G,+}^{-}}{G} \left[ \bmx 1 \\ \frac{i-k_j}{G} \emx - \bmx 1 \\ \frac{i-k}{G} \emx\right]\epsilon_i 
+ \sum_{i=k+G+1}^{k_j+G}  \frac{\bC_{+}^{-1} +  \bD_{G,+}^{-}}{G}\bmx 1 \\ \frac{i-k_j}{G} \emx  \epsilon_i
\nn \\
=& \sum_{i=k-G+1}^{k_j - G}  \frac{\bC_{-}^{-1} + \bD_{G,-}^{-}}{G}  \bmx 1 \\ \frac{i-k}{G} \emx \epsilon_i 
+ \sum_{i=k_j-G+1}^{k} \frac{\bC_{-}^{-1} + \bD_{G,-}^{-}}{G}  \bmx 0 \\ \frac{k_j-k}{G}\emx \epsilon_i 
\nn \\
& - \sum_{i=k+1}^{k_j} \left[  \frac{\bC_{+}^{-1} +  \bD_{G,+}^{-}}{G}\bmx 1 \\ \frac{i-k}{G}\emx + \frac{\bC_{-}^{-1} +  \bD_{G,-}^{-}}{G} \bmx 1 \\ \frac{i-k_j}{G}\emx \right]\epsilon_i
\nn \\
& + \sum_{i=k_j+1}^{k+G}  \frac{\bC_{+}^{-1} + \bD_{G,+}^{-}}{G}  \bmx 0 \\ \frac{k-k_j}{G} \emx \epsilon_i 
+ \sum_{i=k+G+1}^{k_j+G}  \frac{\bC_{+}^{-1} +  \bD_{G,+}^{-}}{G}\bmx 1 \\ \frac{i-k_j}{G} \emx  \epsilon_i
\nn \\
=& \frac{1}{G} \left[ \sum_{i=k-G+1}^{k_j - G} \bmx \left( 4 - \frac{6}{G+1} \right) + \left( 6 - \frac{6}{G+1} \right) \frac{i-k}{G}  \\  \left( 6 - \frac{6}{G+1} \right) +  \left( 12 + \frac{12}{G^2-1} \right) \frac{i-k}{G} \emx \epsilon_i + \sum_{i=k_j-G+1}^{k} \bmx 6 - \frac{6}{G+1} \\ 12 + \frac{12}{G^2-1} \emx \frac{k_j -k}{G} \epsilon_i \right.
\nn \\
& \left.- \sum_{i=k+1}^{k_j } \bmx \left( 8 + \frac{12}{G^2-1} \right) + 6 \frac{k-k_j}{G} - \frac{6(i-k)}{G(G-1)} - \frac{6(i-k_j)}{G(G+1)}  \\  - \frac{12G}{G^2-1} + \left(12 + \frac{12}{G^2-1} \right) \frac{2i-k-k_j}{G} \emx \epsilon_i + \sum_{i=k_j+1}^{k+G} \bmx 6 + \frac{6}{G-1}\\ -12 - \frac{12}{G^2-1} \emx \frac{k_j -k}{G} \epsilon_i \right.
\nn \\
& \left. + \sum_{i=k+G+1}^{k_j + G} \bmx \left( 4 + \frac{6}{G-1} \right) -\left( 6 + \frac{6}{G-1} \right) \frac{i-k_j}{G} \\ - \left(6 + \frac{6}{G-1} \right) + \left( 12 + \frac{12}{G^2-1} \right) \frac{i-k_j}{G} \emx \epsilon_i \right]
\nn \\
=& \frac{1}{G} \left[ \sum_{i=k-G+1}^{k_j - G} \bmx \left( 4 - \frac{6}{G+1} \right) + \left( 6 - \frac{6}{G+1} \right) \frac{i-k}{G}  \\  \left( 6 - \frac{6}{G+1} \right) +  \left( 12 + \frac{12}{G^2-1} \right) \frac{i-k}{G} \emx \epsilon_i + \sum_{i=k_j-G+1}^{k_j} \bmx 6 - \frac{6}{G+1} \\ 12 + \frac{12}{G^2-1} \emx \frac{k_j -k}{G} \epsilon_i \right.
\nn \\
& - \sum_{i=k+1}^{k_j } \bmx \left( 8 + \frac{12}{G^2-1} \right) - \frac{12(i-k)}{G^2-1} 
\\ 
- \frac{12G}{G^2-1} + \left( 24 + \frac{24}{G^2-1} \right) \frac{i-k}{G} \emx \epsilon_i  + \sum_{i=k_j+1}^{k_j+G} \bmx 6+ \frac{6}{G-1} \\ -12 - \frac{12}{G^2-1} \emx \frac{k_j -k}{G} \epsilon_i 
\nn \\
& \left. + \sum_{i=k+G+1}^{k_j + G} \bmx \left(4 + \frac{6}{G-1} \right) - \left( 6 + \frac{6}{G-1} \right) \frac{i-k}{G} \\ - \left( 6 + \frac{6}{G-1} \right) + \left( 12 + \frac{12}{G^2-1} \right) \frac{i-k}{G} \emx \epsilon_i \right]
\nn \\
=:& \ T_1 + T_2 + T_3 + T_4 + T_5.
\label{eq:decomp:zeta}
\end{align}

Now we control the terms related to bounding $T_1$ in~\eqref{eq:decomp:zeta}. 
Firstly, we bound the partial sums. Note that
\begin{align}
& \nonumber \Pr \left( \max_{k_j - G \le k \le k_j - \xi} \left| \sum_{i=k-G+1}^{k_j - G} \epsilon_i \right| > \beta \right) = \Pr \left( \max_{\xi \le k_j - k \le G } \left| \sum_{i=1}^{k_j - k} \epsilon_{k-G+i} \right| > \beta \right)
\\ \nonumber & \le \Pr \left( \left| \sum_{i=1}^{\xi} \epsilon_{k-G+i} \right| > \frac{\beta }{2} \right) 
+ \Pr \left( \max_{\xi \le k_j - k \le G } \left| \sum_{i=\xi+1}^{k_j - k} \epsilon_{k-G+i} \right| > \frac{\beta }{2} \right)
\\ & {\lesssim} \beta ^{-\gamma} \left({\rm E} \left| \sum_{i=1}^{\xi} \epsilon_{k-G+i} \right|^{\gamma} + {\rm E} \max_{\xi \le k_j - k \le G } \left| \sum_{i=\xi+1}^{k_j - k} \epsilon_{k-G+i} \right|^{\gamma} \right)
\nn
\\
& \underset{(*)}{\lesssim}  \beta^{-\gamma} \left( \xi^{\gamma/2} + \sum_{i=\xi+1}^{G} i ^{\gamma/2-1}  \right)
= O \big(\beta ^{-\gamma} G^{\gamma/2} \big). \label{conv_eq1-1}  
\end{align}
Hereafter, the inequalities marked by $(*)$ follow from the H{\'a}jek-R{\'e}nyi-type inequality in Lemma~\ref{lem:kirch}.
For $\alpha = 1,2$,
\begin{align}
& \nonumber \Pr \left( \max_{k_j - G \le k \le k_j - \xi} \frac{ \left| \sum_{i=k-G+1}^{k_j - G} \epsilon_i \right|}{|k_j - k|^{\alpha}} > \beta\right) = \Pr \left( \max_{\xi \le k_j - k \le G } \frac{ \left| \sum_{i=1}^{k_j - k} \epsilon_{k-G+i} \right|}{|k_j - k|^{\alpha}} > \beta\right)
\\ \nonumber & \le \Pr \left( \frac{\left| \sum_{i=1}^{\xi} \epsilon_{k-G+i} \right|}{\xi^{\alpha}} > \frac{\beta }{2} \right) 
+ \Pr \left( \max_{\xi \le k_j - k \le G } \frac{\left| \sum_{i=\xi+1}^{k_j - k} \epsilon_{k-G+i} \right|}{|k_j - k|^{\alpha}} > \frac{\beta }{2} \right)
\\ \nonumber & \lesssim \beta ^{-\gamma} \left({\rm E} \xi^{-\alpha\gamma}\left| \sum_{i=1}^{\xi} \epsilon_{k-G+i} \right|^{\gamma} + {\rm E} \max_{\xi \le k_j - k \le G } \frac{\left| \sum_{i=\xi+1}^{k_j - k} \epsilon_{k-G+i} \right|^{\gamma}}{|k_j - k|^{\alpha\gamma}} \right)
\\ & \underset{(*)}{\lesssim}  \beta^{-\gamma} \left( \xi^{\gamma/2-\gamma\alpha} + \sum_{i=\xi+1}^{G} i^{-\alpha\gamma} i ^{\gamma/2-1}  \right)
= O \big( \beta ^{-\gamma} \xi^{(1/2-\alpha)\gamma} \big) 
= \left\{ \begin{array}{ll} O \big( (\beta^2  \xi)^{-\gamma/2} \big) & (\alpha=1) \\ O \big( (\beta^2 \xi^3)^{-\gamma/2} \big) & (\alpha=2). \end{array} \right.
\label{conv_eq1-2}
\end{align}
Next we bound trend terms. 
\begin{align}
& \nonumber \Pr \left( \max_{k_j - G \le k \le k_j - \xi} \left| \sum_{i=k-G+1}^{k_j - G} \frac{i-k}{G} \epsilon_i \right| > \beta \right) 
= \Pr \left( \max_{\xi \le k_j - k \le G } \left| \sum_{i=1}^{k_j - k} \frac{i-G}{G}\epsilon_{k-G+i} \right| > \beta \right)
\\ \nonumber 
& \le \Pr \left( \max_{\xi \le k_j - k \le G } \left| \sum_{i=1}^{k_j - k} \frac{i}{G}\epsilon_{k-G+i} \right| > \frac{\beta }{2} \right) 
+ \Pr \left( \max_{\xi \le k_j - k \le G } \left| \sum_{i=1}^{k_j - k} \epsilon_{k-G+i} \right| > \frac{\beta }{2} \right) ~ \l( \underset{\eqref{conv_eq1-1}}{=}O(\beta^{-\gamma}G^{ \gamma/2}) \r)
\\ 
&= O (\beta^{-\gamma} G^{\gamma/2} ), \label{conv_eq1-3} 
\quad \text{since} \\
& \Pr \left( \max_{\xi \le k_j - k \le G } \left| \sum_{i=1}^{k_j - k} \frac{i}{G}\epsilon_{k-G+i} \right| > \frac{\beta }{2} \right) 
\le \Pr \left( \left| \sum_{i=1}^{\xi} \frac{i}{G}\epsilon_{k-G+i} \right| > \frac{\beta }{4} \right) 
\nn \\
&+ \Pr \left( \max_{\xi+1 \le k_j - k \le G } \left| \sum_{i=\xi+1}^{k_j - k} \frac{i}{G}\epsilon_{k-G+i} \right| > \frac{\beta }{4} \right)
\nn \\
& \lesssim (\beta G)^{-\gamma} \left({\rm E} \left| \sum_{i=1}^{\xi}i \epsilon_{k-G+i} \right|^{\gamma} + {\rm E} \max_{\xi+1 \le k_j - k \le G } \left| \sum_{i=\xi+1}^{k_j - k}i\epsilon_{k-G+i} \right|^{\gamma} \right)
\nn \\
& \underset{(*)}{\lesssim}  (\beta G)^{-\gamma} \left( \xi^{3\gamma/2} + \sum_{i=\xi+1}^{G} i ^{3\gamma/2-1}  \right)
= O \big((\beta G)^{-\gamma} G^{3\gamma/2} \big) = O ( \beta^{-\gamma}G^{ \gamma/2}). \nn
\end{align}
Next, 
\begin{align}
& \nonumber \Pr \left( \max_{k_j - G \le k \le k_j - \xi} \left|\frac{ \sum_{i=k-G+1}^{k_j - G} \frac{i-k}{G} \epsilon_i}{k_j - k} \right| > \beta \right) = \Pr \left( \max_{\xi \le k_j - k \le G } \left|\frac{ \sum_{i=1}^{k_j - k} \frac{i-G}{G}\epsilon_{k-G+i}}{k_j-k} \right| > \beta \right)
\\ \nonumber & \le \Pr \left( \max_{\xi \le k_j - k \le G } \left| \sum_{i=1}^{k_j - k} \frac{ \frac{i}{G}\epsilon_{k-G+i}}{k_j-k} \right| > \frac{\beta}{2} \right) 
+ \Pr \left( \max_{\xi \le k_j - k \le G } \left| \frac{\sum_{i=1}^{k_j - k} \epsilon_{k-G+i}}{k_j-k} \right| > \frac{\beta}{2} \right) ~ \l( \underset{\eqref{conv_eq1-2}}{=}O((\beta^2 \xi )^{-\gamma/2}) \r)
\\  
& = O \big(\beta^{-\gamma}( \xi^{-\gamma/2} + G^{-\gamma/2}) \big), \label{conv_eq1-4} \quad \text{since}
\\
& \Pr \left( \max_{\xi \le k_j - k \le G } \left| \sum_{i=1}^{k_j - k} \frac{ \frac{i}{G}\epsilon_{k-G+i}}{k_j-k} \right| > \frac{\beta }{2} \right) 
\le  \Pr \left( \left| \frac{\sum_{i=1}^{\xi} \frac{i}{G}\epsilon_{k-G+i}}{\xi} \right| > \frac{\beta }{4} \right) 
\nn \\
&+ \Pr \left( \max_{\xi+1 \le k_j - k \le G } \left| \frac{\sum_{i=\xi+1}^{k_j - k} \frac{i}{G}\epsilon_{k-G+i}}{k_j-k} \right| > \frac{\beta}{4} \right)
\nn \\
& \lesssim (\beta G)^{-\gamma} \left( \xi^{-\gamma} {\rm E} \left| \sum_{i=1}^{\xi}i \epsilon_{k-G+i} \right|^{\gamma} + {\rm E} \max_{\xi+1 \le k_j - k \le G } \left|\frac{ \sum_{i=\xi+1}^{k_j - k}i\epsilon_{k-G+i}}{k_j-k} \right|^{\gamma} \right)
\nn \\
& \underset{(*)}{\lesssim}  (\beta G)^{-\gamma} \left(\xi^{-\gamma} \xi^{3\gamma/2} + \sum_{i=\xi+1}^{G} i^{-\gamma} i^{3\gamma/2-1}  \right)
= O \big((\beta G)^{-\gamma} G^{\gamma/2} \big) = O ( \beta^{-\gamma}G^{-\gamma/2}). \nn
\end{align}
Finally, 
\begin{align}
& \nonumber \Pr \left( \max_{k_j - G \le k \le k_j - \xi} \left|\frac{ \sum_{i=k-G+1}^{k_j - G} \frac{i-k}{G} \epsilon_i}{(k_j - k)^2} \right| > \beta \right) = \Pr \left( \max_{\xi \le k_j - k \le G } \left|\frac{ \sum_{i=1}^{k_j - k} \frac{i-G}{G}\epsilon_{k-G+i}}{(k_j - k)^2} \right| > \beta \right)
\\ \nonumber & \le \Pr \left( \max_{\xi \le k_j - k \le G } \left| \sum_{i=1}^{k_j - k} \frac{ \frac{i}{G}\epsilon_{k-G+i}}{(k_j - k)^2} \right| > \frac{\beta}{2} \right) 
+ \Pr \left( \max_{\xi \le k_j - k \le G } \left| \frac{\sum_{i=1}^{k_j - k} \epsilon_{k-G+i}}{(k_j - k)^2} \right| > \frac{\beta }{2} \right) ~ \l( \underset{\eqref{conv_eq1-2}}{=}O((\beta^2  \xi^3 )^{-\gamma/2}) \r)
\\  & = O \big( \beta^{-\gamma} \xi^{-\gamma/2} (G^{-\gamma} + \xi^{-\gamma}) \big), \label{conv_eq1-5} \quad \text{from that}
\\
& \Pr \left( \max_{\xi \le k_j - k \le G } \left| \sum_{i=1}^{k_j - k} \frac{ \frac{i}{G}\epsilon_{k-G+i}}{(k_j - k)^2} \right| > \frac{\beta }{2} \right) 
\le  \Pr \left( \left| \frac{\sum_{i=1}^{\xi} \frac{i}{G}\epsilon_{k-G+i}}{\xi^2} \right| > \frac{\beta }{4} \right) 
\nn \\
&+ \Pr \left( \max_{\xi+1 \le k_j - k \le G } \left| \frac{\sum_{i=\xi+1}^{k_j - k} \frac{i}{G}\epsilon_{k-G+i}}{(k_j - k)^2} \right| > \frac{\beta}{4} \right)
\nn \\
& \lesssim (\beta G)^{-\gamma} \left( \xi^{-2\gamma} {\rm E} \left| \sum_{i=1}^{\xi}i \epsilon_{k-G+i} \right|^{\gamma} + {\rm E} \max_{\xi+1 \le k_j - k \le G } \left|\frac{ \sum_{i=\xi+1}^{k_j - k}i\epsilon_{k-G+i}}{(k_j - k)^2} \right|^{\gamma} \right)
\nn\\
& \underset{(*)}{\lesssim}  (\beta G)^{-\gamma} \left(\xi^{-2\gamma} \xi^{3\gamma/2} + \sum_{i=\xi+1}^{G} i^{-2\gamma} i^{3\gamma/2-1}  \right)
= O \big((\beta G)^{-\gamma} \xi^{-\gamma/2} \big) .
\nn
\end{align}
Therefore, combining \eqref{conv_eq1-1}--\eqref{conv_eq1-5}, we have that for some constants $c_1$, $c_2$ and sequences $o_{1,G}$, $o_{2,G}$ which converge to $0$ with rate $O(G^{-1})$, we have 
\begin{align} 
\label{conv_eq1-fin1}
& \Pr \left( \max_{k_j-G \le k \le k_j - \xi} \left| \sum_{i=k-G+1}^{k_j - G} \left( (c_1 + o_{1,G}) + (c_2 + o_{2,G}) \frac{i-k}{G} \right) \epsilon_i \right| > \beta \right) 
= O \big( \beta^{-\gamma} G^{\gamma/2} \big),
\\ 
& \Pr \left( \max_{k_j-G \le k \le k_j - \xi} \frac{1 }{|k_j - k|}\left|  \sum_{i=k-G+1}^{k_j - G} \left( (c_1 + o_{1,G}) + (c_2 + o_{2,G}) \frac{i-k}{G} \right) \epsilon_i \right| > \beta \right) 
\nn
\\
& = O \big( \beta^{-\gamma} (\xi^{-\gamma/2} + G^{-\gamma/2} ) \big),
\label{conv_eq1-fin2}
\\ 
& \Pr \left( \max_{k_j-G \le k \le k_j - \xi} \frac{1 }{|k_j - k|^2}\left|  \sum_{i=k-G+1}^{k_j - G} \left( (c_1 + o_{1,G}) + (c_2 + o_{2,G}) \frac{i-k}{G} \right) \epsilon_i \right| > \beta \right) 
\nn
\\
& = O \big( \beta^{-\gamma} \xi^{-\gamma/2}(\xi^{-\gamma} + G^{-\gamma} ) \big).
\label{conv_eq1-fin3}
\end{align}

As a next step, we bound the terms related to the summation from $i=k_j-G+1$ to $k_j$, namely $T_2$ in~\eqref{eq:decomp:zeta}. From Markov's inequality and Lemma~\ref{convrate_lem1}, we have
\begin{align*}
\Pr \left( \max_{\xi \le k_j - k \le G }  \left|\sum_{i=k_j-G+1}^{k_j} \frac{k_j-k}{G} \epsilon_i \right| > \beta \right)
& = \Pr \left( \left|\sum_{i=1}^{G} \epsilon_{k_j - G + i} \right| > \beta \right) \lesssim \beta ^{-\gamma} G^{\gamma/2},
\\
\Pr \left( \max_{\xi \le k_j - k \le G } \frac{1}{k_j - k} \left|\sum_{i=k_j-G+1}^{k_j} \frac{k_j-k}{G} \epsilon_i \right| > \beta  \right)
& = \Pr \left( \left|\sum_{i=1}^{G} \epsilon_{k_j - G + i} \right| > \beta G \right)
 \lesssim  \beta^{-\gamma} G^{-\gamma/2},
\\
\Pr \left( \max_{\xi \le k_j - k \le G } \frac{1}{(k_j - k)^2} \left|\sum_{i=k_j-G+1}^{k_j} \frac{k_j-k}{G} \epsilon_i \right| > \beta \right)
& = \Pr \left( \frac{1}{\xi} \left|\sum_{i=1}^{G}  \epsilon_{k_j - G + i} \right| > \beta G\right)
\lesssim \beta^{-\gamma} G^{-\gamma/2} \xi^{-\gamma} ,
\end{align*}
which yield
\begin{align} 
\label{conv_eq2-fin1}
\Pr \left( \max_{k_j-G \le k \le k_j - \xi} \left| \sum_{i=k_j-G+1}^{k_j} (c_1 + o_{1,G})\frac{k_j-k}{G} \epsilon_i \right| > \beta \right) 
& = O \big( \beta^{-\gamma} G^{\gamma/2} \big),
\\ \label{conv_eq2-fin2}
\Pr \left( \max_{k_j-G \le k \le k_j - \xi} \frac{1 }{|k_j - k|}\left|  \sum_{i=k_j-G+1}^{k_j} (c_1 + o_{1,G})\frac{k_j-k}{G} \epsilon_i  \right| > \beta \right) 
& = O \big( \beta^{-\gamma} G^{-\gamma/2}  \big),
\\ \label{conv_eq2-fin3}
\Pr \left( \max_{k_j-G \le k \le k_j - \xi} \frac{1 }{|k_j - k|^2}\left|  \sum_{i=k_j-G+1}^{k_j} (c_1 + o_{1,G})\frac{k_j-k}{G} \epsilon_i  \right| > \beta \right) 
& = O \big( \beta^{-\gamma} \xi^{-\gamma} G^{-\gamma/2} \big) .
\end{align}

We can also handle other terms related to $T_3, T_4$ and $T_5$ in~\eqref{eq:decomp:zeta} similarly, obtaining
{\small
\begin{align} 
\label{conv_eq3-fin1}
\Pr \left( \max_{k_j-G \le k \le k_j - \xi} \left| \sum_{i=k+1}^{k_j} \left( (c_1 + o_{1,G}) + (c_2 + o_{2,G}) \frac{i-k}{G} \right) \epsilon_i \right| > \beta \right) 
& = O \big( \beta^{-\gamma} G^{\gamma/2} \big) ,
\\ \label{conv_eq3-fin2}
\Pr \left( \max_{k_j-G \le k \le k_j - \xi} \frac{1 }{|k_j - k|}\left|  \sum_{i=k+1}^{k_j} \left( (c_1 + o_{1,G}) + (c_2 + o_{2,G}) \frac{i-k}{G} \right) \epsilon_i \right| > \beta \right) 
& = O \big( \beta^{-\gamma} (\xi^{-\gamma/2} + G^{-\gamma/2} ) \big) ,
\\ \label{conv_eq3-fin3}
\Pr \left( \max_{k_j-G \le k \le k_j - \xi} \frac{1 }{|k_j - k|^2}\left|  \sum_{i=k+1}^{k_j} \left( (c_1 + o_{1,G}) + (c_2 + o_{2,G}) \frac{i-k}{G} \right) \epsilon_i \right| > \beta \right) 
& = O \big( \beta^{-\gamma} \xi^{-\gamma/2}(\xi^{-\gamma} + G^{-\gamma} ) \big),
\\
\label{conv_eq4-fin1}
\Pr \left( \max_{k_j-G \le k \le k_j - \xi} \left| \sum_{i=k_j+1}^{k_j+G} (c_1 + o_{1,G})\frac{k_j-k}{G} \epsilon_i \right| > \beta \right) 
& = O \big( \beta^{-\gamma} G^{\gamma/2} \big) ,
\\ \label{conv_eq4-fin2}
\Pr \left( \max_{k_j-G \le k \le k_j - \xi} \frac{1 }{|k_j - k|}\left|  \sum_{i= k_j+1}^{k_j+G} (c_1 + o_{1,G})\frac{k_j-k}{G} \epsilon_i  \right| > \beta \right) 
& = O \big( \beta^{-\gamma} G^{-\gamma/2}  \big) ,
\\ \label{conv_eq4-fin3}
\Pr \left( \max_{k_j-G \le k \le k_j - \xi} \frac{1 }{|k_j - k|^2}\left|  \sum_{i= k_j+1}^{k_j+G} (c_1 + o_{1,G})\frac{k_j-k}{G} \epsilon_i  \right| > \beta \right) 
& = O \big( \beta^{-\gamma} \xi^{-\gamma} G^{-\gamma/2} \big) ,
\\
\label{conv_eq5-fin1}
\Pr \left( \max_{k_j-G \le k \le k_j - \xi} \left| \sum_{i=k+G+1}^{k_j +G} \left( (c_1 + o_{1,G}) + (c_2 + o_{2,G}) \frac{i-k}{G} \right) \epsilon_i \right| > \beta \right) 
& = O \big( \beta^{-\gamma} G^{\gamma/2} \big) ,
\\ \label{conv_eq5-fin2}
\Pr \left( \max_{k_j-G \le k \le k_j - \xi} \frac{1 }{|k_j - k|}\left|  \sum_{i=k +G +1}^{k_j +G} \left( (c_1 + o_{1,G}) + (c_2 + o_{2,G}) \frac{i-k}{G} \right) \epsilon_i \right| > \beta \right) 
& = O \big( \beta^{-\gamma} (\xi^{-\gamma/2} + G^{-\gamma/2} ) \big) ,
\\ \label{conv_eq5-fin3}
\Pr \left( \max_{k_j-G \le k \le k_j - \xi} \frac{1 }{|k_j - k|^2}\left|  \sum_{i=k +G +1}^{k_j +G} \left( (c_1 + o_{1,G}) + (c_2 + o_{2,G}) \frac{i-k}{G} \right) \epsilon_i \right| > \beta \right) 
& = O \big( \beta^{-\gamma} \xi^{-\gamma/2}(\xi^{-\gamma} + G^{-\gamma} ) \big) .
\end{align}}
Let us write $o_G=O(G^{-1})$ for simplicity. Then, we have
{\small
\begin{align*}
&\Pr \left( \max_{k_j - G \le k \le k_j - \xi} G \frac{\big| \zeta_{k_j}^{(0)} - \zeta_{k}^{(0)} \big|}{k_j-k} > \beta \right)
\\ & \le \Pr \left( \max_{k_j - G \le k \le k_j - \xi} \frac{1}{k_j-k} \left|\sum_{i=k-G+1}^{k_j- G} \left( (4+o_G) + (6+o _G) \frac{i-k}{G}\right) \epsilon_i \right| > \frac{\beta}{5} \right) & \left(\underset{\eqref{conv_eq1-fin2}}{=} O \big( \beta^{-\gamma} (\xi^{-\gamma/2} + G^{-\gamma/2}) \big) \right)
\\& + \Pr \left( \max_{k_j - G \le k \le k_j - \xi}\frac{1}{k_j-k} \left|\sum_{i=k_j-G+1}^{k_j} (6+o _G) \frac{k_j-k}{G}\epsilon_i \right| > \frac{\beta}{5} \right) & \left(\underset{\eqref{conv_eq2-fin2}}{=} O \big( \beta^{-\gamma} G^{-\gamma/2} \big) \right)
\\ & + \Pr \left( \max_{k_j - G \le k \le k_j - \xi}\frac{1}{k_j-k} \left|\sum_{i=k+1}^{k_j} \left( (8+o _G) + o _G \frac{i-k}{G} \right) \epsilon_i \right| > \frac{\beta}{5} \right) & \left(\underset{\eqref{conv_eq3-fin2}}{=} O \big( \beta^{-\gamma} (\xi^{-\gamma/2} + G^{-\gamma/2}) \big) \right)
\\& + \Pr \left( \max_{k_j - G \le k \le k_j - \xi}\frac{1}{k_j-k} \left|\sum_{i=k_j+1}^{k_j+G} ( 6+o_G)\frac{k_j-k}{G} \epsilon_i \right| > \frac{\beta}{5} \right) & \left(\underset{\eqref{conv_eq4-fin2}}{=} O \big( \beta^{-\gamma}  G^{-\gamma/2} \big) \right)
\\& + \Pr \left( \max_{k_j - G \le k \le k_j - \xi}\frac{1}{k_j-k} \left|\sum_{i=k+G+1}^{k_j+G} \left( (4+o _G) + (6+o _G) \frac{i-k}{G} \right) \epsilon_i \right| > \frac{\beta}{5} \right) & \left(\underset{\eqref{conv_eq5-fin2}}{=} O \big( \beta^{-\gamma} (\xi^{-\gamma/2} + G^{-\gamma/2}) \big) \right)
\\& = O(\beta^{-\gamma} (\xi^{-\gamma/2} + G^{-\gamma/2}) )
\end{align*}}
which gives~\ref{lem:hr:two} for $p = 0$, and the case of $p=1$ is handled analogously.
Similarly, 
combining \eqref{conv_eq1-fin1}, \eqref{conv_eq2-fin1}, \eqref{conv_eq3-fin1}, \eqref{conv_eq4-fin1} and \eqref{conv_eq5-fin1} give~\ref{lem:hr:one},  
\eqref{conv_eq1-fin3}, \eqref{conv_eq2-fin3}, \eqref{conv_eq3-fin3}, \eqref{conv_eq4-fin3} and \eqref{conv_eq5-fin3} give~\ref{lem:hr:three}.
Also we have 
\begin{align*}
& \Pr \left( G \big| \zeta_{k_j}^{(0)} \big| > \beta \right)
\le \Pr \left(   \left| \sum_{i=k_j+1}^{k_j+G} (4+o) \epsilon_i\right| > \frac{\beta}{4}\right) + \Pr \left(   \left|  \sum_{i=k_j+1}^{k_j+G} (6+o)\frac{i-k_j}{G} \epsilon_i\right| > \frac{\beta}{4} \right)
\\& + \Pr \left(   \left| \sum_{i=k_j-G+1}^{k_j}( 4+o) \epsilon_i\right| > \frac{\beta}{4}\right) + \Pr \left(   \left| \sum_{i=k_j-G+1}^{k_j}(6+o)\frac{i-k_j}{G}  \epsilon_i\right| > \frac{\beta}{4} \right)
\\& = O(\beta^{-\gamma} G^{\gamma/2}), \quad \text{and}
\\
& \Pr \left(   \max_{k_j - G \le k \le k_j - \xi}G \frac{\big| \zeta_{k_j}^{(0)} \big|}{k_j-k} > \beta \right)
= \Pr \left(   G \frac{\big| \zeta_{k_j}^{(0)} \big|}{\xi} > \beta \right) = O(\beta^{-\gamma} G^{\gamma/2} \xi^{-\gamma}),
\end{align*}
and therefore $\big|\zeta_{k_j}^{(p)} + \zeta_{k}^{(p)}\big| \le \big|\zeta_{k_j}^{(p)} - \zeta_{k}^{(p)}\big| + 2 \big|\zeta_{k_j}^{(p)} \big|$ combined with \ref{lem:hr:one}--\ref{lem:hr:two} gives \ref{lem:hr:four}--\ref{lem:hr:five}.
\end{proof}

\begin{proof}[Proof of Theorem~\ref{5.1}~\ref{thm:est:two}]
Following the proof of Theorem~3.2 in \cite{eichinger2018}, 
we show that for given $j \in \{1, \ldots, J_n\}$,
\begin{align}
& \Pr\l( \wh k_j - k_j > \xi, \, \mc M_n \r) = O\l( d_j^{-\gamma} \xi^{-3\gamma/2} \r), \label{eq:claim:one}
\\
& \Pr\l( \wh k_j - k_j < - \xi, \, \mc M_n \r) = O\l( d_j^{-\gamma} \xi^{-3\gamma/2} \r),
\quad \text{where} \label{eq:claim:two}
\\
& \cM_n = \left\{ \wh{J}_n = J_n, \, 
\max_{1 \le j \le \wh{J}_n} \vert\wh{k}_j - k_j \vert < G , \ 
\wh{\tau}^2 > 0 \right\}. \nn
\end{align}
Then, by Theorem~\ref{5.1}~\ref{thm:est:one} and the Condition~imposed on $\wh\tau^2$, we have $\Pr (\cM_n) \to 1$ and therefore 
the assertion of Theorem~\ref{5.1}~\ref{thm:est:two} follows from~\eqref{eq:claim:one}--\eqref{eq:claim:two}.
We focus on showing~\eqref{eq:claim:two}, and~\eqref{eq:claim:one} is shown using the analogous arguments.

Letting $\tilde{w}_j = \min (w_j, k_j + G - 1)$ and $\tilde{v}_j = \max(v_j, k_j - G + 1)$, 
we have
\begin{align*} 
\wh{k}_j = \underset{\tilde{v}_j \le k \le \tilde{w}_j}{\arg\max}~W_{k, n}^2(G) 
\quad \text{on} \quad \cM_n, 
\end{align*}
as we assumed that $\wh{\tau}_k = \wh{\tau}$ does not depend on $k$. 
Then we have that for $1 \le \xi \le G$, 
\begin{align*}
\wh{k}_j < k_j - \xi \text{ \ if and only if \ } 
\max_{\tilde{v}_j \le k < k_j - \xi} W_{k,n}^2(G) & \ge 
\max_{k_j - \xi \le k \le \tilde{w}_j} W_{k,n}^2(G), \quad \text{or equivalently,}
\\
\max_{\tilde{v}_j \le k < k_j - \xi} \l( W_{k, n}^2(G) - W_{k_j, n}^2(G) \r) & \ge 
\max_{k_j - \xi \le k \le \tilde{w}_j} \l( W_{k,n}^2 (G) - W_{k_j,n}^2(G) \r). 
\end{align*}
As $\max_{k_j - \xi \le k \le \tilde{w}_j} (W_{k,n}^2(G)- W_{k_j,n}^2(G) ) \ge 0$, 
we have 
\begin{align*} 
\Pr \left( \wh{k}_j < k_j - \xi, \, \cM_n \right) 
\le 
\Pr \left( \max_{\tilde{v}_j \le k < k_j - \xi} \l ( W_{k,n} ^{2}(G)  - W_{k_j,n}^2(G) \r) \ge 0,  \, \cM_n \right). 
\end{align*}
Write $\wh{\b\delta} (k) = \wh{\bbeta}^{+}(k) - \wh{\bbeta}^{-}(k)$. Then 
\begin{align*}
W_{k,n}^2(G) - W_{k_j,n }^2(G)
&= \frac{G}{8 \wh\tau^2} \left( \wh{\b\delta}(k)^{\top} \bmx 1 & 0 \\ 0 & \frac{1}{3} \emx \wh{\b\delta} (k) -  \wh{\b\delta}(k_j)^{\top} \bmx 1 & 0 \\ 0 & \frac{1}{3} \emx \wh{\b\delta} (k_j) \right) 
\\
& = \frac{G}{8\wh\tau^2} \l( \wh{\b\delta} (k) -  \wh{\b\delta} (k_j)\r)^{\top}  \bmx 1 & 0 \\ 0 & \frac{1}{3} \emx  \l( \wh{\b\delta} (k) +  \wh{\b\delta} (k_j)\r).
\end{align*}
By Lemma~\ref{lemA.4}, 
\begin{align*}
\wh{\b\delta}(k) - \wh{\b\delta} (k_j) =& \l( \mbf A(\kappa)+ \mbf O_G (\kappa) - \mbf I   \r) \b\Delta_{j} + \bzeta_{k} - \bzeta_{k_j},
\\
\wh{\b\delta}(k) + \wh{\b\delta} (k_j) =& \l( \mbf A(\kappa) + \mbf O_G (\kappa) + \mbf I \r) \b\Delta_{j} + \bzeta_{k} + \bzeta_{k_j}
\end{align*}
and further $\b\Delta_j = (0, \Delta^{(1)}_j)^\top$ with $\vert \Delta^{(1)}_j \vert = G d_j$. Then, we have
\begin{align}
& W_{k,n}^2(G) - W_{k_j,n }^2(G) 
= - \frac{G}{8\wh\tau^2}\l( D_1^{\top} D_2 + D_1^{\top} E_2 + E_1^{\top} D_2 + E_1^{\top} E_2  \r), \quad \text{where}
\nn \\
& D_1 = \bmx 1 & 0 \\ 0 & \frac{1}{\sqrt{3}} \emx\big(\mbf I - \mbf A(\kappa) -\mbf O_G (\kappa) \big) \bDelta_{j} = \bmx \kappa(1-\kappa)^2   \\ \kappa^2 (3-2\kappa)/{\sqrt{3}} \emx\Delta^{(1)}_j -  \bmx 1 & 0 \\ 0 & \frac{1}{\sqrt{3}} \emx \mbf O_G(\kappa) \bDelta_j, 
\nn \\
& D_2 = \bmx 1 & 0 \\ 0 & \frac{1}{\sqrt{3}} \emx\big(\mbf I + \mbf A(\kappa) + \mbf O_G (\kappa) \big) \bDelta_{j} = \bmx - \kappa(1-\kappa)^2 \\ (2\kappa^3 -3\kappa^2+2)/{\sqrt{3}} \emx\Delta^{(1)}_j +  \bmx 1 & 0 \\ 0 & \frac{1}{\sqrt{3}} \emx \mbf O_G(\kappa) \bDelta_j,
\nn \\
& E_1 = - \b\zeta_k + \b\zeta_{k_j}, \quad E_2 = \b\zeta_k + \b\zeta_{k_j}. 
\label{eq:def}
\end{align}
Then
\begin{align}
D_1^{\top} D_2 
\nn  = & \bDelta_j^{\top} (\bI - \bA(\kappa) )^{\top}  \bmx 1 & 0 \\ 0 & \frac{1}{3} \emx (\bI + \bA(\kappa) ) \bDelta_j 
\nn \\& - \bDelta_j^{\top} \bO_G (\kappa)^{\top}  \bmx 1 & 0 \\ 0 & \frac{1}{3} \emx \bO_G (\kappa) \bDelta_j - 2 \bDelta_j^{\top} \bA(\kappa)^{\top}  \bmx 1 & 0 \\ 0 & \frac{1}{3} \emx \bO_G(\kappa) \bDelta_j
\nn \\=& \bDelta_j^{\top} \bmx 1 - (1-\kappa)(1-3\kappa) & \mp (\kappa^2-\kappa)(1-\kappa) \\ \mp 6\kappa(1-\kappa) & 1-(1-\kappa)(-2\kappa^2+\kappa+1) \emx^{\top}
\nn \\
& \times \bmx 1 & 0 \\ 0 & \frac{1}{3} \emx\bmx 1 + (1-\kappa)(1-3\kappa) & \pm (\kappa^2-\kappa)(1-\kappa) \\ \pm 6\kappa(1-\kappa) & 1+(1-\kappa)(-2\kappa^2+\kappa+1) \emx \bDelta_j
\nn \\
 & - \left( \frac{\kappa (1-\kappa) }{G \mp 1} \right)^{2} \left[ (2-\kappa)^2 + \frac{1}{3} \left( \frac{2\kappa - 1 - \mp 3G}{G \pm 1} \right)^{2} \right] \vert \Delta^{(1)}_j \vert^2
\nn \\
& + \frac{2 \kappa (1-\kappa)^2}{G \mp 1} \bmx \mp (\kappa^2 - \kappa) 
\nn \\ -2\kappa^2 + \kappa +1 \emx^{\top} \bmx 1 & 0 \\ 0 & \frac{1}{3} \emx \bmx 2-\kappa \\ \frac{2\kappa - 1 \mp 3G}{G \pm 1} \emx \vert \Delta^{(1)}_j \vert^2
\nn \\
=&\frac{ \kappa^2}{3} (-7\kappa^4 + 24 \kappa^3 - 27 \kappa^2 + 8 \kappa + 3 + O(G^{-2})) \vert \Delta^{(1)}_j \vert^2  
\nn \\
& \pm \frac{2\kappa (1-\kappa)^2}{G \mp 1} \left[ \kappa (1-\kappa)(2-\kappa) - (-2\kappa^2 + \kappa +1) \left( 1 \pm \frac{2}{3} \frac{\kappa+1}{G \pm 1} \right) \right]\vert \Delta^{(1)}_j \vert^2.
\end{align}
From that
\begin{align*}
-7\kappa^4 + 24 \kappa^3 - 27 \kappa^2 + 8 \kappa  + 3
= \kappa (1 - \kappa) (7\kappa^2 - 17\kappa + 10) + 3 - 2\kappa \ge 3 - 2\kappa \ge 1, \quad \text{and}
\\
2 (1-\kappa)^2 \left| \kappa (1-\kappa)(2-\kappa) - (-2\kappa^2 + \kappa +1) \right|  = 2(1-\kappa)^3 (1+\kappa^2) \le 2 
\end{align*}
for $\kappa \in [0, 1]$, we get 
\begin{align*}
|D_1^{\top} D_2| 
& \ge \left[ \left( \frac{1}{3} + O(G^{-1}) \right) \kappa^2 - \left( \frac{2}{G} + O(G^{-2}) \right) \kappa \right]  \vert \Delta^{(1)}_j \vert^2.
\end{align*}
For any $c_0 \ge 7$, if $\xi \ge c_0$ such that $G \kappa \ge 7$, 
we have for large enough $G$,
\begin{align*} 
\left( \frac{1}{3} + O(G^{-1}) \right) \kappa^2 - \left( \frac{2}{G} + O(G^{-2}) \right) \kappa \ge \frac{\kappa^2}{22} \iff G \kappa \ge \frac{2 + O(G^{-1})}{\frac{1}{3} - \frac{1}{22} + O(G^{-1})} \approx 6.947+O(G^{-1}), \end{align*}
and therefore
\begin{align}
|D_1^{\top} D_2| \ge \frac{\kappa^2}{22}  \vert \Delta^{(1)}_j \vert^2 > 0.
\label{eq:dd}
\end{align}
 Then,
\begin{align*}
&\Pr \left( \max_{k: \, \xi \le k_j - k \le G} \l( W_{k,n}^2(G) - W_{k_j,n }^2(G) \r) \ge 0, \, \cM_n \right)
\\
=& \Pr \left( \max_{k: \, \xi \le k_j - k \le G}\left(- D_1^{\top} D_2 \left(1 + \frac{D_2^{\top} E_1}{D_1^{\top} D_2} + \frac{D_1^{\top} E_2}{D_1^{\top} D_2} + \frac{E_1^{\top} E_2}{D_1^{\top} D_2}  \right) \right) \ge 0, \, \cM_n \right)
\\
\le & \Pr \left( \max_{k: \, \xi \le k_j - k \le G} \left\vert   \frac{D_2^{\top} E_1}{D_1^{\top} D_2} + \frac{D_1^{\top} E_2}{D_1^{\top} D_2} + \frac{E_1^{\top} E_2}{D_1^{\top} D_2}  \right\vert \ge 1, \, M_n \right)
\\
\le &  \Pr \left( \max_{k: \, \xi \le k_j - k \le G} \left\vert   \frac{D_2^{\top} E_1}{D_1^{\top} D_2} \right\vert \ge \frac{1}{3} \right) +  
\Pr \left( \max_{k: \, \xi \le k_j - k \le G} \left\vert   \frac{D_1^{\top} E_2}{D_1^{\top} D_2} \right\vert \ge \frac{1}{3} \right) 
\\
& + \Pr \left( \max_{k: \, \xi \le k_j - k \le G} \left\vert  \frac{E_1^{\top} E_2}{D_1^{\top} D_2} \right\vert \ge \frac{1}{3} \right) 
\\
=:& \, P_1 + P_2 + P_3 .
\end{align*}
From~\eqref{eq:def} and~\eqref{eq:dd}, we have for large $G$
\begin{align*} 
\frac{1}{22} \left\vert \frac{D_2^{\top} E_1}{D_1^{\top} D_2} \right\vert 
\le &
\frac{\kappa (1-\kappa)^{2}\big\vert \zeta_{k_j}^{(0)} - \zeta_{k}^{(0)} \big\vert +
\frac{ 2\kappa^3 - 3\kappa^2 + 2}{\sqrt{3}} \big\vert \zeta_{k_j}^{(1)} - \zeta_{k}^{(1)}\big\vert }
{\kappa^2 \vert\Delta^{(1)}_j\vert} 
\\& +  \frac{\kappa(1-\kappa)}{G - 1} \frac{ (2-\kappa) \vert \zeta_{k_j}^{(0)} - \zeta_{k}^{(0)} \vert + \frac{2\kappa - 1 - 3G}{G + 1} \vert \zeta_{k_j}^{(1)} - \zeta_{k}^{(1)} \vert  }{\kappa^2 \vert \Delta^{(1)}_j \vert }  \\
\le &
\frac{(G + O(1))\big\vert \zeta_{k_j}^{(0)} - \zeta_{k}^{(0)}\big\vert}{\vert\Delta^{(1)}_j\vert \; \vert k-k_j\vert } + 
\frac{(2 G^2 + O(G)) \big\vert \zeta_{k_j}^{(1)} - \zeta_{k}^{(1)}\big\vert}{\sqrt{3} \vert\Delta^{(1)}_j\vert \; \vert k-k_j\vert^{2}} \\
\le &
\frac{2G\big\vert \zeta_{k_j}^{(0)} - \zeta_{k}^{(0)}\big\vert}{\vert\Delta^{(1)}_j\vert \; \vert k-k_j\vert } + 
\frac{3 G^2  \big\vert \zeta_{k_j}^{(1)} - \zeta_{k}^{(1)}\big\vert}{\sqrt{3} \vert\Delta^{(1)}_j\vert \; \vert k-k_j\vert^{2}}
.
\end{align*}
Then by Lemma~\ref{lem:hr}~\ref{lem:hr:two} and~\ref{lem:hr:three},
\begin{align*}
P_1 \le & \, \Pr \left( \max_{k: \, \xi \le k_j - k \le G} 
\frac{G \big\vert \zeta_{k_j}^{(0)} - \zeta_{k}^{(0)}\big\vert}{ \vert k-k_j\vert } \ge \frac{\vert\Delta^{(1)}_j\vert }{264} \right)
+ \Pr \left( \max_{k: \, \xi \le k_j - k \le G} 
\frac{G \big\vert \zeta_{k_j}^{(1)} - \zeta_{k}^{(1)}\big\vert}{\vert k-k_j\vert^{2}} \ge \frac{\vert\Delta^{(1)}_j\vert}{132\sqrt{3} G} \right)
\\
=& \, O\l( \vert\Delta^{(1)}_j\vert^{-\gamma} \l( \xi^{-\gamma/2} + G^{-\gamma/2} \r)
+ \vert\Delta^{(1)}_j\vert^{-\gamma} G^{\gamma} \l( \xi^{-3\gamma/2} + G^{-\gamma/2} \xi^{-\gamma} + G^{-\gamma} \xi^{-\gamma/2}\r) \r)
\\
=& \,  O\l( \vert\Delta^{(1)}_j\vert^{-\gamma} \l( G^{\gamma} \xi^{-3\gamma/2} \r)\r).
\end{align*}
Next, we have from~\eqref{eq:def} and~\eqref{eq:dd},
\begin{align*} 
\frac{1}{22} \left\vert  \frac{D_1^{\top} E_2}{D_1^{\top} D_2} \right\vert 
\le &  
\frac{\kappa (1-\kappa)^{2}\big\vert \zeta_{k_j}^{(0)} + \zeta_{k}^{(0)}\big\vert + \frac{\kappa^2 (3-2\kappa)}{\sqrt{3}} \big\vert \zeta_{k_j}^{(1)} + \zeta_{k}^{(1)}\big\vert }{ \kappa^2 \vert\Delta^{(1)}_j \vert } 
+ \frac{ (1-\kappa)}{G - 1} \frac{ 2 \big\vert \zeta_{k_j}^{(0)} - \zeta_{k}^{(0)} \big\vert +  3 \big\vert \zeta_{k_j}^{(1)} - \zeta_{k}^{(1)} \big\vert }
{ \kappa \vert\Delta^{(1)}_j \vert } \\
\le & \frac{2G \big\vert \zeta_{k_j}^{(0)} + \zeta_{k}^{(0)}\big\vert}{\vert\Delta^{(1)}_j\vert \; \vert k-k_j\vert } + \frac{2}{\vert\Delta^{(1)}_j\vert} \big\vert \zeta_{k_j}^{(1)} + \zeta_{k}^{(1)}\big\vert,
\end{align*}
such that by Lemma~\ref{lem:hr}~\ref{lem:hr:four} and~\ref{lem:hr:five},
\begin{align*}
P_2 \le & 
\Pr \left( \max_{k: \, \xi \le k_j - k \le G} 
\frac{G \big\vert \zeta_{k_j}^{(0)} + \zeta_{k}^{(0)}\big\vert}{\vert k-k_j\vert } \ge \frac{\vert\Delta^{(1)}_j\vert}{264} \right)
+ \Pr \left( \max_{k: \, \xi \le k_j - k \le G} 
G \big\vert \zeta_{k_j}^{(1)} + \zeta_{k}^{(1)}\big\vert \ge 
\frac{\vert\Delta^{(1)}_j\vert G}{264} \r)
\\
=& \, O\l( \vert\Delta^{(1)}_j\vert^{-\gamma} \l(\xi^{-\gamma/2} + G^{-\gamma/2} + G^{\gamma/2} \xi^{-\gamma}) \r) 
+ \vert\Delta^{(1)}_j\vert^{-\gamma} G^{-\gamma/2} \r)
= O\l(\vert\Delta^{(1)}_j\vert^{-\gamma} G^{\gamma/2} \xi^{-\gamma} \r).
\end{align*}
Similarly, we have from~\eqref{eq:def} and~\eqref{eq:dd},
\begin{align*}
\frac{1}{22} \left\vert   \frac{E_1^{\top} E_2}{D_1^{\top} D_2} \right\vert 
& \le 
\frac{ \big\vert \zeta_{k_j}^{(0)} - \zeta_{k}^{(0)}\big\vert \big\vert \zeta_{k_j}^{(0)} + \zeta_{k}^{(0)}\big\vert +\big\vert \zeta_{k_j}^{(1)} - \zeta_{k}^{(1)}\big\vert \big\vert \zeta_{k_j}^{(1)} + \zeta_{k}^{(1)}\big\vert }
{\kappa^2 \vert\Delta^{(1)}_j\vert^2}
\\
&= \frac{G \big\vert \zeta_{k_j}^{(0)} + \zeta_{k}^{(0)}\big\vert}{\vert\Delta^{(1)}_j\vert \; \vert k-k_j\vert} \cdot 
\frac{G \big\vert \zeta_{k_j}^{(0)} - \zeta_{k}^{(0)}\big\vert}{\vert\Delta^{(1)}_j\vert \; \vert k-k_j\vert}
 + \frac{G \big\vert \zeta_{k_j}^{(1)} + \zeta_{k}^{(1)}\big\vert}{\vert\Delta^{(1)}_j\vert \; \vert k-k_j\vert} \cdot 
 \frac{G \big\vert \zeta_{k_j}^{(1)} - \zeta_{k}^{(1)}\big\vert}{\vert\Delta^{(1)}_j\vert \; \vert k-k_j\vert},
\end{align*}
we have from Lemma~\ref{lem:hr}~\ref{lem:hr:two} and~\ref{lem:hr:five}, 
\begin{align*}
P_3 \le & \,
\Pr \left( \max_{k: \, \xi \le k_j - k \le G} 
\frac{G \big\vert \zeta_{k_j}^{(0)} - \zeta_{k}^{(0)}\big\vert}
{\vert k-k_j\vert} \ge \frac{\vert\Delta^{(1)}_j\vert}{2 \sqrt{33}} \right) 
+ \Pr \left( \max_{k: \, \xi \le k_j - k \le G} 
\frac{G \big\vert \zeta_{k_j}^{(0)} + \zeta_{k}^{(0)}\big\vert}
{\vert k-k_j\vert}  \ge \frac{\vert\Delta^{(1)}_j\vert}{2 \sqrt{33}} \right)
\\
& + \Pr \left( \max_{k: \, \xi \le k_j - k \le G} 
\frac{G \big\vert \zeta_{k_j}^{(1)} - \zeta_{k}^{(1)}\big\vert}
{\vert k-k_j\vert}  \ge \frac{\vert\Delta^{(1)}_j\vert}{2 \sqrt{33}} \right) + 
\Pr \left( \max_{k: \, \xi \le k_j - k \le G} 
\frac{G \big\vert \zeta_{k_j}^{(1)} + \zeta_{k}^{(1)}\big\vert}
{\vert k-k_j\vert}  \ge \frac{\vert\Delta^{(1)}_j\vert}{2 \sqrt{33}} \right)
\\
=& \, O\l( \vert\Delta^{(1)}_j\vert^{-\gamma} (\xi^{-\gamma/2} + G^{-\gamma/2}) +
\vert\Delta^{(1)}_j\vert^{-\gamma} \l(\xi^{-\gamma/2} + G^{-\gamma/2} + G^{\gamma/2} \xi^{-\gamma} \r) \r)
\\
=& \, O\l( \vert\Delta^{(1)}_j\vert^{-\gamma} G^{\gamma/2} \xi^{-\gamma} \r).
\end{align*}

Finally, collecting the bounds on $P_1$--$P_3$, we have
\begin{align*}
&\Pr \left( \max_{k: \, \xi \le k_j - k \le G} \l( W_{k,n}^2(G) - W_{k_j,n }^2(G) \r) \ge 0, \, \cM_n \right)
\\
= & \, 
O\l( \big\vert \Delta^{(1)}_j \big\vert^{-\gamma} \l( G^\gamma \xi^{-3\gamma/2} + G^{\gamma/2} \xi^{-\gamma} + G^{-\gamma/2} \xi^{-\gamma/2} \r) \r)
= O\l( d_j^{-\gamma} \xi^{-3\gamma/2} \r),
\end{align*}
which proves~\eqref{eq:claim:two}.
\end{proof}

\subsection{Proof of Theorem~\ref{thm_sigma}}

\begin{proof}[Proof of Theorem~\ref{thm_sigma}~\ref{thm_sigma:tau}] 
For $k$ satisfying $\min_{1 \le j \le J_n} \vert k - k_j \vert \ge G$, we have
\begin{align*}
G \wh{\sigma}_{k,+}^{2} &= 
\sum_{i = k+1}^{k+G} \l( X_i - \mbf x_{i,k}^{\top} \wh{\bbeta}^{+}(k) \r)^{2} 
= \sum_{i = k+1}^{k+G} \left(\epsilon_i - \mbf x_{i,k}^{\top} \bC_{G, +}^{-1} \sum_{l = k + 1}^{k+G} \mbf x_{l, k} \epsilon_l \right)^{2}
\\& = \sum_{i=k+1}^{k+G} \epsilon_i^{2} - \left( \sum_{i=k+1}^{k+G} \mbf x_{i,k} \epsilon_i \right)^{\top} \bC_{G, +}^{-1} \left( \sum_{i=k+1}^{k+G} \mbf x_{i,k} \epsilon_i \right)
\end{align*}
such that
\begin{align*}
\max_{k: \, \min_j |k-k_j| \ge G} G \l\vert \wh{\sigma}_{k, +}^{2} - \sigma^{2} \r\vert & \le \max_k \left\vert  \sum_{i=k+1}^{k+G}( \epsilon_i^{2} - \sigma^2 )\right\vert +  \max_k \left\vert  \left( \sum_{i=k+1}^{k+G} \mbf x_{i,k} \epsilon_i \right)^{\top} \bC_{G, +}^{-1} \left( \sum_{i=k+1}^{k+G} \mbf x_{i,k} \epsilon_i \right) \right\vert 
\\& \le 
\max_k \left\vert  \sum_{i=k+1}^{k+G} (\epsilon_i^{2} - \sigma^2) \right\vert +  \max_k \left\Vert   \sum_{i=k+1}^{k+G} \mbf x_{i,k} \epsilon_i \right\Vert^{2} \; 
\l\Vert\bC_{G, +}^{-1} \r\Vert_{2} 
\\& =  \max_k \left\vert  \sum_{i=k+1}^{k+G} (\epsilon_i^{2} - \sigma^2 )\right\vert + O_P \left( \log(n/G) \right)  
\end{align*}
where the last equality follows from Lemmas~\ref{lem:invC} and~\ref{lemA.3}.
Defining $\upsilon_i = \epsilon_i^2 - \sigma^2$, we get $\E(\upsilon_i) = 0$ and $\E (\upsilon_i^{2}) < \infty$. 
Then we have, for any $\delta > 0$, 
\begin{align*}
\Pr\left( \max_{G \le k \le n-G} \left\vert \sum_{i=k+1}^{k+G} \upsilon_i \right\vert > \delta \right) 
& \le \Pr \left( \max_{G \le k \le n-G} \left\vert \sum_{i=1}^{k+G}\upsilon_i  \right\vert > \frac{ \delta}{2}  \right) +  \Pr \left( \max_{G \le k \le n-G} \left\vert  \sum_{i=1}^{k}\upsilon_i  \right\vert > \frac{ \delta}{2}  \right) 
 \\ & { \le \frac{Cn}{(\delta/2)^{2}}, }
\end{align*}
{where the last inequality follows from Lemma \ref{lem:kirch}.} 
Thus,
$\max_{G \le k \le n-G} \vert  \sum_{i=k+1}^{k+G} (\epsilon_i^{2} - \sigma^2) \vert = O_P(\sqrt n)$
and with~\eqref{eq:cond:G}, we conclude
\begin{align*}
\max_{k: \, \min_j |k-k_j| \ge G} \l\vert \wh{\sigma}_{k,+}^{2} - \sigma^{2} \r\vert = O_P \left( \frac{\sqrt{n}}{G} + \frac{\log(n/G)}{G} \right)  = o_P \left( \frac{1}{\log(n/G)} \right).
\end{align*}
We similarly derive the same result for $\wh{\sigma}_{k,-}^{2}$, such that
\begin{align} \label{eq:sigma_k}
\max_{k: \, \min_j |k-k_j| \ge G} \l\vert\wh{\sigma}_{k}^{2} - \sigma^{2} \r\vert = O_P \left( \frac{\sqrt{n}}{G} \right) =  o_P \left( \frac{1}{\log(n/G)} \right).
\end{align}

By definition of $\wh{\tau}_k^{2}$ and $\tau^2$ we get 
\begin{align*}
\max_{k: \, \min_j |k-k_j| \ge G} \big| \wh{\tau}_k^{2} - \tau^2 \big| & \le \max_{k: \, \min_j |k-k_j| \ge G} \big| \wh{\Gamma}_k(0) - \sigma^2 \big| 
\\& + \left|\sum_{h=1}^{S_n} \left( \cK \left( \frac{h}{S_n} \right) \wh{\Gamma}_{k,+}(h) - \Gamma(h)\right)\right|
\\& + \left|\sum_{h=1}^{S_n} \left( \cK \left( \frac{h}{S_n} \right) \wh{\Gamma}_{k,-}(h) - \Gamma(h)\right)\right|.
\end{align*}
As we consider $k$ sufficiently far from change points $k_j$, we get 
\[ X_i - \bx_{i,k}^{\top} \wh{\bbeta}^{+}(k) = \epsilon_{i} - \bx_{i,k}^{\top} \bC_{G,+}^{-1} \sum_{j=k+1}^{k+G} \bx_{j,k} \epsilon_{j} \]
and hence
\begin{align*}
& \wh{\Gamma}_{k,+}(h) - \Gamma(h) 
\\
=& \frac{1}{G-2} \sum_{i=k+1}^{k+G-h} \left[ \big( \epsilon_{i} \epsilon_{i+h} - \Gamma(h) \big) - \epsilon_{i} \bx_{i+h,k}^{\top} \bC_{G,+}^{-1} \sum_{j=k+1}^{k+G} \bx_{j,k} \epsilon_{j} 
-  \epsilon_{i+h} \bx_{i,k}^{\top} \bC_{G,+}^{-1} \sum_{j=k+1}^{k+G} \bx_{j,k} \epsilon_{j} \right.
\\ 
& + \left.  \epsilon_{i} \bx_{i+h,k}^{\top} \bC_{G,+}^{-1} \sum_{j=k+1}^{k+G} \bx_{j,k} \epsilon_{j}  \epsilon_{i+h} \bx_{i,k}^{\top} \bC_{G,+}^{-1} \sum_{j^{\prime}=k+1}^{k+G} \bx_{j^{\prime},k} \epsilon_{j^{\prime}} \right] - \frac{h-2}{G-2} \Gamma(h).
\end{align*}
By following similar arguments as those used in the proof of Theorem~2.3 of \cite{eichinger2018}, we get 
\begin{align*}
\Pr \left( \max_{0 \le k \le n-G} \left| \frac{1}{G-2} \sum_{h=1}^{S_n} \cK \left( \frac{h}{S_n} \right) \sum_{i=k+1}^{k+G-h} \big( \epsilon_{i} \epsilon_{i+h} - \Gamma(h) \big) \right| > c \right) 
& \lesssim \frac{n S_n^2}{G^2 c^2} 
\end{align*}
which yields 
\[ \max_{0 \le k \le n-G} \left| \frac{1}{G-2} \sum_{h=1}^{S_n} \cK \left( \frac{h}{S_n} \right) \sum_{i=k+1}^{k+G-h} \big( \epsilon_{i} \epsilon_{i+h} - \Gamma(h) \big) \right| = O_P \left( \frac{\sqrt{n} S_n}{G} \right). \]
Now denote as $\bC_{G,+}^{-1} = \frac{1}{G} \begin{bmatrix} c_{11} & c_{12} \\ c_{21} & c_{22} \end{bmatrix}$ for simplicity. Then 
{\small
\begin{align*}
& \Pr \left( \max_{0 \le k \le n-G} \left| \frac{1}{G-2} \sum_{h=1}^{S_n} \cK \left( \frac{h}{S_n} \right) \sum_{i=k+1}^{k+G-h}  \epsilon_{i} \bx_{i+h,k}^{\top} \bC_{G,+}^{-1} \sum_{j=k+1}^{k+G} \bx_{j,k} \epsilon_{j}  \right| > c \right) 
\\& \le \sum_{k=0}^{n-G}  \Pr \left( \left| \frac{1}{G-2} \sum_{h=1}^{S_n} \cK \left( \frac{h}{S_n} \right) \sum_{i=k+1}^{k+G-h}  \epsilon_{i} \bx_{i+h,k}^{\top} \bC_{G,+}^{-1} \sum_{j=k+1}^{k+G} \bx_{j,k} \epsilon_{j}  \right| > c \right) 
\\& \le \frac{n}{(G-2)^2 c^2} \sum_{h=1}^{S_n} \sum_{h^{\prime}=1}^{S_n} \cK \left( \frac{h}{S_n} \right) \cK \left( \frac{h^{\prime}}{S_n} \right)
\sum_{i=k+1}^{k+G-h} \sum_{i^{\prime}=k+1}^{k+G-h^{\prime}} \E \left[\epsilon_i \bx_{i+h,k}^{\top} \bC_{G,+}^{-1} \sum_{j=k+1}^{k+G} \bx_{j,k} \epsilon_{j} \epsilon_{i^{\prime}} \bx_{i^{\prime}+h,k}^{\top} \bC_{G,+}^{-1} \sum_{j^{\prime}=k+1}^{k+G} \bx_{j^{\prime},k} \epsilon_{j^{\prime}} \right] 
\\& = \frac{n}{(G-2)^2G^2 c^2} \sum_{h=1}^{S_n} \sum_{h^{\prime}=1}^{S_n} \cK \left( \frac{h}{S_n} \right) \cK \left( \frac{h^{\prime}}{S_n} \right)
\sum_{i=k+1}^{k+G-h} \sum_{i^{\prime}=k+1}^{k+G-h^{\prime}}
\sum_{j=k+1}^{k+G} \sum_{j^{\prime}=k+1}^{k+G}\E \big( \epsilon_{i} \epsilon_{j} \epsilon_{i^{\prime}} \epsilon_{j^{\prime}} \big)
\\& \times \left( c_{11} + c_{12} \frac{j-k}{G} + c_{21} \frac{i+h-k}{G} + c_{22} \frac{(j-k)(i+h-k)}{G} \right) 
\\& \times 
 \left( c_{11} + c_{12} \frac{j^{\prime}-k}{G} + c_{21} \frac{i^{\prime}+h-k}{G} + c_{22} \frac{(j^{\prime}-k)(i^{\prime}+h-k)}{G} \right) 
\\& \lesssim \frac{n}{G^4 c^2} \sum_{h=1}^{S_n} \sum_{h^{\prime}=1}^{S_n} \cK \left( \frac{h}{S_n} \right) \cK \left( \frac{h^{\prime}}{S_n} \right) \sum_{i=k+1}^{k+G-h} \sum_{i^{\prime}=k+1}^{k+G-h^{\prime}}
\sum_{j=k+1}^{k+G} \sum_{j^{\prime}=k+1}^{k+G}\E \big( \epsilon_{i} \epsilon_{j} \epsilon_{i^{\prime}} \epsilon_{j^{\prime}} \big)
\\& =  \frac{n}{G^4 c^2} \sum_{h=1}^{S_n} \sum_{h^{\prime}=1}^{S_n} \cK \left( \frac{h}{S_n} \right) \cK \left( \frac{h^{\prime}}{S_n} \right) \sum_{i=k+1}^{k+G-h} \sum_{i^{\prime}=k+1}^{k+G-h^{\prime}}
\sum_{j=k+1}^{k+G} \sum_{j^{\prime}=k+1}^{k+G}\big[ \omega(j-i, i^{\prime}-i, j^{\prime}-i) + \Gamma(j-i) \Gamma(j^{\prime} - i^{\prime}) \big] 
\\& \lesssim \frac{nS_n^2}{G^2 c} , 
\end{align*}}
as $\sum_{h>0} \Gamma(h) = (\tau^2-\sigma^2)/2$ is bounded. Therefore we have 
\[ \max_{k: \, \min_j |k-k_j| \ge G} \big|  \wh{\Gamma}_{k,+}(h) - \Gamma(h) \big| = O_P \left( \frac{\sqrt{n}S_n}{G} \right).\]
We have the same result with $\wh{\Gamma}_{k,-}(h)$ in similar way, and hence 
\begin{align*}
\max_{k: \, \min_j |k-k_j| \ge G} \big| \wh{\tau}_k^{2} - \tau^2 \big| & \le \max_{k: \, \min_j |k-k_j| \ge G} \big| \wh{\Gamma}_k(0) - \sigma^2 \big| 
\\& + \left|\sum_{h=1}^{S_n} \left( \cK \left( \frac{h}{S_n} \right) \wh{\Gamma}_{k,+}(h) - \Gamma(h)\right)\right|
\\& + \left|\sum_{h=1}^{S_n} \left( \cK \left( \frac{h}{S_n} \right) \wh{\Gamma}_{k,-}(h) - \Gamma(h)\right)\right|
\\& \le \max_{k: \, \min_j |k-k_j| \ge G} \big| \wh{\Gamma}_k(0) - \sigma^2 \big| 
\\& + \left|\sum_{h=1}^{S_n} \left( \cK \left( \frac{h}{S_n} \right) \big(\wh{\Gamma}_{k,+}(h) - \Gamma(h)\big)\right)\right| + \sum_{h=1}^{S_n} \left| \cK \left( \frac{h}{S_n} \right) - 1 \right| |\Gamma(h)| 
\\& + \left|\sum_{h=1}^{S_n} \left( \cK \left( \frac{h}{S_n} \right) \big(\wh{\Gamma}_{k,-}(h) - \Gamma(h)\big)\right)\right|+ \sum_{h=1}^{S_n} \left| \cK \left( \frac{h}{S_n} \right) - 1 \right| |\Gamma(h)| 
\\& =  O_P \left( \frac{\sqrt{n} S_n}{G} + \sum_{h \in \bbZ} \left| \cK \left( \frac{h}{S_n} \right) - 1 \right| |\Gamma(h)| \right) . 
\end{align*}
\end{proof}

\begin{proof}[Proof of Theorem~\ref{thm_sigma}~\ref{thm_sigma:sigma}~\ref{thm_sigma:one}]
{ \eqref{eq:sigma_k} and}
\begin{align*}
& \max_{k: \, \min_j |k-k_j| \ge G} \wh{\sigma}_{k}^{2} \le \sigma^{2} + \max_{k: \, \min_j |k-k_j| \ge G}  \l\vert \wh{\sigma}_{k}^{2} - \sigma^{2} \r\vert = O_P (1), \quad \text{and}
\\
& \max_{k: \, \min_j |k-k_j| \ge G}  \frac{1}{\wh{\sigma}_{k}^{2}} \le \frac{1}{\sigma^{2} - \displaystyle \max_{k: \, \min_j |k-k_j| \ge G} \l\vert \wh{\sigma}_{k}^{2} - \sigma^{2} \r\vert} = O_P (1)
\end{align*}
{give the assertion.}
\end{proof}

\begin{proof}[Proof of Theorem~\ref{thm_sigma}~\ref{thm_sigma:sigma}~\ref{thm_sigma:two}]
From the proof of Lemma~\ref{lemA.4}, for $k \in \{k_j - G + 1, \ldots, k_j\}$, we have
\begin{align*} 
\wh{\bbeta}^{-}(k)  =& \bmx \beta_{0,j}(k_j)+ \frac{k-k_j}{G}  \beta_{1,j}\\  \beta_{1,j}\emx + \bzeta_{k,-}, 
\\
\wh{\bbeta}^{+}(k) =& \bmx \beta_{0,j}(k_j)+ \frac{k-k_j}{G}  \beta_{1,j}\\  \beta_{1,j}\emx + \bzeta_{k,+} 
\\
& + \l\{ \mbf I - \frac{\vert k-k_j\vert}{G} \l( \bC_{+} + \bD_{G,+}\r)^{-1} \l( \bC_{\kappa,+} + \bD_{\kappa,G,+} \r) \r\} \bmx \Delta^{(0)}_j+ \frac{k-k_j}{G}  \Delta^{(1)}_j  \\  \Delta^{(1)}_j \emx.
\end{align*}
Hence from~\eqref{eq:ao:kappa},
\begin{align*}
X_i - \mbf x_{i,k}^{\top} \wh{\bbeta}^{-} (k) 
=&  \epsilon_i - \mbf x_{i,k}^{\top} \bzeta_{k,-} , \quad \text{and}  
\\
X_i - \mbf x_{i,k}^{\top} \wh{\bbeta}^{+} (k) =&  \left(\beta_{0,j} (k_j) + \frac{i-k_j}{G} \beta_{1,j}+ \epsilon_i \right) - \mbf x_{i,k}^{\top}  \bmx \beta_{0,j}(k_j)+ \frac{k-k_j}{G}  \beta_{1,j}\\  \beta_{1,j}\emx - \mbf x_{i,k}^{\top} \bzeta_{k,+} 
\\
& - \mbf x_{i,k}^{\top} \l(\mbf A(\kappa) + \mbf O_G(\kappa) \r) \b\Delta_j 
= \epsilon_i - \mbf x_{i,k}^{\top} \bzeta_{k,+} + O(1) 
\end{align*}
with $\kappa = \vert k - k_j \vert/G$, 
provided that $\max_{1 \le j \le J_n} \|\bDelta_j\| = O(1) $.
Similarly we have 
\begin{align*} X_i - \mbf x_{i,k}^{\top} \wh{\bbeta}^{\pm} (k)  =  \epsilon_i - \mbf x_{i,k}^{\top} \bzeta_{k,\pm} + O(1) 
\end{align*}
in the case where $k \in \{k_j + 1, \ldots, k_j+G\}$. 
Therefore, by Cauchy-Schwarz inequality,
\begin{align*}
\max_{G \le k \le n-G} \wh{\sigma}_k^2 &\le \max_{G \le k \le n-G} \frac{1}{G-2} \left( \sum_{i=k-G+1}^{k} ( \epsilon_i - \mbf x_{i,k}^{\top} \bzeta_{k,-})^{2} + \sum_{i=k+1}^{k+G}  ( \epsilon_i - \mbf x_{i,k}^{\top} \bzeta_{k,+})^{2}  \right) + O(1) 
\end{align*}
and the claim follows from~\ref{thm_sigma:one}. 
\end{proof}


\end{document}